%% file: main.tex
\documentclass[3p, final]{elsarticle}

\input{templates/my-packages}
\input{templates/my-macros}

\input{templates/math-commands.tex}

\input{templates/ml-opt-packages.tex}

\usepackage{todonotes}

\begin{document}
	
	\begin{frontmatter}

\input{sections/title}
		\input{sections/abstract}
		
	\end{frontmatter}
	
	\input{sections/introduction}
	\input{sections/literaturereview}
	\input{sections/methodology}

	\input{sections/results}
	\input{sections/conclusions}

	\input{sections/acknowledgments}

\input{sections/references}

	\include{sections/appendix}

\end{document}

%% file: templates/my-packages.tex
\usepackage{color, colortbl} 
\usepackage{xcolor}
\usepackage{environ}

\usepackage{enumitem}
\usepackage{bm}

\usepackage{amsmath,amsfonts,bm}

\usepackage{cancel}

\usepackage{tikz}
\usetikzlibrary{automata}
\usetikzlibrary{calc}
\usetikzlibrary{positioning}  
\usetikzlibrary{arrows,arrows.meta} 
\tikzset{
	main node/.style={circle,draw,minimum size=1cm,inner sep=0pt},
}

\usepackage{adjustbox}
\usepackage{graphicx}
\usepackage{wrapfig}

\usepackage{tabularx,booktabs}
\usepackage[flushleft]{threeparttable}
\usepackage{makecell}

\usepackage{calc,pgf}
\def\pgfmath@calc@minof#1#2{min(#1,#2)}

\usepackage{subcaption}
\usepackage{float}

\usepackage{xpatch}

\makeatletter
\xpatchcmd{\paragraph}{3.25ex \@plus1ex \@minus.2ex}{3pt plus 1pt minus 1pt}{\typeout{success!}}{\typeout{failure!}}
\makeatother

\usepackage{mathpazo}

\usepackage{pslatex}
\usepackage{lipsum}
\usepackage[nocomma]{optidef}

\usepackage{amsmath,amsfonts,amssymb,amsthm}

\usepackage{multicol}
\usepackage{setspace}
\usepackage{ragged2e}

\usepackage{etoolbox}
\usepackage{filecontents}

\usepackage{blkarray, bigstrut}
\usepackage{xparse}

\usepackage{hyperref}
\hypersetup{colorlinks = true,urlcolor=blue, citecolor = black,linkcolor=black}

%% file: templates/my-macros.tex
\newcommand{\detailtexcount}{%
  \immediate\write18{texcount -merge -sum -incbib -dir \jobname.tex > \jobname.wcdetail }%
  \verbatiminput{\jobname.wcdetail}%
}

\newcommand{\quickwordcount}{%
  \immediate\write18{texcount -1 -sum -merge \jobname.tex > \jobname-words.sum }%
  \input{\jobname-words.sum}%
}

\newcommand{\quickcharcount}{%
  \immediate\write18{texcount -1 -sum -merge -char \jobname.tex > \jobname-chars.sum }%
  \input{\jobname-chars.sum} characters (not including spaces)%
}

\newcolumntype{Y}{>{\centering\arraybackslash}X}
\newcolumntype{b}{X}
\newcolumntype{s}{>{\hsize=.5\hsize}X}

\usepackage[most]{tcolorbox}

\newtcolorbox{note}[1][]{%
	enhanced jigsaw, 
	borderline west={2pt}{0pt}{red}, 
	sharp corners, 
	boxrule=0pt, 
	fonttitle={\normalsize\bfseries},
	coltitle={black},  
	title={Note:\ },  
	attach title to upper, 
	#1
}

%% file: templates/math-commands.tex
\usepackage{xspace}
\usepackage{econometrics}

\makeatletter
\newcommand\RedeclareMathOperator{%
	\@ifstar{\def\rmo@s{m}\rmo@redeclare}{\def\rmo@s{o}\rmo@redeclare}%
}
\newcommand\rmo@redeclare[2]{%
	\begingroup \escapechar\m@ne\xdef\@gtempa{{\string#1}}\endgroup
	\expandafter\@ifundefined\@gtempa
	{\@latex@error{\noexpand#1undefined}\@ehc}%
	\relax
	\expandafter\rmo@declmathop\rmo@s{#1}{#2}}
\newcommand\rmo@declmathop[3]{%
	\DeclareRobustCommand{#2}{\qopname\newmcodes@#1{#3}}%
}
\@onlypreamble\RedeclareMathOperator
\makeatother

\def\eqref#1{equation~\ref{#1}}

\def\1{\bm{1}}

\NewDocumentCommand{\grad}{e{_^}}{%
	\mathop{}\!
	\nabla
	\IfValueT{#1}{_{\!#1}}
	\IfValueT{#2}{^{#2}}
}

\newcommand{\minus}{\scalebox{0.75}[1.0]{$-$}}

\newcommand{\GD}{\texttt{GD}\xspace}
\newcommand{\SGD}{\texttt{SGD}\xspace}
\newcommand{\NGD}{\texttt{NGD}\xspace}

\newcommand{\SQP}{\texttt{SQP}\xspace}
\newcommand{\GN}{\texttt{GN}\xspace}
\newcommand{\LM}{\texttt{LM}\xspace}
\newcommand{\method}{\texttt{method}\xspace}
\newcommand{\MPEC}{\texttt{MPEC}\xspace}
\newcommand{\NLLS}{\texttt{NLLS}\xspace}
\renewcommand{\OLS}{\texttt{OLS}\xspace}

\newcommand{\RSS}{\texttt{RSS}\xspace}
\newcommand{\MLE}{\texttt{MLE}\xspace}
\newcommand{\DUE}{\texttt{DUE}\xspace}
\newcommand{\SUE}{\texttt{SUE}\xspace}
\newcommand{\MSA}{\texttt{MSA}\xspace}
\newcommand{\DSD}{\texttt{DSD}\xspace}
\newcommand{\SNL}{\texttt{SNL}\xspace}
\newcommand{\ODE}{\texttt{ODE}\xspace}
\newcommand{\CLT}{\texttt{CLT}\xspace}
\newcommand{\ODLUE}{\texttt{ODLUE}\xspace}
\newcommand{\LUE}{\texttt{LUE}\xspace}
\newcommand{\UE}{\texttt{UE}\xspace}
\newcommand{\logit}{\texttt{logit}\xspace}
\newcommand{\MNL}{\texttt{MNL}\xspace}
\newcommand{\NGLS}{\texttt{NGLS}\xspace}
\newcommand{\BFGS}{\texttt{BFGS}\xspace}
\newcommand{\DGP}{\texttt{DGP}\xspace}

\newcommand{\NRMSE}{\texttt{NRMSE}\xspace}

\def\ru{{\textnormal{u}}}

\def\rx{{\textnormal{x}}}

\def\rvu{{\mathbf{i}}}

\def\rvu{{\mathbf{u}}}

\def\rvx{{\mathbf{x}}}

\def\vzero{{\bm{0}}}
\def\vone{{\bm{1}}}
\def\vmu{{\bm{\mu}}}
\def\vtheta{{\bm{\theta}}}

\def\vf{{\bm{f}}}
\def\vg{{\bm{g}}}

\def\vq{{\bm{q}}}

\def\vt{{\bm{t}}}
\def\vu{{\bm{u}}}
\def\vv{{\bm{v}}}

\def\vx{{\bm{x}}}
\def\vy{{\bm{y}}}
\def\vz{{\bm{z}}}

\def\vpf{{\bm{p}}}

\def\mA{{\bm{A}}}
\def\mB{{\bm{B}}}

\def\mI{{\bm{I}}}

\def\mQ{{\bm{Q}}}

\def\mX{{\bm{X}}}

\def\mZ{{\bm{Z}}}

\def\mIq{{\bm{\Delta}_q}}
\def\mIx{{\bm{\Delta}_x}}
\def\mIqT{{\bm{\Delta}^\top_q}}
\def\mIxT{{\bm{\Delta}^\top_x}}

\DeclareMathAlphabet{\mathsfit}{\encodingdefault}{\sfdefault}{m}{sl}
\SetMathAlphabet{\mathsfit}{bold}{\encodingdefault}{\sfdefault}{bx}{n}

\def\sR{{\mathbb{R}}}

\renewcommand{\E}{\mathbb{E}}

\newcommand{\Var}{\mathrm{Var}}

\DeclareMathOperator*{\argmin}{arg\,min}

\DeclareMathOperator{\sign}{sign}

\newtheoremstyle{boldremark}
{\dimexpr\topsep/2\relax} 
{\dimexpr\topsep/2\relax} 
{}          
{}          
{\bfseries} 
{.}         
{.5em}      
{}          

\newtheorem{definition}{Definition}

\newtheorem{prop}{Proposition}

\theoremstyle{plain}
\newtheorem{assumption}{Assumption}

\theoremstyle{boldremark}
\newtheorem{remark}{Remark}

\DeclareMathAlphabet{\mathcal}{OMS}{cmsy}{m}{n}

%% file: templates/ml-opt-packages.tex
\usepackage{caption}
\usepackage{algorithm,algpseudocode}

\floatname{algorithm}{Algorithm}
\renewcommand{\algorithmicrequire}{\textbf{Input:}}

\usepackage{vwcol}  

\newcounter{algsubstate}

\newenvironment{algsubstates}
{\setcounter{algsubstate}{0}%
	\renewcommand{\State}{%
		\refstepcounter{algsubstate}%
		\Statex {\footnotesize\alph{algsubstate}:}\space}}
{}

\DeclareMathOperator{\st}{s.t.}

\usepackage{optidef}

%% file: sections/title.tex
\title{Statistical inference of travelers' route choice preferences\\ with system-level data}

\author[add1]{Pablo Guarda}

\author[add1,add2]{Sean Qian}
\ead{seanqian@cmu.edu}
\address[add1]{Department of Civil and Environmental Engineering, Carnegie Mellon University}
\address[add2]{Heinz College, Carnegie Mellon University}

%% file: sections/abstract.tex
\begin{abstract}

Traditional network models encapsulate travel behavior among all origin-destination pairs based on a simplified and generic travelers' utility function. Typically, the utility function consists of travel time solely and its coefficients are equated to estimates obtained from discrete choice models and stated preference data. While this modeling strategy is reasonable, the inherent sampling bias in individual-level experimental data may be further amplified over network flow aggregation, leading to inaccurate flow estimates. In addition, individual-level data must be collected from surveys or travel diaries, which may be labor intensive, costly and limited to a small time period. To address these limitations, this study extends classical bi-level formulations to estimate travelers’ utility functions with multiple attributes using system-level data. This data tends to be less subject to sampling bias than individual-level data. It is cheaper to collect and it has became increasingly diverse and available. To leverage system-level data, we formulate a methodology grounded on non-linear least squares to statistically infer travelers' utility function in the network context using traffic counts, traffic speeds, the number of traffic incidents and sociodemographic information obtained from the US Census, among other attributes. The analysis of the mathematical properties of the optimization problem and of its pseudo-convexity motivate the use of normalized gradient descent, an algorithm developed in the machine learning community that is suitable for pseudo-convex programs. More importantly, we develop a hypothesis test framework to examine statistical properties of coefficients attached to utility terms and to perform attributes selection. Experiments on synthetic data show that the coefficients of the travelers’ utility function can be consistently recovered and that hypothesis tests are a reliable statistic to identify which attributes are determinants of travelers’ route choices. Besides, a series of Monte-Carlo experiments showed that statistical inference is robust to noise in the Origin-Destination matrix and in the traffic count measurements, and to various levels of sensor coverage. The methodology is also deployed at a large scale using real-world multi-source data in Fresno, CA collected before and during the COVID-19 outbreak.

\end{abstract}


%% file: sections/introduction.tex
\section{Introduction}

System-level data, as opposed to individual-level data, measures characteristics of aggregated flow in transportation networks, and it is a valuable source of information for studying travel demand. An example is the use of traffic count data for estimating origin-destination (O-D) matrices \citep{Fisk1989, Yang1992,Cascetta2001a, Ma2018a, Krishnakumari2020a}. Interestingly, few attention has been given to use of system-level data to estimate travelers' utility functions in route choice models. These models are key to depict travel demand in transport planning applications and their parameters are usually estimated with data collected from individual-level surveys and experiments. Modeling route choices at the system level using  individual-level data, however, may be inaccurate and expensive due to the inherent sampling bias and high collection cost of experimental data. The increase in availability and diversity of system-level data offers opportunities to overcome those limitations and to understand the impact of a broader set of factors that may influence travelers' route choice decisions. Some of the new sources of system-level data include crashes, weather and pavement conditions, land use characteristics, socio-demographics information from Census, travel time reliability and other relevant attributes. Clearly those factors may determine how a traveler makes his/her route choice, but surveying those for each individual traveler would be infeasible. Thus, this paper proposes a statistical method to better understand why and how a traveler make route choices using a diverse set of system-level data only.

Researchers have already leveraged system level data within traditional network models to estimate travelers' route choice preferences  \citep{Robillard1974, Fisk1977, Daganzo1977,Anas1990, Liu1996, Yang2001,Lo2003,Garcia-Rodenas2009,Russo2011a,Wang2016}. However, their primary focus has been on the estimation of a utility function dependent on travel time only. Furthermore, most existing algorithms has been deployed to network of relatively small size, which raises questions on their potential to scale up to large transportation networks and to contribute to real world applications. A main difficulty to solve this estimation problem comes from the fact that network models are not data-driven. Discrete choice models, the gold standard to analyze individual-level data in travel behavior studies, may be considered a promising alternative to solve this problem given their data-driven nature. However, these models are not able to capture the endogeneity of travel times arising from the interdependence between travelers' decisions and from the effect of traffic congestion in transportation networks.

With the goal of estimating the coefficients of travelers' utility functions with multiple attributes and using system-level data, this paper enhances traditional network models with optimization algorithms developed in the machine learning literature. Mathematical properties of the optimization problem and its pseudo-convexity motivates the use of normalized gradient descent, a first order method that is suitable for pseudo-convex optimization. Our solution algorithm is tailored to perform a robust estimation of the coefficients of the travelers' utility function and to also reproduce the traffic conditions observed at the system level. To ease the identification of the determinants of travelers' route choices with system level data, we formulate hypothesis tests on the utility function coefficients.

The paper is structured as follows. We first survey the existing literature and identify the main contributions of our work. Then, we present an example on a toy network to illustrate the optimization problem arising from the estimation of the travelers' utility function coefficients using system level data. The following three sections describe the formulation of our methodology on a general transportation network, analyze the mathematical properties of the optimization problem and propose a solution algorithm to solve the problem. Subsequently, we describe the framework to perform statistical inference on the utility function coefficients and we present the results of numerical experiments conducted in networks of small and medium size. Next, we show the estimation results obtained in a large scale network and using real world system level data. Finally, we describe our main conclusions and suggests avenues for further research. The mathematical notation used for the remainder of the paper is included in \ref{appendix:ssec:notation}.

%% file: sections/literaturereview.tex
\section{Literature review and research gaps}

The computation of network equilibrium requires to specify a travelers' route choice model. The choice of route choice model defines the class of network equilibria \citep{Prashker2004}. The simplest and most traditional class of equilibria is known as deterministic user equilibrium (\DUE). At \DUE no traveler has incentive to change to alternative routes and travelers' are assumed to pick, in a deterministic manner, the route that maximize their utility. The multinomial \logit model (\MNL) remains as a gold standard in travel behavior research to depict travelers' decision making \citep{McFadden1973}. A well-known application of the \MNL model in network modeling is for the computation of Stochastic User Equilibrium with \logit assignment (\SUE-\logit). In contrast to other probabilistic models of travelers' choices, the \logit model exhibits a better compromise between behavioral realism and mathematical tractability. Similar to the \MNL model, the \SUE-\logit enjoys of this mathematical convenience in the network modeling context. For instance, under mild conditions, \SUE-\logit has solution uniqueness in link and path flow space whereas \DUE only has solution uniqueness in link flow space. Besides, the existence of a closed form expression for the choice probabilities in the multinomial \logit model can be leveraged for the computation of \SUE-\logit, e.g the Dial \logit algorithm \citep{Dial1971,Bell1995} or to solve the \ODE problem \citep{Ma2018}. 

A stream of research that brought our attention concerns the estimation of the travelers' utility function from traffic count data and which induces a traffic assignment consistent with \SUE-\logit \citep{Robillard1974, Fisk1977, Daganzo1977,Anas1990, Liu1996, Yang2001,Lo2003,Garcia-Rodenas2009,Russo2011a,Wang2016}. The solution of this problem and of network equilibrium in general requires to know a priori the coefficients of the travelers' utility function of the route choice model. Hence, a standard practice is to set the values of these coefficients equal to estimates obtained in previous travel behavior studies. Nevertheless, there are multiple advantages to estimate these coefficients. First, it avoids to search for external estimates which may be cumbersome and it may provide inaccurate coefficients. Second, it can improve the generalization performance on the estimated O-D matrix when, for instance, the coefficients fitted from existing count data become close to the population coefficients. The following sections describe the relevant literature that has studied the aforementioned problem.

\subsection{Overview of the Logit Utility Estimation (LUE) problem}


Seminal work in the literature focused on the problem of estimating the coefficient $\hat{\theta} \in \sR$ of a utility function dependent solely on travel time and where both the O-D matrix and travel costs among links/paths are assumed exogenous. To our knowledge, the \LUE problem was first studied by \cite{Robillard1974}, who estimated the coefficient $\hat{\theta}$ using link and path flows obtained from a traffic assignment consistent with the \cite{Dial1971} method. The estimator $\hat{\theta}$ was assumed to be Gaussian distributed and it was obtained via maximum likelihood estimation (\MLE). A main limitation of this work was the requirement of knowing the traffic counts at every link/path in the network. The problem of estimating $\hat{\theta}$ from a subset of traffic count in the network was then addressed by \citet{Fisk1977}. Another major limitation of this work was assuming full knowledge of the average cost of traversing paths between O-D pairs, which is plausible in settings where costs are equated to travel times but arguably unrealistic in real world scenarios where travelers make route choices based on multiple attributes such as monetary cost and travel time reliability.

\cite{Daganzo1977} focused on a more general setting where link flows observations were consistent with \SUE but not necessarily with a \logit assignment. In contrast to \citet{Robillard1974}, the distribution of the estimator $\hat{\theta}$ was derived by imposing distributional properties on the link flows and by relying on the statistical properties of the \MLE estimator. This work is the first that conducts hypothesis testing on the coefficients of the travelers' utility function. However, similar to previous research, it assumes a utility function dependent on travel time only and exogenous link costs, which limits the application to uncongested networks or to settings where link costs at equilibrium are known. Under the assumption that link flow data is consistent with \SUE-\logit, \citet{Anas1990} extended prior work by allowing for endogenous travel costs. The authors studied the impact of estimating $\hat{\theta}$ without accounting for the endogeneity of the travel costs and observed that estimates of $\hat{\theta}$ become biased and that statistical inference is less consistent. Besides, the estimation of $\hat{\theta}$ was shown to perform better with least squares than maximum likelihood.

\subsection{Overview of the Origin-Destination and Logit Utility Estimation (ODLUE) problem}

In the following decades, the research interest shifted to addressing the joint estimation of the O-D matrix and $\hat{\theta}$ from traffic count data following \SUE-\logit, which we coined as O-D and Logit Utility estimation (\ODLUE) problem\footnote{We denominated this problem as O-D Logit Utility estimation (\ODLUE) because the O-D matrix is estimated in addition top the travelers' utility function coefficients as it is done in \LUE}. Interestingly, the extension of the \LUE problem to settings where the utility function was dependent on multiple attributes besides travel time did not gain the same attention. By this time, the \ODE literature had made significant progress to estimate O-D matrices with traffic count data consistent with \SUE-\logit and these advances could be directly leveraged to solve the \LUE problem. A rich set of solution methods had been developed in the \ODE literature, which include sensitivity analysis \citep{Patriksson2004} and the alternating optimization of the upper and lower levels of the classic \ODE bilevel formulation \citep{Ma2017}.

To our best knowledge, \citet{Liu1996} conducted the first study that studied the \ODLUE problem and that presented an application using system level data collected from a real transportation network. Traffic count and travel time data was collected from a small network at the Purdue University campus and estimates of $\hat{\theta}$ were obtained for different time periods of the day. The authors proposed a two stage calibration method that leveraged the closed form of the \logit route choice probabilities to solve for $\hat{\theta}$ using the secant method. The O-D matrix was estimated via a least square minimization. Key limitations of this work were to not account for the congestion effects in the transportation network and to assume that traffic counts and travel time data were available for the full set of links. This assumption could be mild in small networks where traffic counts can be collected at every link but implausible in larger networks where this data is typically available for a small proportion of the links.

\citet{Yang2001} overcame some of the aforementioned limitations by solving a bilevel optimization problem that resembled the traditional mathematical program with equilibrium constraints (\texttt{\MPEC}) that had been used for solving the \ODE problem. A main contribution of this work was to estimate both $\hat{\theta}$ and the O-D matrix while also accounting for the impact of traffic congestion. A Sequential Quadratic Programming (\SQP) algorithm was used to find the parameters of interest and to minimize the gap between observed and predicted traffic counts. The optimization problem also included constraints to ensure that the \SUE-\logit equilibrium conditions were satisfied over parameters' updates. The authors also derived for the first time the analytical expression of the gradients of the outer level objective respect to the parameters associated to the O-D matrix and to the travelers' utility function. \citet{Lo2003} proposed a similar framework but where a \MLE instead of a non linear least square (\NLLS) problem was solved at the upper level of the bilevel formulation. The authors performed \ODLUE on both synthetic and real world data collected from the Tuen Mun Corridor network in Hong Kong. Later, \cite{Wang2016} implemented a similar approach using real world data gathered from a small network in Seattle, WA. They also performed experiments on synthetic data to study the robustness of the parameter estimates to noise in the ground truth O-D matrix and the parameter $\hat{\theta}$ used to generate synthetic traffic count data.

\subsection{Extensions of the O-D Logit Utility Estimation (ODLUE) problem}

A number of studies have estimated additional parameters on top of the O-D matrix and of the travelers' utility function coefficients.  \citet{Russo2011a} and \cite{Caggiani2011} estimated the shape parameters of the link performance function using traffic counts and travel time measurements consistent with \SUE-\logit.  To leverage the travel time data, the outer level objective of their \NLLS problem incorporated the squared difference between predicted and observed travel times. Results on synthetic data suggested that all parameters could be consistently recovered. However, few attention was paid to the identifiability of the parameters and on overfitting due to the additional degrees of freedom introduced in the model. In fact, \ODE is known to be an underdetermined problem which is subject to identifiability issues even when a full set of traffic counts is available in a network \citep{Yang2018b}. Naturally, this issue is worsen if additional parameters are introduced in the model. Besides the problem of identifiability (also known as observability), a high model complexity may cause overfitting and thus, a poor generalization of the model on unseen data even if a good fit is observed on training data. \citet{Russo2011a} also argued that in an uncongested network, the optimization problem is convex to claim solution uniqueness. However, no attention was paid to the non-convexity of the feasible set in the optimization problem, i.e. which makes the problem non-convex, and also to the fact that convexity is not a sufficient condition for solution uniqueness.

Other line of research extended the application of the \ODLUE beyond the traffic assignment stage of the the classic four step model \citep{Ortuzar2011}. Under the assumption of an exogenous demand matrix, \citet{Cascetta1997} simultaneously estimated the parameters of a category index trips generation model, the parameters associated to the distribution step and the mode specific constants, and the travel time and cost coefficients of the utility function of a multinomial logit model of mode choice. The authors tested their method with real data collected from two Italian cities and found that the magnitude and sign of the estimated parameters were consistent with their expectations. However, the treatment of congestion in the road network in their work to capture the endogeneity of travel times is missing. In similar research, \cite{Garcia-Rodenas2009} estimated the utility function parameters of a route choice model and the mode choice parameters using a nested logit model via Nonlinear Generalized Least Squares (\NGLS). A bilevel optimization problem was formulated to account for traffic congestion and it was solved via an alternating optimization of the upper and inner level problems. Similar to most prior work, no statistical tests were conducted to check if those parameters were statistically significant or to assess if the increase in goodness of fit compensated the additional degrees of freedom introduced in the full model.

More recently, \cite{Wu2018a} used a computational graph to estimate the parameters of a model that incorporated the steps of trip generation, spatial distribution and path flow-based traffic assignment of the 4-step model. Data from household travel surveys, mobile phones and on traffic counts were leveraged to estimate these parameters and the model was deployed on a network with more than 2500 nodes and 5000 links. This research illustrated the potential of computational graphs to increase the scalability of models that have been developed in the \ODLUE literature and to also integrate heterogeneous data sources that inform about transport decisions. Nevertheless, a main drawback of this work is to not account for traffic congestion and to focus on a travelers' utility function dependent on travel time only.

\subsection{Contributions of this research}

This paper enhances existing formulations of the \LUE problem with the goal to statistically infer the coefficients of travelers' utility functions with multiple attributes using system-level data. A bilevel optimization program is formulated to both minimize the gap between estimated and observed traffic counts, and to account for the endogenous effect of traffic congestion on travelers' route choices. More importantly, a framework of hypothesis tests is proposed to examine statistical properties of the coefficient estimates. Under the pseudo-convexity of the optimization problem, we implement a variant of gradient descent suitable for pseudo-convex optimization. First order optimization methods seem promising to speed up computation and thus, to make model training scalable to large transportation networks. 

Below are the four main contributions of this paper to existing literature: 

\begin{enumerate}
    \item It presents a methodology to statistically infer the coefficients of travelers' utility functions consisting of multiple attributes and using system-level data. It does not only extend the \LUE problem, but it also enables learning realistic travel behavior from diverse datasets at the system level. 
    \item It conducts for the first time a mathematical analysis of the non-convexity of the \LUE problem respect to the utility function coefficients estimated from traffic count data. 
    \item For the solution of the \LUE problem, it shows that the integration of first order and second order optimization methods outperform existing approaches in the literature.
    \item It presents a statistical framework to perform hypothesis testing and attributes selection on the coefficients of multi-attribute utility functions
    \item  It implements the methodology in a large-scale transportation network and with real system-level data. 
\end{enumerate}

All our analyses are replicable, open-sourced and ready to be used for the transportation community (Section \ref{sec:model-implementation-data}).

%% file: sections/methodology.tex

\subsection{Illustrative example}
\label{sec:illustrative-example}

Consider a network with two parallel links and thus, with two alternative paths only (Figure \ref{fig:illustrative-example}). The costs functions associated to each link/path are $t_1(x_1) = t^{0}_{1}(1+\alpha(\frac{x_1}{\gamma_1})^\beta)$ and $t_2(x_2) = t^{0}_{2}(1+\alpha(\frac{x_2}{\gamma_2})^\beta)$, where $\gamma_1,\gamma_2 \in \sR$ are the link capacities and $t^0_1=t^0_2$ are the links' free flow travel times. The parameters of the link performance functions are $\alpha = 0.15$ and $\beta = 4$. Assume that the travelers' utility function is dependent of the travel time $t$ and the monetary cost $c$ of traversing each link of a path only. In addition, suppose that the coefficients $\theta_t, \theta_c \in \sR$ are linearly weighting each attribute of the utility function and that they are common among travelers. Thus, the deterministic component of the utility attained to each link/path can be expressed as $v_1 (x_1) = \theta_t t_1(x_1) + \theta_c c_1$ and $v_2 (x_2) = \theta_t t_2(x_2) + \theta_c c_2$. 

	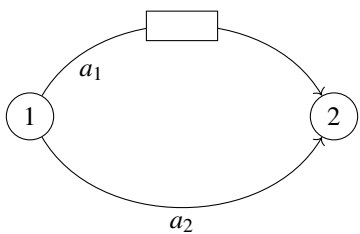
\begin{figure}[H]
		\centering
			
			\begin{tikzpicture}[x=2cm,y=2cm]
				
				\node at (0,1) [circle, draw] (N1) {$1$} ;
				\node at (2,1) [circle, draw] (N2) {$2$} ;

				\path[draw, ->] (N1) edge[bend left = 60] node[below left = 0.4cm and 0.9cm]{$a_1$}  (N2) ;
				\path[draw, ->] (N1) edge[bend right= 60] node[below]{$a_2$} (N2);
				
				
				\node at (1,1.6) [rectangle, draw, fill = white] (Toll) {$\phantom{aaaa}$};

			\end{tikzpicture}    
			
			\caption{Illustrative example}
			\label{fig:illustrative-example}
			
		\end{figure}

		Consider $\bar{q} \in \sR_+$ individuals traveling between origin destination 1-2 and making route choices consistent with a \logit model. $\bar{q}$ is small enough such that links' travel times are approximately equal to their free flow travel times, i.e. $t_1(x_1) \approx t_1^0, \ t_2(x_2) \approx t_2^0$. Suppose link $a_1$ has a toll fee equal to 1 USD, i.e $c_1 = 1, c_2 = 0$ and a sensor recording traffic counts. Then, path choice probabilities can be modeled with a sigmoid function and link flows can be obtained as follows: 
		\begin{equation}
			\label{eq:link-flow-solution-illustrative-example}
			x_1 = \frac{\exp(\theta_t t_1+\theta_c c_1)}{\exp{(\theta_t t_1+\theta_c c_1)}+\exp{(\theta_t t_2+\theta_c c_2)}} \ \bar{q} =  \sigma(\theta_c) \bar{q},\quad x_2 = \bar{q} -x_1 = (1-\sigma(\theta_c))\bar{q} = \sigma(-\theta_c)\bar{q}
		\end{equation}
		
		Suppose the goal is to estimate the coefficients $\theta_t$ and $\theta_c$ of the travelers' utility function with the traffic count measurement $\bar{x}_1$. From the link flow solutions (Eq. \ref{eq:link-flow-solution-illustrative-example}), we note first that $\theta_t$ is not identifiable because travel times are the same in the two alternative paths. In contrast, because the monetary costs are different between paths, $\theta_c$ is identifiable and equal to the solution $\theta^{\star}_c$ of the following nonlinear least square minimization problem:
		\begin{align}
			\label{eq:objective-function-illustrative-example}
			\theta^{\star}_c
			= \arg \min_{\theta_c } \ (x_1(\theta_c) - \bar{x}_1)^2 
			= \arg \min_{\theta_t } \left(\sigma(\theta_c) - \frac{\bar{x}_1}{\bar{q}}\right)^2
			\rightarrow 
			\theta^{\star}_c = \ln(\bar{q}/\bar{x}_1-1)
		\end{align}
		
		Consider two scenarios for the solution of this problem (Eq. \ref{eq:objective-function-illustrative-example}). In the first scenario $\bar{x}_1 = 0.5 \bar{q}$ meaning that $\bar{x}_1$ is deterministic and that $\theta_c^\star = 0$. In the second scenario, $\bar{x}_1 = x_1 + \epsilon$ , $\epsilon \sim \mathcal{N}(0,1)$, i.e. $\bar{x}_1$ is not deterministic and $\bar{x}_1 \sim \mathcal{N}(0.5\bar{q},1)$. Note that a single realization of $\bar{x}_1$ in the second scenario could lead us to conclude with the same probability that the traveler's utility and the monetary cost are positively or negatively associated. However, in expectation and similar to the first scenario, this association does not exist. 
		
		This example illustrates important features of the problem studied in this paper. First, it shows that the non-linear least squares formulation of the problem is suitable to estimate the coefficients of the travelers' utility function with traffic count data. Second, the identification issue associated to $\theta_t$ illustrates the special considerations that must be made to specify the travelers' utility function and to prevent a naive estimation of coefficients. Third, the estimation of $\theta_c$ in the non-deterministic scenario illustrates the relevance of statistical inference tools to assess if a feature is a determinant of the travelers' route choices in the presence of randomness of the data generating process.
		
		For larger networks, the analytical derivation of the optimal solution for travelers' utility function parameters is intractable and thus methods that can iteratively approach to the optimal solution are needed. Besides, due to the class of non-linearity of the link flows respect to the utility function coefficients, the objective function cannot be guaranteed to be convex, which imposes additional difficulties to search for an optimal solution. Another challenge is related to the impact of traffic congestion on travelers' choices and which is not addressed in the illustrative example. 
		
		\section{Estimation of travelers' utility function coefficients with system level data}
		\label{sec:learning-parameters}
		
		This section discusses in detail our methodology to estimate the travelers' utility function coefficients using system level data. We first introduce our main assumptions and then we describe the components of the bilevel optimization program formulated to estimate the utility function coefficients. 
  
		\subsection{Assumptions}
		\label{ssec:assumptions}
		
		\medskip

		\begin{assumption}[SUE-Logit]
			\label{assumption:sue} 
			Network traffic flow follows stochastic user equilibrium with logit assignment  
		\end{assumption}

		With the motivation of bridging transportation network analysis and the study of travel behavior with discrete choice models, we focused on the problem of estimating the travelers' utility function from link flow measurements consistent with stochastic user equilibrium under \logit assignment (\SUE-\logit). On one hand, \SUE-\logit is one of the many alternative representations of user equilibrium used in the networking modeling community, hence it may be considered a strong assumption. On the other hand, \SUE-\logit is in line with the state of the practice in travel behavior research where the \logit model remains as one of the gold standards to depict individual route choices. The good compromise between behavioral realism and the mathematical tractability of \logit is a key property for the estimation in discrete choice models and it will be also leveraged in the formulation and solution of our problem. 
		
		\begin{assumption}[Exogenous and deterministic O-D]
			\label{assumption:od}
			The origin-destination (O-D) demand matrix is deterministic and exogenous  \end{assumption}
		
		O-D demand estimation (\ODE) is a complex problem that has been largely studied in the network modeling literature. As an underdetermined problem, there is no guarantee on solution uniqueness for the estimated O-D matrix and this makes difficult to verify the consistency of the solution with a ground truth O-D matrix. The \ODLUE literature has extended \ODE to estimate the travelers' utility function coefficients on top of the O-D matrix. The \ODLUE problem is harder to solve due to the additional degrees of freedom introduced by an endogenous O-D matrix and therefore, as \ODE, it also suffers of solution non-uniqueness. 
		
		The assumption of an exogenous and deterministic O-D matrix allows us to focus on the estimation of the utility function coefficients, a first step to understand the mathematical properties of our problem. It can always be easily extended with any \ODE method to solve for \ODLUE. In the context of the existing literature, our methodology can be seen as a case where a reasonably accurate reference O-D matrix is available and assumed as the ground truth O-D matrix. As we will show later, we can still obtain satisfactory results if the exogenous O-D demand is noisy but not too far from the ground truth O-D.  Note also that our methodology can be extended to an iterative process where \ODE and \LUE alternate if the goal is to simultaneously estimate both the travelers' utility function coefficients and the O-D matrix. 
		
		\begin{assumption}[Linearity and homogeneity of utility function]\label{assumption:utility} The utility function is a linear weight of attributes and coefficients and the coefficients are common/homogeneous for all individuals traveling in the transportation network. \end{assumption}
		
		The choice of a linear-in-parameters utility function and with homogeneous coefficients is a standard assumption for the estimation of discrete choice models in travel behavior studies. The preference of linear over non-linear specifications of the utility function is motivated by parsimony arguments and on the convenience of dealing with a concave likelihood function in the estimation of multinomial logit models (\MNL) models. In our case, this assumption significantly facilitates the proofs of theoretical guarantees in our optimization problem. Note that this assumption could be relaxed by allowing a variation of the utility function coefficients at the O-D pair level or by letting the coefficients to be drawn from some arbitrary probability density function. Both modeling strategies can be seen as analogies to the systematic taste variation used in \MNL models and to the mixing distributions used in Mixed Logit models \citep{Train2002}.
		
		\begin{assumption}[Monotonic increase of and exogeneity of parameters of link performance functions]\label{assumption:link-performance-functions} The link performance functions are monotonically increasing respect to their link flows and their parameters are exogenous \end{assumption}
		
		The exogeneity of the parameters of the link performance is a standard assumption in transportation network analysis. Following the state of the practice in prior literature, we adopt the Bureau of Public Roads (BPR) link travel time function. For each link, the travel time monotonically increases respect to the link flow. For the sake of simplicity, we also assume that there is no interaction between links and hence, the travel time of a link is dependent on the traffic flow in that link only. A direct consequence of this assumption is the uniqueness of the link flow solution of \SUE-\logit. Some work in the literature has proposed methods to estimate the parameters of the link performance function, namely, relaxing the exogeneity assumption. Although it is feasible to estimate these parameters on top of the travelers' utility function coefficients within our modeling framework, the introduction of these additional degrees of freedom may cause identifiability issues or it could make the statistical inference of the utility function coefficients less reliable.

		\begin{assumption}[Constrained path set]
			\label{assumption:path-set}
			The set of paths between each O-D pair has a fixed size and it is dynamically updated over iterations with column generation methods
		\end{assumption}
		Most prior work on the \LUE problem assumes that the consideration set in an O-D pair is equal to the set all reasonable paths in that pair \citep{Yang2001,Wang2016,Lo2003,Liu1996}. Under this assumption, each path choice probability can be decomposed into a weight of exponentials of their link utilities. On one hand, this property avoids the enumeration of all paths in the network and it can significantly speed up the computation of \SUE-\logit that it is required to estimate the travelers' utility function coefficients. On the other hand, it implicitly assumes that travelers' consider all paths to travel between an O-D pair and hence, it may induce overlapping among paths. Path overlapping is known to affect the computation of path choice probabilities in multinomial logit models (\MNL) and it has been identified as a drawback of algorithms that rely on the set of all reasonable paths \citep{Yang2001,Wang2016,Lo2003,Liu1996}. 
		
		Navigation apps are widely used nowadays and they recommend a constrained set of options to travel between an O-D pair. This makes us believe that in modern real world applications the assumption of a constrained path set between O-D is more behaviorally plausible than considering all reasonable paths. Besides, a constrained path set is also expected to induce less path overlapping, which is detrimental \logit route choice models. A challenge of working with a constrained path set is the impossibility of performing a link level decomposition of the path choice probabilities and hence, the potential increase of computational cost in the network loading stage of \SUE-\logit. In addition, while path sets can be reasonable approximated with existing information, e.g. Google Maps recommendation, they may be inaccurate or not available for all O-D pairs. To improve our prior about the true composition of the path set, the path sets are dynamically updated via column generation methods \citep{Damberg1996,Garcia-Rodenas2009}. The size of the consideration sets is treated as a hyperparameter and it is constrained to small values to reduce computational burden.
  
    To capture the correlation between paths in the consideration set due to shared link segments, path utilities are corrected with the same principle used in the path size logit (PSL) model \citep{Ben-Akiva1999a}. Path utilities incorporate the logarithm of a factor that increase with the amount of overlapping among paths within the same consideration set. A path with no overlapping links needs no utility adjustment since the PSL factor has a size of one. Thus, the PSL correction reduces the chances of a violation of the Independence of Irrelevant Alternatives (IIA) assumption of the \MNL model. Previous literature suggests to generate paths that do not overlap with existing paths during the column generation phase \citep{Damberg1996}. However, this approach requires to solve a combinatorial problem on all possible subset of paths and it implicitly assumes that travelers' consideration set have paths with few or no overlap.

		\subsection{Stochastic user equilibrium with logit assignment (\SUE-\logit)}
		\label{ssec:stochastic-user-equilibrium}
		
		In a congested network, travel times are endogenous and dependent of the traffic flows. At equilibria, both traffic flows and travel times are expected to reach an stationary point where travelers have no incentives to switch to alternative paths. Depending on the underlying behavioral representation used to model travelers' route choices, different types of network equilibrium are induced. In Deterministic User Equilibrium (\DUE), travelers' choices are deterministic, meaning that the modeler fully knows the specification of the traveler' utility function. In contrast, in stochastic user equilibrium (\SUE), travelers' choices are assumed to be probabilistic and modelers' may be ignorant of a set of unobservable components of the traveler' utility function.  In the network context, at \SUE, no user believes he can improve his travel cost by unilaterally changing routes \citep{Daganzo1977b}. Under \SUE with \logit assignment (\SUE-\logit), travelers' are assumed to make choices consistent with a multinomial \logit model (\MNL).
		
		\subsubsection{Multinomial Logit (\MNL) route choice model}
		\label{ssec:multinomial-logit-route-choice-model}

			The \logit model remains as the gold standard to model route choices due to its good compromise between behavioral realism and mathematical tractability (McFadden, 1973). Consider a traveler $l \in L$ choosing a path $i$ within her consideration set $J_{l}$ and based on the latent utility $U_{jl}$ of each alternative path $j \in J_{l}$. Suppose that the modeler knows the observable component $V_{jl}$ of the travelers' latent utility but ignores a component $e_{jl}$ which can be safely assumed to be stochastic. Therefore, if any attribute relevant to route choice decision is unobserved, travelers' choices could look probabilistic but not deterministic from the modelers' perspective. In particular, if $e_{jl} \overset{\text{i.i.d}}{\sim} \ EV(0,\mu)$ the probability $p_i$ that a traveler $l$ chooses a path $i \in J_{l}$ will have the following closed form: 
			
			\begin{align}
				\label{eq:mnl-choice-probabilities}
				\displaystyle p_i = \dfrac{\displaystyle \exp(\mu V_i)}{\displaystyle \sum_{j \in A_{l}} \exp(\mu V_j) }
			\end{align}
			
			where $\mu > 0$ is a scale parameter of the Extreme Value (EV) Type 1 distribution and which is set to 1 for convenience and identification purposes. $\mu$ is inversely proportional to the variance of the random component, and thus, it is expected to be lower as smaller is the unobserved component of the latent utility. At a extreme case where the latent utility function is fully observed, $\mu \to \infty$ and the choice probabilities reduce to an indicator function taking the value 1 if the utility of an alternative is the highest within a choice set and 0 otherwise, i.e. the \DUE case.

			\subsubsection{SUE-logit with multi-attribute utility function}
			\label{ssec:sue-logit-multi-attribute-utility-function}

		The original formulation of the \SUE-\logit problem entails that route choices are made based on travel cost/time only \citep{Fisk1980}. Utility, however, may depend on additional attributes such as travel time reliability, waiting time or monetary costs (e.g. toll fees). To make a more explicit bridge between the travelers' utility function specified in route choice models and to extend the analysis for a multi-attribute case, we reformulated the \SUE-\logit problem as follows:
		
		\begin{maxi}
			{\vx, \vf}{  \sum_{a \in A}\int_{0}^{x_a} v_a(u,\vtheta) du - \left\langle \vf, \ln \vf\right\rangle }{}{}
			\addConstraint{\mIq\vf}{= \vq}{}
			\addConstraint{\mIx\vf}{= \vx}{}
			\addConstraint{\vx,\vf}{\geq \vzero}{}
			\label{eq:multi-attribute-sue-logit}
		\end{maxi}
	
		\bigskip
		where $\vx \in \sR_{\geq 0}^{|A|}, \vf \in \sR_{\geq 0}^{|H|}, \vq \in \sR_{+}^{|W|}, \mIq \in \sR_{\{0,1\}}^{|W| \times |H|}, \ \mIx \in \sR_{\{0,1\}}^{|A| \times |H|}$. Besides, $v_a (u, \vtheta) = \theta_t t_a(u) + \sum_{k \in K_{\mZ}} \theta_k \cdot Z_{ak}$ is the utility associated to link $a$ and at a traffic flow level $u$, $Z_{ak}$ is the value of the exogenous attribute $k \in K_{\mZ}$ at link $a$, $K_{\mZ}$ is the set of exogenous attributes at each link, $\vtheta_Z \in \sR^{|K_Z|}$ and $\vtheta = \begin{bmatrix} \theta_t & \vtheta_Z\end{bmatrix} \in \sR^{|K|}$ are the vector of coefficients associated to the exogenous attributes and to all attributes, respectively. Exogenous attributes may include the number of traffic lights, streets intersections or the level of income of the area where a link is located. From the specification of $v_a(\cdot)$, it follows that travel time ($t$) is the only endogenous attribute in the utility function and it is linearly weighting the preference coefficient $\theta_t$. The first order optimality condition of Problem \ref{eq:multi-attribute-sue-logit} gives the following path flow solution:
		
		\begin{equation}
			\label{eq:sue-logit-path-flows-solution}
			f^\star_{h} = q_{w} \frac{\displaystyle \exp\left( \sum_{a \in A} {v_a^{\star}} \delta_{ah}\right)}{\displaystyle\sum_{j \in H_{w}}^{} \exp\left(\sum_{a \in A} {v_a^{\star}} \delta_{aj}\right)} 
		\end{equation}
		
		where $\delta_{ah} = \mathbb{I}(\textmd{link } a \in \textmd{ path } h)$ and $v^\star_{a}$ is the link utility of link $a \in A$ at \SUE-\logit. As expected, the path flow vector $\vf^\star_{h}$ at \SUE-\logit  follows a \logit distribution. From a micro-level travel behavior standpoint, $q_{w}$ individuals traveling in the O-D pair $w \in W$ are making route choices according to a \logit model (Eq. \ref{eq:mnl-choice-probabilities}), and hence the path flow $\vf^\star_{h}$ is the aggregate of these individual decisions (Eq. \ref{eq:sue-logit-path-flows-solution}). A detailed derivation of the extension of the \citep{Fisk1980} formulation from a single to a multi-attribute case is included in \ref{appendix:ssec:derivations-sue-logit}.

		\begin{remark}
			In line with the rules of parameterization used to derive the \MNL,  the specification of the utility function defined to compute \SUE-\logit should account for the fact that the utility function coefficients are scaled by a factor $\mu \in \sR_{+}$ proportional to the variance of the unobservable component of the utility function i.e. $\vtheta = \mu \tilde{\vtheta}$ where $\tilde{\vtheta}$ is the unscaled vector of \logit coefficients and which is not identifiable. A direct consequence in the single attribute case is that only the sign but not the magnitude of the coefficients should assumed to be known by the modeler. This is particularly relevant in cases where estimates from external travel behavior studies are used to set the utility function coefficients of the route choice model defined for the computation of \SUE-\logit. Therefore, accounting for the scale factor of the \logit parameters here is critical for a correct economic and behavioral interpretation of the coefficients of multi-attribute utility functions but that is oftentimes overlooked in the \LUE and \ODLUE literature. 
		\end{remark}
		
		\subsubsection{Stochastic network loading}
		\label{sssec:stochastic-network-loading}
		
		Suppose the vector of the travelers' preferences $\vtheta \in \sR^{|D|}$, the values of the matrix of exogenous attributes $\mZ$ and the travel times at \SUE-\logit are known and that the goal is to find the resulting path/link flows in the transportation network. This is precisely the problem that stochastic network loading (\SNL) aims to solve. Assumptions about the composition of the consideration set can significantly speed up the computation of \SNL. A well-known example is the Dial algorithm which, under the assumption of path sets containing all reasonable paths, allows to decompose path probabilities into link level weights that can be then used to recursively obtain the resulting links flows in the transportation network. Alternatively, and in line with Assumption \ref{assumption:path-set}, we assume path sets with a constrained size and that are dynamically updated via a column generation method. The pseudo-code of our \SNL method is shown in Algorithm \ref{alg:stochastic-network-loading}, \ref{appendix:ssec:snl}.

		\subsubsection{Solution methods for \SUE-\logit}
		
		If travel times were insensitive to variation of link flows, namely, exogenous, a single computation of \SNL would suffice to obtain the link and path flows at \SUE-\logit. In practice, travel times are endogenous variables and thus, the computation of \SUE-\logit requires to perform \SNL multiple times. 
		
		The method of successive average (\MSA) is one of the prominent algorithms used to compute \SUE-\logit in the \ODLUE literature. At the initial iteration, \SNL is computed to generate a feasible link flow solution. For each of the following iterations, the convex combination of the \SNL link flow solutions at the current and the previous \MSA iteration is computed to obtain a new feasible solution. \MSA defines a step size $\lambda_i = 1/(1+i)$ to set the weights $\lambda_i$ and $1-\lambda_i$ that are used to compute the convex combination of solutions at iteration $i$. Because the feasible set of the \SUE-\logit is convex, the convex combination necessarily gives a feasible point. The process is repeated until some convergence criterion has been achieved, such as the difference between the current and previous feasible link flow solution.
		
		Frank-Wolfe (F-W) is a general purpose method for constrained convex optimization and that can outperform \MSA by allowing to make a smarter choice of the step size parameter $\lambda$. A key thing to notice is that the convex combination of two feasible link flows solutions obtained via \SNL provides a feasible descent direction for the objective function of traffic equilibria problem (Eq.  \ref{eq:multi-attribute-sue-logit}). Thus, the value of $\lambda$ used for the convex combination of solutions is chosen to best improve the \SUE-\logit objective. This solution strategy for \SUE-\logit  was originally implemented by \citet{MingyuanChen1991} with a utility function dependent on travel time only and it is also integrated in the disaggregate simplicial decomposition (\DSD) algorithm developed by \cite{Damberg1996}. An interesting feature of \DSD is the introduction of a column generation phase that dynamically update the consideration sets among O-D pairs. Consideration sets can be augmented with new paths according to different criterion such as the choice probabilities or the level of dissimilarity of the set of candidate paths respect to the existing paths among consideration sets.

		\subsection{Nonlinear least squares}
		
		Non-linear least squares (\NLLS) is a standard estimation method used in the \ODE and \ODLUE  literature. In contrast to maximum likelihood estimation \citep{Lo2003}, there is no need to make distributional assumptions about the data generating process \citep{Cascetta1997} of the traffic counts and there is a variety of specialized algorithms to minimize the \NLLS optimization objective \citep{Bazaraa2006}. Note that when the response function is linear, \NLLS reduces to ordinary least squares (\OLS) and thus, the \NLLS solution could be proved to be unique under mild conditions. Under a non-linear response function, there may exist multiple local minima and saddle points that could make the optimization to be sensitive to the starting points for optimization and to converge toward a local but not the global minima. Therefore, the performance of the optimization algorithms heavily depends on the class of non-linearity of the response function respect to the parameters of interest. 
		
		\subsubsection{Problem formulation}
		\label{sssec:nlls-problem-formulation}
		
		Consider the regression equation:
		
		\begin{equation}
			\label{eq:NLLS-regression-equation}
			\vy = m(\vbeta,\mX) + \vu
		\end{equation}
		\smallskip
		
		where $m(\cdot)$ is the response function, $\mX \in \sR^{|N| \times |K|}$ is the matrix of values of a set of exogenous attributes, $\vy \in \sR^{|N|}$ is the vector of values of the dependent variables,  $\vbeta \in \sR^{|K|}$ is the true parameter vector and $\rvu  \in \sR^{|N|}$ is a vector of random perturbations of some arbitrary distribution. 
		
		The \NLLS problem consists in finding the \NLLS estimator $\hat{\vbeta}_{NLLS}$ that minimizes the residuals, namely, the deviation between the vector of observed measurements $\bar{\vy}$ and the predictions of the response function:
		
		\begin{equation}
			\label{eq:NLLS-argmin-problem}
			\hat{\vbeta}^{\star} 
			= \arg \min_{\hat{\vbeta}} \|\bar{\vy} - m(\hat{\vbeta},\mX)\|_2^2 
		\end{equation}
		
		\medskip
		The application of the \NLLS formulation to our  problem is direct. Let's be $\hat{\vtheta} \in \sR^{|D|}$ the vector of estimated coefficients in the travelers' utility function,  $\vt \in \sR^{|A^o|}$ the vector of link travel times, $\mZ \in \sR^{|A| \times |K_Z|}$ the matrix of exogenous link attributes and $\vx(\hat{\vtheta}, \mZ, \vt)$ the vector with the link flow functions of any link $a \in A^o$. In particular, when all attributes of the travelers' utility function are exogenous and known, such that $\vt = \bar{\vt}$, $\vx(\hat{\vtheta}, \mZ, \bar{\vt})$ becomes a vector valued function that resembles the response function $m(\vbeta,\mX)$ in \NLLS and which has the following closed form:
		
		\begin{align}
			\label{eq:NLLS-link-flow-equation}
			\vx(\hat{\vtheta}, \mZ, \bar{\vt})  
			= \mIx \vf
			&= \mIx \left((\mIqT \vq) \circ  \vpf(\hat{\vtheta}, Z, \bar{\vt})  \right)
		\end{align}

		where $\vpf(\vtheta, \mZ, \bar{\vt})$ is the vector of path choice probabilities:

		\begin{align}
			\label{eq:NLLS-path-choice-probabilities-equation}
			\vpf(\hat{\vtheta}, \mZ, \bar{\vt}) &= \exp\left(\displaystyle \mIxT \vv_{x}(\hat{\vtheta}, \mZ, \bar{\vt})     \right) \oslash \left(\displaystyle   \mIqT \mIq \exp(\mIxT \vv_{x}(\hat{\vtheta}, \mZ, \bar{\vt})  )\right)
		\end{align}
		
		Note that $\oslash$ is an operator for element wise division and $\vv_{x}(\vtheta, \mZ, \bar{\vt})$ is the vector of link utilities:
		
		\begin{equation}
			\label{eq:NLLS-link-utilities}
			\vv_x (\hat{\vtheta}, \mZ, \bar{\vt})
			= \bar{\vt} \hat{\theta}_t  +   \mZ \hat{\vtheta}_Z
			= 
			\begin{bmatrix}
				\bar{\vt} & \mZ
			\end{bmatrix}
			\begin{bmatrix}
				\hat{\theta}_t
				\\
				\hat{\vtheta}_Z
			\end{bmatrix}
			=
			\begin{bmatrix}
				\bar{\vt} & \mZ
			\end{bmatrix}
			\hat{\vtheta}
		\end{equation}
		
		\medskip

		Finally, if $\bar{\vx} \in \sR^{|A^{o}|}$ is the vector of observed traffic counts $\bar{x}$, the \NLLS estimator $\hat{\vtheta}^{\star}$ of $\vtheta \in \sR^{|D|}$ can be obtained as follows:
		
		\begin{equation}
			\label{eq:NLLS-argmin-learning-problem}
			\hat{\vtheta}^{\star} = \arg \min_{\hat{\vtheta}} \|\vx(\hat{\vtheta}, \mZ, \bar{\vt}) -\bar{\vx}\|_{2} 
		\end{equation}	
	
	\subsubsection{Identifiability of utility function coefficients}
	\label{sssec:identifiability-illustrative-example}
	
	The use of the \MNL model in our problem imposes rules of identifiability that are similar to those used in discrete choice models. For the computation of the \logit choice probabilities, differences in utility among the alternatives is what matters \citep{Walker2002a}. Therefore, if an attribute of the paths' utility function is the same within each path set, then the coefficient weighting that attribute will be, by construction, not identifiable. From a behavioral point of view, travelers would perceive no difference in utilities when facing a path choice decision and thus, their choices will be no informative about the strength of their preference for that attribute. From an optimization standpoint, there will be no gradient respect to the attribute's coefficient because the choice probabilities are invariant to changes in the value of that coefficient. 
	
	A second rule of identifiability refers to the inclusion of alternative specific constants, which may lead to an over-specification of the utility function. The practical implication in our problem is that the specific constant in one of the paths connecting each O-D pair should be fixed, e.g. setting a constant to zero. A third rule of identifiability is related to the minimum amount of traffic counts measurements required to estimate a certain number of coefficients in the utility function, which can be referred as empirical identification \citep{Walker2002a}. Similar to the rank condition in \OLS, the number of observations should higher or equal than the number of parameters but, in practice, a larger amount of data may be needed due to the existence of collinearity between observations.

	To show the relevance of the identification rules, let's consider the network in Figure \ref{fig:illustrative-example} and suppose that there are link flow measurements available for the two links. Assume the network is uncongested, $t_1(x_1) \approx t_1^0, \ t_2(x_2) \approx t_2^0$, and that the observed link counts perfectly matched the true link flow solution at equilibria, i.e. $(\bar{x}_1,\bar{x}_2)  = (x^{\star}_1,x^{\star}_2), \ \bar{x}_1 + \bar{x}_2 =\bar{q} $.
	From Eq. \ref{eq:link-flow-solution-illustrative-example}, Section \ref{sec:illustrative-example}, we note that the link flow solution of the inner level problem can be written in closed form as: 
	
	\begin{equation}
		\label{eq:sue-solution-small-uncongested-network}
		x^{\star}_1 = \frac{\bar{q}}{1+\exp{(\theta_t(t^0_1-t^0_2) + \theta_c(c_1-c_2))}} \ , \  x^{\star}_2 = \frac{\bar{q}}{1+\exp{(\theta_t(t^0_2-t^0_1) + \theta_c(c_2-c_1))}}
	\end{equation}
	
	which satisfies both the conservation constraints between travel demand and path flows and the non-negativity constraints of link flows (Eq. \ref{eq:multi-attribute-sue-logit}, Section \ref{sec:illustrative-example}) for $\forall \theta_t \in \sR, \forall\theta_c \in \sR$. From the two link count measurements, we can derive a system of two equations which are dependent on the two unknowns coefficients of the utility function:
	\begin{align}
		\label{eq:inverse-sue-solution-uncongested-small-network}
		\theta_t(t_1(\bar{x}_1)-t_2(\bar{x}_2)) + \theta_c(c_1-c_2) &= \ln\left(\frac{\bar{q}}{\bar{x}_1}-1\right) = \ln\left(\frac{\bar{x}_1 + \bar{x}_2}{\bar{x}_1}-1\right) = (\ln\left(\bar{x}_2\right)-\ln\left(\bar{x}_1\right)) \nonumber \\
		\theta_t(t_2(\bar{x}_1)-t_1(\bar{x}_2)) + \theta_c(c_2-c_1) &= \ln\left(\frac{\bar{q}}{\bar{x}_2}-1\right) = \ln\left(\frac{\bar{x}_1 + \bar{x}_2}{\bar{x}_2}-1\right) = -(\ln\left(\bar{x}_1\right)-\ln\left(\bar{x}_2\right))
	\end{align}
	
	Note that in line with the first identification rule, Eq. \ref{eq:inverse-sue-solution-uncongested-small-network} can be solved for $\theta_c$ or $\theta_t$ if the difference of travel time or cost are zero, respectively (see Section \ref{sec:illustrative-example}). Besides, from the third rule of empirical identification, the two traffic count measurements would suffice to identify the two parameters. However, the two equations in  Eq. \ref{eq:inverse-sue-solution-uncongested-small-network} are identical, hence linearly dependent and thus, one equation is not providing additional information to identify another coefficient of the utility function.

	\begin{remark}

		The objective function of the optimization problem may be globally minimized in a point that is not attainable or where the gradient does not vanish. Suppose that the traffic counts are consistent with \UE, such that  $(\bar{x}_1,\bar{x}_2) = (0,\bar{q})$ or $(\bar{x}_1,\bar{x}_2) = (\bar{q},0)$. With utility function dependent on travel time only, these link flow measurements are consistent with the limiting cases where $\theta_t \to -\infty$ or $\theta_t \to \infty$, respectively. While the identification rules are satisfied, the data generating process does not follow \SUE-\logit. To our knowledge, testing for the \SUE-\logit assumption remains an open question in the literature. If this assumption does not hold in practice, the optimization algorithm chosen to fit the travel time coefficient is expected to improve the objective function over iterations even when the global optima is not attainable.

	\end{remark}
	
	\begin{remark}
		
		Suppose that the monetary cost and the free flow travel times are equal in the two links, i.e. $c_1 = c_2, \ t^{0}_1 = t^{0}_2$. Under \UE, by definition, the link travel times will be the same for any level level of demand $\bar{q}$. Note that the later will be true even if the free flow travel time of the links are different, provided that the demand level surpasses a certain threshold such that both links are utilized. Because the travelers are uniformly distributed among paths and the path utilities are the same, \SUE-\logit will reproduce the same equilibria. In line with the first rule of identifiability, the travel time coefficient $\theta_t$ is not identifiable but regardless, any value of $\theta_t$ will perfectly reproduce the observed link flows. An optimization library will not necessarily warn about these identification problems and it may return an arbitrary estimate for $\theta_t$. A basic sanity check is to analyze if the objective function improves over iterations. Note, however, that this identifiability issue will not affect the goodness of fit but only statistical inference.

	\end{remark}

	\subsubsection{Solution methods}
	
	Second order optimization methods remain as the gold standard to estimate the travelers' utility function from traffic counts in the \LUE and \ODLUE literature. One of the earliest examples in the \ODLUE literature is found in \cite{Liu1996}, who use the secant method, a Quasi-Newton method for unidimensional optimization, to estimate the coefficients of a utility function dependent on travel time only. More recent examples are \citet{Yang2001} and \citet{Wang2016} who solved the outer level problem of the \ODLUE via Sequential Quadratic Programming (SQP), a variant of the Newton method for constrained optimization. Applications of second order methods are also found in travel behavior research where Broyden–Fletcher–Goldfarb–Shannon (\BFGS) remains as a preferred optimizer to find the maximum likelihood estimates in discrete choice models \citep{Train2002a}.
	
	Second order optimization methods can achieve a superquadratic or superlinear convergence in convex problems. They are appealing when datasets are of moderate size and thus, when the computational cost for the calculation or approximation of the Hessian matrix is reasonable. However, their convergence guarantees heavily relies on how well the curvature of the objective function informs about the direction of its steepest descent/ascent. Besides, in non-convex problems, their performance strongly depends on how close the initial points for optimization are to a local optima and on whether these points are located within a locally convex region of the objective function. Furthermore, the frequent sign changes of the curvature in non-convex functions may cause that second order optimization methods does not converge to a local minima or that they get stuck at saddle points or flat regions of the optimization landscape. 
	
	First order methods have become a gold standard for non-convex optimization due in part to the huge computational gains that are attained when avoiding Hessian matrix computation. The machine learning community has made continued efforts to develop multiple variants of these methods that are able to accelerate the optimization and without compromising computational costs and that can perform reasonably well in highly non-convex optimization landscapes \citep{Staib2019a}. Some optimizers that remains in state-of-the-art applications include stochastic gradient descent (\SGD) \citep{Robbins1951a} and the Adagrad \citep{Duchi2011} and Adam \citep{Kingma2015} optimizers. 

	In the context of quasiconvex optimization problems, \cite{HazanLevy2015} proposed a modified version of gradient descent (\GD) in which the gradient direction is normalized by its norm.  The existence of flat regions in quasiconvex problems cause that the gradients become small despite that they may be pointing to the right direction of improvement. In this setting, the normalization of the gradient is helpful to keep improving the objective function in flat regions as well as to avoid gradient explosion in sharp regions of the feasible space. The application of normalized gradient descent (\NGD) in unidimensional unconstrained optimization problems generates parameter updates according to the sign of the first derivative of the objective function. 
	
	Interestingly, the \ODLUE and \LUE literature has not explored the use of first order optimization methods to estimate the utility function coefficients using system level data. Further sections of this paper will show some theoretical and empirical results to support the choice of \NGD over the standard second order optimization methods in previous work. 
	
	\subsection{Bilevel optimization}
	\label{ssec:bilevel-optimization}
	
	\subsubsection{Problem formulation}
	
	The sole application of the \NLLS method to estimate the travelers' utility function coefficients from network level data requires that travel times are known and exogenous. In uncongested networks, this assumption may be enforced by setting the travel times to be equal to the free flow travel times. In congested networks, this assumption is not plausible and hence, it is necessary to solve both a \NLLS and a \SUE-\logit problems of a bilevel formulation in an iterative fashion. Below is the standard formulation of the problem solved in the \LUE literature and that we extend to account for travelers' utility functions with multiple coefficients:
	
	\begin{equation}
		\label{eq:bilevel-formulation-learning-problem}
		\begin{aligned}
			\min \limits_{\hat{\vtheta}} \  & \ell(\hat{\vtheta}) = \| \vx(\hat{\vtheta}) - \bar{\vx}\|_2^2 \\
			\st & \ \vx(\hat{\vtheta}) \in \arg\max_{\vx} g(\hat{\vtheta}, \vx,\vf) \\
			&
			\begin{aligned}[t]
				\max_{\vx,\vf} \ & g(\hat{\vtheta}, \vx,\vf) 
				= \displaystyle \   -\left\langle \vf, \ln \vf\right\rangle +  \sum_{a \in A} x_a \sum_{k \in K_{\mZ}} \theta_k \cdot Z_{ak} +   \sum_{a \in A}\int_{0}^{x_a}  \theta_t t_a(u) \ du \\
				\st & \ \mIq\vf = \vq \\
				& \ \mIx\vf  = \vx\\
				& \ \vx,\vf \geq \vzero  
			\end{aligned}\\
		\end{aligned}
	\end{equation}
	
	where $\ell(\hat{\vtheta})$ and $g(\hat{\vtheta}, \mZ, \vx,\vf)$ are the objective functions at the upper and lower level and $\hat{\vtheta} \in \mathbb{R}^{|D|}$ is the vector of estimated utility function coefficients. The idea behind expressing the objective functions in terms of $\hat{\vtheta}$ can be understood as follows. At the inner level, any value of $\hat{\vtheta}$ will lead to a different assignment of path/link flows in the network, hence, the objective function $g(\hat{\vtheta}, \vx,\vf)$ depends on $\hat{\vtheta}$. Note also that, thanks to the \SUE-\logit property (Eq. \ref{eq:sue-logit-path-flows-solution}), $\ell(\hat{\vtheta})$ can be approximated as an analytical function of $\hat{\vtheta}$.

	\subsubsection{Solution methods}
	
	A brute-force strategy to solve the bilevel formulation in Problem \ref{eq:bilevel-formulation-learning-problem} would consist of performing a grid search on all feasible values of $\hat{\vtheta}$. Therefore, for each value of $\hat{\vtheta}$, \SUE-\logit would be first computed at the inner level and then the objective function at the upper level would be evaluated using the closed form expression derived in Eq. \ref{eq:NLLS-link-flow-equation}, Section \ref{ssec:sue-logit-multi-attribute-utility-function}. Lastly, a candidate solution would be the value of $\theta$ that minimizes $\ell(\hat{\vtheta})$. Note however that this solution strategy becomes intractable in a higher dimensional space resulting from a multi-attribute travelers' utility function or from the joint estimation of the utility function coefficients and an O-D matrix. 
	
	A standard heuristic employed in the \ODE and \ODLUE literature to solve bilevel formulations is an alternating optimization of the inner and outer level problem. Over iterations, this heuristic is expected to converge toward a stationary point that, at the outer level, minimizes the gap between observed and predicting flows and, at the inner level, induces a link/path flow solution that satisfies traffic equilibria. The solution strategy in our problem boils down to performing gradient-based updates of the utility function coefficients at the outer level and to compute \SUE-\logit at the inner level using the coefficients $\hat{\vtheta}$ previously updated at the outer level. The resulting travel times from the \SUE-\logit computation are then feed to the outer level and the process is repeated until some criterion is convergence is met. There are many variants of solution methods in the literature but most can be summarized with the following steps: 
	
	
	\begin{enumerate}
		\item Initialization of estimated parameters (e.g. utility function coefficients or O-D matrix) 
		\item Inner level problem: solve \SUE-\logit and update travel times in the network
		\item Outer level problem: solve \NLLS with new travel times and update estimated parameters
		\item Repeat 2 and 3 until fulfilling some convergence criteria
	\end{enumerate} 
	
	In our application and similar to the \LUE literature, the utility function coefficients in $\hat{\vtheta}$ are the only free parameters.

	\begin{remark}
		A key advantage of \SUE-\logit in network modeling applications is the existence of a closed form and unique solution at path and link flow spaces. This property can significantly ease the solution of the typical bilevel optimization structure that arises in \ODLUE problems. During the alternating optimization, the closed form solution is valid but only for the set of travel times obtained in the inner level problem. Thus, it is convenient to perform small updates of $\hat{\vtheta}$ at the outer level problem such that the inner level solution remains valid.  
	
	\end{remark}

	\section{Mathematical properties of the optimization problem}
	\label{sec:mathematical-properties-learning-problem}
	
	This section studies the mathematical properties of our optimization problem. Note that our problem falls into the \LUE class and since \LUE is a subclass of \ODLUE, the analysis of mathematical properties of our problem is also relevant for the \ODLUE problem.

	\subsection{Exogenous case}
	\label{ssec:exogenous-utility-attributes}
	
	
	To ease the analysis of the mathematical properties, it is convenient to start assuming that the utility function is dependent on exogenous attributes only. This assumption significantly facilitates the solution of the bi-level formulation because the inner level problem  at each iteration consists of a single pass of \SNL. This assumption is plausible in an uncongested network where travel times are close to the free flow travel times or when there is access to travel time measurements for every link in the network. Alternatively, this assumption may be enforced by either setting the travel time coefficient of the utility function or the parameter $\alpha$ of the cost performance to zero but this is arguably less realistic.
		
		\subsubsection{Non-convexity} 
		\label{sssec:non-convexity}

		The non-convexity of the \ODLUE and \LUE problems respect to the travelers' utility function coefficients has been discussed in prior literature but via counterexamples only \citep{Yang2001}. To our knowledge, no formal proof of the non-convexity has been given for neither \LUE and \ODLUE. A key observation to prove the non-convexity of these problems is that their objective function $f$ is upper-bounded respect to the utility function coefficients. In a one-dimensional case where $\theta \in \mathbb{R}$, the upper-bounds of $f$ corresponds to the solutions of \UE when $\theta \to \infty^+$ and  $\theta \to \infty^-$. To extend this intuition to a multi-dimensional case where $\vtheta \in \mathbb{R}^{|D|}$, it is useful to introduce a general definition of convexity \citep{Boyd2004}.

		\begin{definition}[convexity]
			\label{def:convexity-definition}
			A function $f$ is convex iff $\forall y_1, y_2 \in S, \lambda \in [0,1]$:
			$$
			f(\lambda \vy_1 + (1-\lambda)\vy_2) \leq \lambda f(\vy_1) + (1-\lambda)f(\vy_2)
			$$
		\end{definition}

		\begin{prop}[bounds of objective function]
			\label{prop:upper-lower-bounds-uncongested-network}
			The objective function of the \LUE problem is lower-bounded and upper-bounded 
		\end{prop}
		
		\begin{proof}
		
			The objective function $\ell: \mathbb{R}^{|D|} \to \mathbb{R}$ of the \LUE problem can be expressed as:
			$$
			\ell(\vtheta) = \|\vx(\theta) -\bar{x}\|_2^2
			= \sum_{i \in N} \left(x_i(\vtheta)- \bar{x}_i\right)^2
			$$
			where $N$ is the set of observed traffic count measurements. Each traffic function $x_i, \forall i \in N$ can be bounded as: 
			\begin{align}
				\label{eq:bound-traffic-count-function}
				0 &\leq x_i(\theta) \leq \sum_{w \in W} q_w = Q
			\end{align}
			where $Q \in \sR_+$ is the total demand, namely, the sum of all cells in the O-D matrix $\mQ \in \sR^{V\times V}$. Then, lower and upper bounds of the objective function can be found as follows:
			\begin{align}
				\label{eq:bound-loss-function}
				0 &\leq x^2_i(\theta) \leq Q^2 \nonumber \\
				-2\bar{x}_i +\bar{x}_i^2 &\leq x^2_i(\theta) -2\bar{x}_i +\bar{x}_i^2\leq Q^2 -2\bar{x}_i+\bar{x}_i^2 \nonumber  \\
				\bar{x}_i(\bar{x}_i-2)&\leq (x_i(\theta) -\bar{x}_i)^2\leq (Q -\bar{x}_i)^2 \nonumber  \\
				\sum_{i \in N} \bar{x}_i(\bar{x}_i-2)  &\leq \sum_{i \in N} (x_i(\theta) -\bar{x}_i)^2 \leq \sum_{i \in N} (Q -\bar{x}_i)^2
				\nonumber \\
				(\bar{\vx}^\top-2\cdot\vone^\top)\bar{\vx}  &\leq \|\vx(\theta) -\bar{\vx}\|_2^2 \leq \|\vone^\top \vq -\bar{\vx}\|_2^2
			\end{align}
			
			which completes the proof
		\end{proof}
		
		Let's now prove the non-convexity of the \LUE problem with the following proposition:

		\begin{prop}[Non-convexity of \LUE problem]
			\label{prop:non-convexity}
			Suppose the coefficients of the travelers' utility function in the \LUE problem are identifiable. Then, the \LUE problem is not convex if its objective function is upper-bounded
		\end{prop}
		
		\begin{proof}
			Let's prove this by contradiction. Assume that the \LUE objective function  $\ell(\vtheta) = \|\vx(\vtheta) -\bar{\vx}\|_2^2, \ \forall \vtheta \in \mathbb{R}^{|D|}$ is convex. Then, let's use Definition \ref{def:convexity-definition} of convexity at points $\vy_1 = \frac{\vtheta_1 - (1-\lambda)\vtheta_2}{\lambda}$ and $\vy_2 = \vtheta_2$, with $\vy_1, \vy_2,\vtheta_1,\vtheta_2 \in \mathbb{R}^{|D|}$ and $\lambda \in ]0,1]$:
			\begin{align*}
				\ell(\lambda \vy_1 + (1-\lambda)\vy_2) &\leq \lambda \ell(\vy_1) + (1-\lambda)\ell(\vy_2)\\
				\ell\left(\lambda \left(\frac{\vtheta_1 - (1-\lambda)\vtheta_2}{\lambda}\right) + (1-\lambda)\vtheta_2\right) &\leq \lambda \ell\left(\frac{\vtheta_1 - (1-\lambda)\vtheta_2}{\lambda}\right) + (1-\lambda)\ell(\vtheta_2)\\
				\ell(\vtheta_1)  &\leq \lambda \ell\left(\frac{\vtheta_1 - (1-\lambda)\vtheta_2}{\lambda}\right) + (1-\lambda)\ell(\vtheta_2) 
			\end{align*}
			
			which implies that: 
			
			\begin{equation}
				\label{eq:non-convexity-proof1}
				\ell\left(\frac{\vtheta_1 - (1-\lambda)\vtheta_2}{\lambda}\right) \geq \frac{\ell(\vtheta_1)-(1-\lambda)\ell(\vtheta_2)}{\lambda}= \frac{\ell(\vtheta_1)-\ell(\vtheta_2)}{\lambda} + \ell(\vtheta_2)
			\end{equation}
			
			If $\vtheta$ is identifiable, $f$ is not a constant function and $\exists \lambda \in [0,1]: \ell(\vtheta_1) > \ell(\vtheta_2)$. From Eq. \ref{eq:non-convexity-proof1}, when $\ell(\vtheta_1) > \ell(\vtheta_2)$, the RHS grows with no bound as $\lambda \to 0^{+}$. Hence, the function $f$ is not upper bounded, which generates a contradiction and it completes the proof.

		\end{proof}
		
		Convexity  is not a necessary but a sufficient condition for global optimality. Therefore, before ruling out the global optimality of our problem, it is convenient to study a generalization of convexity known as a pseudo-convexity that imposes weaker conditions on the optimization problem and that provides sufficient conditions for global optimality.

		\subsubsection{Pseudo-convexity}
		\label{ssec:pseudo-convexity-small-network}

		Pseudo-convex functions are a subclass of the quasi-convex functions. Hence, the properties of pseudo-convex functions are also valid for quasi-convex functions, while the converse is not true \citep{Mangasarian1965}. Pseudo-convexity is a generalization of convexity that, as quasi-convexity, extends the notion of unimodality to higher dimensional spaces. In contrast to quasi-convex functions but similar to convex functions, the first order necessary optimality condition of pseudo-convex functions suffices for global optimality \citep{Crouzeix1982}. This property will be key later to prove the existence of a global minimizer in our optimization problem. A formal definition of pseudo-convexity is the following \citep{Bazaraa2006}:
		
		\begin{definition}[Pseudo-convexity]
			\label{def:pseudo-convexity}
			$f$ is a \textit{pseudo-convex} function on the feasible set $S$ iff $\forall x_1,x_2 \in S$: 
			
			\begin{equation}
				f(x_2) < f(x_1) \implies \nabla f(x_1)^\top(x_2-x_1) < 0
				\label{eq:pseudo-convexity-definition-1a}
			\end{equation}
			
			or equivalently, 
			
			\begin{equation}
				\label{eq:pseudo-convexity-definition-1b}
				\nabla f(x_1)^\top(x_2-x_1) \geq 0 \implies  f(x_2) \geq f(x_1)  
			\end{equation}
			
		\end{definition}

		The definition of pseudo-convexity can be assessed via the augmented hessian of the objective function \citep{Bazaraa2006}. Formally, $f$ is pseudoconvex if there exists a scalar $v, \  0 \leq v < \infty$ such that its augmented Hessian $H(x) + v \nabla f(x) \nabla f(x)^T$ is positive semidefinite for all $x \in X$ \citep{Crouzeix1982}. This proposition can be also evaluated in terms of the signs of the principal minors of the bordered Hessian \citep{Mereau1974}. 
		
		For the remaining analyses of mathematical properties, it will be also helpful to introduce Definition \ref{def:coordinate-wise-pseudo-convexity} of coordinate-wise pseudo-convexity. A function $f$ is coordinate-wise pseudo-convex if, when all coefficients of the utility function except for one coefficient are fixed, it is pseudo-convex respect to the non-constant coefficient. Note that if the function is dependent on a single attribute, the definitions of coordinate-wise pseudo-convexity and pseudo-convexity are equivalent. Formally:

		\begin{definition}[Coordinate wise pseudo-convexity]
			\label{def:coordinate-wise-pseudo-convexity}
			
			$\ell: \sR^{|D|} \to \sR$ is a coordinate-wise \textit{pseudo-convex} function iff $\forall i \in D, \forall \vtheta^1,\vtheta^2 \in S_i$, such that $S_i = \{\vtheta^1,\vtheta^2 \in \mathbb{R}^{|D|}\ | \ \theta_j^1 = \theta_j^2, \forall j \neq i \}$:
			\begin{equation}
				\ell(\vtheta^2) < \ell(\vtheta^1) \implies \dfrac{\partial \ell(\vtheta^1)}{\partial \theta_i}(\theta_i^2-\theta_i^1) < 0
				\label{eq:coordinate-wise-pseudo-convexity-definition-1a}
			\end{equation}
			
			or equivalently, 
			
			\begin{equation}
				\label{eq:coordinate-wise-pseudo-convexity-definition-1b}
				\dfrac{\partial \ell(\vtheta^1)}{\partial \theta_i}(\theta_i^2-\theta_i^1) \geq 0 \implies  \ell(\vtheta^2) \geq \ell(\vtheta^1)  
			\end{equation}

		\end{definition}

		Proposition \ref{prop:pseudoconvexity-uncongested-network}, \ref{appendix:ssec:coordinate-wise-properties} shows that coordinate-wise monotonicity of the traffic count (response) functions implies coordinate-wise pseudo-convexity of the objective function of the \LUE problem under exogenous travel times. The definition of coordinate-wise monotonicity is the following:

		\begin{definition}[Coordinate-wise monotonicity of response functions]
			\label{assumption:monotonocity-traffic-count-functions}
			The traffic count (response) functions are monotonic respect to each coefficient of the utility function if, when all coefficients of the utility function except for one coefficient are kept constant, the response functions are monotonic respect to the non-constant coefficient.
		\end{definition}

		Propositions \ref{prop:monotonocity-softmax} and \ref{prop:monotonicity-softmax-sum}, \ref{appendix:ssec:monotonicity-traffic-count-functions} illustrate two cases where under mild assumption the coordinate-wise monotonicity of the traffic flow functions holds in a general transportation network. While it is easy to create counter examples where the coordinate-wise monotonicity of the traffic flow functions is violated, analysis on synthetic data suggests that the assumption holds in practice (see Section \ref{sssec:monotonicity-pseudoconvexity-small-networks}).

		\subsubsection{Existence and uniqueness of global minima}
		
		In absence of measurement error in the traffic counts, the proofs for existence and uniqueness of a global optimality are direct (Propositions \ref{prop:existence-uncongested-network-no-error} and \ref{prop:uniqueness-uncongested-network-no-error}, \ref{ssec:existence-local-minima-small-network}). While the absence of measurement error can be enforced in experiments with synthetic data, it is implausible in real world problems where traffic count data is expected to be noisy. Fortunately, the pseudo-convexity of the objective function in our problem can be leveraged to extend the proofs of existence and uniqueness in presence of measurement error. A key property of pseudo-convex functions that also holds in convex functions is that every local minima is also global minima \citep{Mangasarian1965}. This relaxes the positive (semi)definite assumption of the Hessian of the objective function used to prove (strict) global optimality in convex problems. Following the proof in \citet{Bazaraa2006} for multi-variable pseudo-convex functions, Proposition \ref{prop:global-optimality-pseudo-convexity-uncongested-network} shows that this property can be also applied to prove the existence of global minima in our problem. 
		
		\begin{prop}[Sufficient condition for the existence of a global minima]
			\label{prop:global-optimality-pseudo-convexity-uncongested-network}
			Assume the objective function $f$ of the \LUE problem under exogenous travel times is coordinate-wise pseudo-convex. Then, if the gradient of $f$ vanishes at $\vtheta^{\star} \in \mathbb{R}^{|D|}$, $\vtheta^{\star}$ is a global minima.
		\end{prop}

		\begin{proof}
			By assumption, $\grad_{\vtheta} \ell (\vtheta^{\star}) = 0$, which implies that $\dfrac{\partial \ell(\vtheta^{\star})}{\partial \theta_i} = 0, \forall i \in D$. By applying Definition \ref{def:coordinate-wise-pseudo-convexity} of coordinate-wise pseudo-convexity to each coordinate $i \in D$, $\ell(\vtheta^2) \geq \ell(\vtheta^{\star}), \ \forall \vtheta^2 \in \mathbb{R}^{|D|}$, which completes the proof. 
		\end{proof}
		
		\begin{remark}
			
			Proposition \ref{prop:existence-local-optima-uncongested-network-with-error}, \ref{appendix:ssec:coordinate-wise-properties} gives a sufficient condition for the existing of a vanishing gradient (derivative) in the unidimensional case. While this seems restrictive, it addresses the standard setting studied in prior literature with a utility function dependent on travel time only. 

		\end{remark}
		
		Proposition \ref{prop:uniqueness-uncongested-network} proves the uniqueness of the global optima in our problem by leveraging the the strict quasi-convexity of differentiable pseudo-convex functions proven by \citet{Mangasarian1965}. A formal definition of strict quasiconvexity is the following \citep{Bazaraa2006}:
	
		\begin{definition}[Strict quasi-convexity]
			\label{def:strict-quasi-convexity}
			f is said to be strictly quasi-convex if for every $x_1, x_2 \in S, x_1 \neq x_2, \lambda \in [0,1]$:
			$$
			f(x_1) < f(x_2) \implies f(\lambda x_1 + (1-\lambda_1)x_2) < f(x_1)
			$$
			or equivalently:
			$$
			f(\lambda x_1 + (1-\lambda) x_2)< \max\left(f(x_1), f(x_2)\right)
			$$
		\end{definition}
		
		Analogously to Definition \ref{def:coordinate-wise-pseudo-convexity}, the definition of strict quasi-convexity can be also applied coordinate-wise:

		\begin{definition}[Coordinate-wise strict quasi-convexity]
			\label{def:coordinate-wise-strict-quasi-convexity}
			$\ell: \sR^{|D|} \to \sR$ is said to be coordinate-wise strictly quasi-convex iff $\forall i \in D, \forall \vtheta^1,\vtheta^2 \in S_i$, such that $S_i = \{\vtheta^1,\vtheta^2 \in \mathbb{R}^{|D|}\ | \ \theta_j^1 = \theta_j^2, \forall j \neq i \}$:
			$$
			\ell(\vtheta_1) < \ell(\vtheta_2) \implies \ell(\lambda \vtheta_1 + (1-\lambda_1)\vtheta_2) < \ell(\vtheta_1)
			$$
			or equivalently:
			$$
			\ell(\lambda \vtheta_1 + (1-\lambda) \vtheta_2)< \max\left(\ell(\vtheta_1), \ell(\vtheta_2)\right)
			$$
		\end{definition}

		\begin{prop}[Sufficient conditions for uniqueness of the global minima]
			\label{prop:uniqueness-uncongested-network}
			If the objective function of \LUE problem is (coordinate-wise) pseudo-convex and a global minima exist, then the global minima is unique.
		\end{prop}

		\begin{proof}
			
			
			Let's prove this by contradiction. Let's be $\ell: \mathbb{R}^{|D|} \to \mathbb{R}$ the objective function of the problem and $\vtheta \in \mathbb{R}^{|D|}$ the vector of coefficients of the utility function. Suppose that there are two global minima $\vtheta^{\star}_1,\vtheta^{\star}_2 \in \mathbb{R}^{|D|}: \ell(\vtheta^1) = \ell(\vtheta^{\star}_2) = \ell^{\star}$ and $\theta^{\star}_1 \neq \theta^{\star}_2$. By the multivariate mean value theorem (MMVT), for $\vtheta^{\star}_3 \in \mathbb{R}^{|D|}$ , $\exists \theta^{\star}_3 \in (1-\lambda)\vtheta^{\star}_1 + \lambda \vtheta^{\star}_2 : \grad \ell (\vtheta = \vtheta^{\star}_3) = 0, \ \lambda \in [0,1]$ and by Proposition \ref{prop:global-optimality-pseudo-convexity-uncongested-network}, $\ell(\theta^{\star}_3) = f^{\star}$. By using MMVT recursively, $\forall \vtheta^{\star}_i \in [\theta^{\star}_1,\theta^{\star}_2], \ \ell(\theta^{\star}_i) = \ell^{\star}, \ \theta^{\star}_i = (1-\lambda)\vtheta^{\star}_1 + \lambda \vtheta^{\star}_2, \ \lambda \in [0,1]$ and hence, $ [\vtheta^{\star}_1,\vtheta^{\star}_2]$ forms a convex set. By Property 2, \citet{Mangasarian1965}, if $\ell$ is pseudo-convex, $\ell$ is also strictly quasi-convex. Similarly, if $\ell$ is coordinate-wise pseudo-convex, $\ell$ is also coordinate-wise strictly quasi-convex. By applying Definition \ref{def:coordinate-wise-strict-quasi-convexity} of strict quasi-convexity on $\vtheta^1, \vtheta^2 \in \mathbb{R}^{|D|}$ gives $\ell^{\star} = \ell(\lambda \theta^1 + (1-\lambda) \theta^2)< \max\left(\ell(\theta^1), \ell(\theta^2)\right) = \ell^{\star}$, which leads to a contradiction and it completes the proof.  
		
		\end{proof}

		\subsection{Endogenous case}

		In a congested network, travel times are endogenous and this requires to solve a bilevel formulation of the \LUE problem (see Section \ref{ssec:bilevel-optimization}). Note that even if the inner and outer level problems of the bilevel formulation were convex, the analysis of global optimality is complex \citep{Dempe2016}. Because convex functions are a subclass of the pseudo-convex functions, the pseudo-convexity of the outer problem is expected to make the analysis even harder. Despite the above difficulties, this section proves two properties of the inner level and outer level problems of the \LUE bilevel formulation that can improve the convergence of an alternating optimization algorithm.

		\subsubsection{Inner level problem}
		
		The \SUE-\logit problem with a utility function dependent on travel time only is known to be strictly convex in path (and link) flow space under link performance functions that are monotonically increasing and dependent on the traffic flow on the links only. To our knowledge, Proposition \ref{prop:uniqueness-sue-logit-multi-attribute utility function}
		extends, for the first time, the aforementioned property to multi-attribute utility functions.

		\begin{prop}[Uniqueness of \SUE-\logit path flow solution with a multi-attribute utility function]
			\label{prop:uniqueness-sue-logit-multi-attribute utility function}
			Suppose that the travelers' utility function is linear-in-parameters and that the link cost functions are monotonically increasing and dependent on the traffic flow on the link only. Then, the \SUE-\logit problem has a unique solution in path flow space.
		\end{prop}
		
		\begin{proof}
			
			Let's transform the objective function of Problem \ref{eq:multi-attribute-sue-logit}, Section \ref{ssec:sue-logit-multi-attribute-utility-function} into a minimization problem with the following argument:
			\begin{align*}
				g(\vx,\vf, \vtheta) 
				&= -\sum_{a \in A}\int_{0}^{x_a} v_a(u,\vtheta)du - \left\langle \vf, \ln \vf\right\rangle
				= -\sum_{a \in A}\int_{0}^{x_a} \left(\theta_t t_a(u) - \sum_{k \in K_{\mZ}} \theta_k \cdot Z_{ak} \right)du + \left\langle \vf, \ln \vf\right\rangle\\
				&= g_1(\vx, \vtheta) +g_2(\vf)
			\end{align*}
			
			From Theorem 1 in \citet{Evans1973}, it follows that $g_2(\vf) = -\left\langle \vf, \ln \vf\right\rangle$ is strictly convex on the feasible set $S$ defined by the constraint in Problem \ref{eq:multi-attribute-sue-logit}, Section \ref{ssec:sue-logit-multi-attribute-utility-function}. Note that the matrix $\mZ \in \mathbb{R}^{A \times k_{\mZ}}$ with link attributes values is exogenous and thus, $g_1(\vx, \vtheta)$ can be rewritten as:
			$$
			g_1(\vx, \vtheta) = -\theta_t \sum_{a \in A}\int_{0}^{x_a} t_a(u)  du -  \sum_{a \in A} x_a \left(\sum_{k \in K_{\mZ}} \theta_k \cdot Z_{ak}\right)
			$$
			To analyze the convexity of the first term $g_1(\vx, \vtheta)$, it suffices to analyze its second derivative respect to any path flow variable $f_h$:
			\begin{align*}
				\dfrac{\partial g_1(\vx,\vf, \vtheta)}{\partial  f_h} 
				&= -\theta_t \dfrac{\partial }{\partial  f_h} \left(\sum_{a \in A}\int_{0}^{x_a} t_a(u)  du\right)-  \sum_{a \in A} \dfrac{\partial x_a}{\partial  f_h} \left(\sum_{k \in K_{\mZ}} \theta_k \cdot Z_{ak}\right)
			\end{align*}
			
			Then, replacing the link flow variables by using the balance constraint for path and link flows in the second term of $g_1$:
			
			$$
			\dfrac{\partial^2 g_1}{\partial^2 f_h}
			= -\theta_t \dfrac{\partial^2 }{\partial^2  f_h} \left(\sum_{a \in A}\int_{0}^{x_a} t_a(u)  du\right)-  \sum_{a \in A} \cancelto{0}{\dfrac{\partial }{\partial  f_h}\left(\sum_{k \in K_{rs}} q_{rs} \delta_{ak}\right)}\left(\sum_{k \in K_{\mZ}} \theta_k \cdot Z_{ak}\right)
			= -\theta_t \dfrac{\partial^2 }{\partial^2  f_h} \left(\sum_{a \in A}\int_{0}^{x_a} t_a(u)  du\right)
			$$
			
			Given the class of link cost functions used in this proposition (Assumption \ref{assumption:link-performance-functions}, Section \ref{ssec:assumptions}) and the fact that $\theta_t < 0$, it follows that $g_1$ is convex on the path flows variables $f_h$ \citep{Fisk1980}. As a consequence, $g(\vx,\vf, \vtheta)$ is a sum of a convex and strictly convex function, and thus, it is strictly convex. Also, all constraints in Problem \ref{eq:multi-attribute-sue-logit}, Section \ref{ssec:sue-logit-multi-attribute-utility-function} are affine and thus they define a convex set. Thus, the solution of problem \ref{eq:multi-attribute-sue-logit} has a unique global minima, which completes the proof. 
			
		\end{proof}

		\subsubsection{Outer level problem}
		
		A pseudo-convex function can be convex and concave on its feasible space. Despite of the above, its gradient always point out to the global descent direction \citep{HazanLevy2015}. Proposition \ref{prop:gradients-direction-small-uncongested-network} shows that this property also holds in the outer level objective of the \LUE problem. 

		\begin{prop}
			\label{prop:gradients-direction-small-uncongested-network} Assume that the objective function of the outer level of the \LUE problem is coordinate-wise pseudo-convex and that a global optimal solution exist. Then, the negative gradient of the objective function points out to the global descent direction
		\end{prop}
		
		\begin{proof}
		
			Let's be $\vtheta^{\star} \in \mathbb{R}^{|D|}$ the global optima of the \LUE problem and $\bar{\vtheta} \in \mathbb{R}^{|D|}$ a feasible point that is not a global optima. By the coordinate-wise pseudo-convexity of $\ell: \sR^{|D|} \to \sR$:

			\begin{equation}
				\label{eq:coordinate-wise-pseudoconvexity-congested-case}
				\ell(\vtheta^{\star}) < \ell(\bar{\vtheta}) \implies \dfrac{\partial \ell(\bar{\vtheta})}{\partial \theta_i}(\theta_i^{\star}-\theta_i) < 0, \quad \forall i \in D
			\end{equation}
			
			Since $\ell(\bar{\vtheta}) < \ell(\vtheta^{\star})$, $\dfrac{\partial \ell(\bar{\vtheta})}{\partial \theta_i}(\theta_i^{\star}-\bar{\theta}_i) < 0, \forall i \in D$. Then, if $\bar{\theta}_i < \theta_i^{\star}$, $\dfrac{\partial \ell(\bar{\vtheta})}{\partial \theta_i} >0$, and if $\bar{\theta}_i > \theta_i^{\star}$, $\dfrac{\partial \ell(\bar{\vtheta})}{\partial \theta_i} <0$. Therefore, if $\bar{\vtheta} - \vtheta^{\star} < 0, \  \minus\grad_{\theta} \ell(\bar{\vtheta}) < 0$ and if $\bar{\vtheta} - \vtheta^{\star}  > 0,\ \minus\grad_{\theta} \ell(\bar{\vtheta}) > 0$, which proves that the negative gradient is always pointing to the global descent direction.

		\end{proof}
		
		\begin{remark}
			In a congested network, each value of $\vtheta \in \mathbb{R}^{|D|}$ generates a different traffic assignment at the inner level of the bilevel formulation. Therefore, link travel times will change between iterations of an alternating optimization algorithm. As a consequence, the objective function of the outer level problem will be (coordinate-wise) pseudo-convex respect to the travelers' utility function coefficients only in a small neighborhood where travel times remain approximately constant. 
		\end{remark}

		\section{Solution algorithm and convergence guarantees}

		\subsection{Inner level optimization}
		
		The solution method for the inner level problem of the bilevel formulation is described in Algorithm~\ref{alg:inner-level-optimization}, \ref{appendix:ssec:implementation-inner-level-optimization}. The stochastic network loading (\SNL) is performed according to Algorithm~\ref{alg:stochastic-network-loading}, \ref{appendix:ssec:snl}. The column generation phase is adapted from \citet{Damberg1996} and it is performed once per each iteration of the alternating optimization algorithm (Step 1, Algorithm \ref{alg:inner-level-optimization}, \ref{appendix:ssec:implementation-inner-level-optimization}). Note that in the original implementation of \DSD, the column generation phase is performed at each iteration of \SUE-\logit but in the context of a bilevel optimization, this path set augmentation strategy would not select good paths if the current solution for $\vtheta$ is far from the global optima. The best convex combination of solutions at each iteration of \SUE-\logit (Step c, Algorithm \ref{alg:inner-level-optimization}, \ref{appendix:ssec:implementation-inner-level-optimization}) is searched over a uniform grid in $[0,1]$ and with an arbitrary granularity depending on the desired level of accuracy. To reduce computational cost, the size of the augmented path set is constrained to a fix quantity ($k_g$) and the path set augmentation is only performed in a proportion ($\rho_{W}$) of the O-D pairs at each iteration. The sample of O-D pairs is selected according to the level of demand. To capture the correlation between paths in the consideration set due to overlapped link segments, we correct path utilities using the path size logit (PSL) factor \citep{Ben-Akiva1999a}. The utility term associated to the PSL factor is weighted by a coefficient $\beta \in \sR$ that can be estimated \citep{Bierlaire2008}. For the scope of this paper, the coefficient $\beta$ is treated as an hyperparameter with default value equal to one. The utility term associated to the path size correction is updated before generating new paths and after selecting paths depending if the composition of the path set change when performing the steps 1 and 3 of the inner level optimization (Algorithm \ref{alg:inner-level-optimization}, \ref{appendix:ssec:implementation-inner-level-optimization}),

		\subsection{Outer level optimization}
		\label{ssec:outer-level-optimization-solution-algorithm}

		To study the convenience of using first or second order methods, the outer level objective function of the bilevel formulation is optimized with an hybrid algorithm (see Algorithm \ref{alg:outer-level-optimization}, \ref{appendix:ssec:implementation-outer-level-optimization}). In a \textit{non-refined stage}, the algorithm uses Normalized Gradient Descent (\NGD) (Algorithm \ref{alg:first-order-optimization},\ref{appendix:ssec:first-order-optimization-methods}). Subsequently, the best solution obtained from the \textit{non-refined stage} is used a starting point for the \textit{refined stage}, where second order optimization methods are used to obtain a more accurate solution (Algorithm \ref{alg:outer-level-optimization}, \ref{appendix:ssec:implementation-outer-level-optimization}). Two specialized second order methods were initially considered. The first is Gauss-Newton (\GN), one of the standard optimizer in the Econometrics literature on non-linear regression \citep{DelPino1989, Amemiya1983b,Gallant1975} and which can be seen as an unconstrained version of the sequential quadratic programming (\SQP) method. The second is the Levenberg–Marquardt (\LM) algorithm which, by interpolating between \GN and gradient descent, can increase the robustness of \GN in flat regions of the optimization landscape \citep{Marquardt1963}. Given that \GN is a particular case of the \LM method (Algorithm \ref{alg:second-order-optimization-methods}, \ref{appendix:ssec:second-order-optimization-methods}), only \LM is used in the \textit{refined stage}.

		\subsection{Bilevel optimization}
		\label{ssec:bilevel-level-optimization-solution-algorithm}
		
		Algorithm \ref{alg:bilevel-optimization}, Section \ref{ssec:bilevel-level-optimization-solution-algorithm} shows the pseudocode for the alternating optimization of the inner and outer level of our bilevel formulation (Problem \ref{eq:bilevel-formulation-learning-problem}, Section \ref{ssec:bilevel-optimization}). The algorithm returns the vector of utility function coefficients $\hat{\vtheta} \in \mathbb{R}^{|D|}$ that minimizes the discrepancy between predicted and observed link flows. To guarantee that the predicted link flows follow \SUE-\logit, the last iteration of the alternating algorithm only solves the inner level problem using the value of $\hat{\vtheta}$ obtained in the previous iteration of the outer level problem.
		
		\begin{algorithm}[!htbp]
  
			\captionsetup{font=normalsize}
			\caption{\texttt{BilevelOptimization}} 
			\label{alg:bilevel-optimization}

			\scalebox{1}{%
				\begin{minipage}{\linewidth}
					\begin{algorithmic} 
						
						\Require{Iterations $I$, initial vector of utility function coefficients $\vtheta_0$, inputs of \texttt{InnerLevelOptimization} $T, \mZ, \vt_f, \vq, \vlambda_g,\rho_W$, inputs of \texttt{OuterLevelOptimization} $T_1, T_2,\eta,\bar{\vx}, \vx, \vpf, \texttt{refined-method}, \eta_1, \eta_2, b, \delta$}
		
						\item[]
						
						\State Step 0: Initialization: 
						$$
						\vtheta:= \vtheta_0
						$$
						
						\State Step 1: Alternating optimization
						\item[]
						\For{$i = 1 \ldots I$}
						
						\State $\vx, \vpf, \bar{\vt} \gets$ \texttt{InnerLevelOptimization}($\vtheta,T, \mZ, \vt_f, \vq, \vlambda_g,\rho_W$)
						
						\item[]
						
						\If{$i < I$}
						\State $\vtheta \gets$ \texttt{OuterLevelOptimization}($\vx, \vpf,\bar{\vx},\vtheta, T_1, T_2, \texttt{refined-method}, \eta_1, \eta_2, b, \delta$)	
						\EndIf
						
						\EndFor
						
						\State
						
						\State 
						\Return {$\vtheta^{\star} = \argmin_{\{\vtheta_0,\ldots,\vtheta_{I-1} \}} \ell_i(\vtheta_i):= \|\bar{\vx}-\vx(\vtheta_i)\|^2$ }
					\end{algorithmic}
				\end{minipage}%
			}
		\end{algorithm}
		
		\renewcommand{\algorithmicrequire}{\textbf{Input:}}

		\subsection{Convergence guarantees}
		\label{ssec:convergence-guarantees-learning-algorithm}
		
		Our bilevel formulation can be casted as a mathematical program with equilibrium constraints (\MPEC). For this class of problems, the feasible region is typically non-convex \citep{Boyles2020} and this makes difficult to find theoretical guarantees even for convergence toward local optima \citep{Sabach2017,Dempe2016}. Thus, we found convenient to analyze a case where all attributes of the utility function are exogenous (Section \ref{ssec:exogenous-utility-attributes}) and where we could leverage the pseudo-convexity of the optimization problem (Section \ref{ssec:pseudo-convexity-small-network}).

		\subsubsection{Theoretical $\epsilon$-convergence under exogenous travel times}
		\label{sssec:convergence-ngd-uncongested-network}

		The analysis of mathematical properties in Section \ref{sec:mathematical-properties-learning-problem} provided sufficient conditions for the existence and uniqueness of a global optima. Now it remains to prove that the global optima is attainable when running normalized gradient descent (\NGD) on a finite number of steps. We will consider the \LUE problem with a utility function dependent of a single exogenous attribute and where only the no-refined stage of the outer level algorithm is conducted. While this seems restrictive, it addresses the standard setting studied in prior literature with a utility function dependent on travel time only.

		\citet{HazanLevy2015} proved the convergence of \NGD to a global minimum for functions that are both strictly quasi-convex and \textit{locally-Lipschitz}. These results are relevant for our problem, since pseudo-convex functions are also quasi-convex  \citep{Mangasarian1965}. \textit{Locally-Lipschitz} functions have the property of having its first derivative bounded by an arbitrary positive constant and they are required to be \textit{Lipschitz} in a small region around the optimum. Generally, the \textit{Lipschitz constant} is relevant to study the theoretical convergence speed of first optimization methods. Formally:
		
		\begin{definition}[Lipschitzness]
			\label{def:lipschitzness-uncongested-network}
			A unidimensional function $f: \mathbb{R} \to \mathbb{R}$ is Lipschitz iff $\ \forall x,y \in \mathbb{R}, \ G>0$:
			\begin{align}
				\label{eq:lipschitzness-uncongested-network}
				|f(x)-f(y)| &\leq G\|x-y\| 
			\end{align}

		\end{definition}

		\begin{prop}
			\label{prop:lipschitzness-uncongested-network}
			Suppose that the travelers' utility function in the \LUE problem depends on a single exogenous attribute. Then, the objective function is Lipschitz
		\end{prop}

		\begin{proof}

			Let's define the the objective function of the \LUE problem a $\ell: \sR^{|D|} \to \sR$. Starting from the LHS of Eq. \ref{eq:lipschitzness-uncongested-network} and using the reverse triangle inequality:

			\begin{align}
				\label{eq:proof-lipschitzness-uncongested-network-1}
				|\ell(\theta^1_t)-\ell(\theta^2_t)| \nonumber
				&= \Big|\|\vx(\theta^1_t) -\bar{x}\|_2^2 -\|\vx(\theta^2_t) -\bar{x}\|_2^2\Big|\\ \nonumber
				&\leq \|(\vx(\theta^1_t) -\bar{x}) - (\vx(\theta^2_t) -\bar{x})\|_2^2\\
				&= \|(\vx(\theta^1_t) - \vx(\theta^2_t)\|_2^2
			\end{align}
			We can now bound the range of the each traffic flow function $x_i, \forall i \in N$ and the norm of the difference between the vector of traffic functions in Eq. \ref{eq:proof-lipschitzness-uncongested-network-1} as follows:
			\begin{align}
				\label{eq:proof-lipschitzness-uncongested-network-2}
				0 \leq x_i(\theta_t) \leq Q
				\implies 
				-Q &\leq x_i(\theta^1_t)-x_i(\theta^2_t) \leq Q \nonumber \\
				0 &\leq \left(x_i(\theta^1_t)-x_i(\theta^2_t)\right)^2 \leq Q^2 \nonumber \\
				0 &\leq \|\vx(\theta^1_t)-\vx(\theta^2_t)\|_2^2 \leq N Q^2
			\end{align}
			where $Q \in \sR_+$ is the total demand, namely, the sum of all cells in the O-D matrix $\mQ \in \sR^{V\times V}$. Now define $|\theta^1_t-\theta^2_t| = \delta >0$ and set $G=NQ^2/\delta > 0$, and replace it into Eq. \ref{eq:proof-lipschitzness-uncongested-network-2}:
			\begin{align}
				0 &\leq \|\vx(\theta^1_t)-\vx(\theta^2_t)\|_2^2 \leq N Q^2 = G \delta  = G|\theta^1_t-\theta^2_t| 
			\end{align}
			
			which proves that the objective function is $G$-Lipschitz 
			
			
		\end{proof}

		Now armed with Proposition \ref{prop:lipschitzness-uncongested-network}, we can provide guarantees for the convergence of \NGD when $\theta \in \mathbb{R}$ and under exogenous travel times:
		
		\begin{prop}
			\label{prop:convergence-ngd-uncongested-network}
			Assume that the objective function of the \LUE problem is (coordinate-wise) pseudo-convex
			and that the travelers' utility function in the depends on a single exogenous attribute. If the optimization problem is solved with normalized gradient descent (\NGD), \NGD converges to an $\epsilon$-optimal solution in a finite number of iterations (poly(1/$\epsilon$))
		\end{prop}
		
		\begin{proof}
			
			By assumption, the objective function $\ell: \mathbb{R} \to \mathbb{R}$ is pseudo-convex and thus strictly quasi-convex. Under exogenous travel times, we can use Proposition \ref{prop:lipschitzness-uncongested-network} to conclude that $\ell$ is locally Lipschitz. Thus, all conditions from Theorem 4.2, \citet{HazanLevy2015} are satisfied and hence the proof is complete. 
		\end{proof}

		\section{Statistical inference on the parameters estimated from system level data}
		\label{sec:inference-parameters}
		
		
		This section derives closed form expressions of the statistical tests used to analyze the coefficients of the travelers' utility function estimated with the methodology presented in Section \ref{sec:learning-parameters}. \NLLS estimation is often preferred over \texttt{MLE} because the former relies on weaker distributional assumptions  \citep{Cameron2005a}.

		\subsection{The \NLLS estimator}
		\label{ssec:nlls-estimator}
		
		The nonlinear regression formulation of the outer level objective of our problem defines each traffic count measurement $\rx_i \in \sR, \forall i \in N$ in the true data generating process as a scalar dependent random variable with conditional mean: 
		
		\begin{equation}
			\label{eq:response-function-nlls-learning-problem}
			\E(\rx_i | \mZ, \bar{\vt}) = x_i(\mZ, \bar{\vt}, \vtheta)
		\end{equation}
		
		\medskip
		where $x_i$ is the traffic count (response) function (Eq. \ref{eq:NLLS-link-flow-equation}, Section \ref{sssec:nlls-problem-formulation}),  $\bar{\vt}$ is the set of link travel times and $\mZ$ is the matrix of exogenous attributes. Note that $\bar{\vt}$ is assumed exogenous in the \NLLS problem solved at the outer level but it is iteratively updated via the alternating optimization with the inner level problem (Section \ref{ssec:bilevel-level-optimization-solution-algorithm}). The regression equation of $\rx_i$ is defined as:
		
		\begin{equation}
			\label{eq:regression-equation-nlls-learning-problem}
			\rx_i = \E(\rx_i | \mZ, \bar{\vt}) + \ru_i
		\end{equation}
		
		\medskip
		where $\ru_i$ is the random error coming from the true data generating process. The differentiation of the objective function of the \NLLS problem (\ref{appendix:sssec:gradients}) leads to the following first order necessary optimality condition: 
		
		\begin{equation}
			\label{eq:first-order-optimality-condition-NLLS}
			\dfrac{\partial \ell(\vtheta)}{\partial \vtheta}
			= \frac{2}{n} \sum_{i \in N} \dfrac{\partial x_i(\vtheta)}{\partial \vtheta} (\rx_i-x_i(\vtheta))
			= 
			\vzero
			\implies 
			[D_{\vtheta} \ x(\vtheta)]^\top \ru = \vzero
		\end{equation}
		
		\medskip
		where $n$ is
		the number of observations and  $D_{\vtheta} \ x(\vtheta)  \in \mathbb{R}^{n \times d}$ is a Jacobian matrix corresponding to the stacked gradient vectors for each observation $n \in N$. Also,  for the ease of notation and for the remainder of the paper $\vx(\vtheta) = \vx(\mZ, \vt, \vtheta)$. Note that Eq. \ref{eq:first-order-optimality-condition-NLLS} defines a set of $d$ equations, namely, one equation for each each element of the vector of utility function coefficients $\vtheta \in \mathbb{R}^{|D|}$, and it restricts the residual to be orthogonal to $D_{\vtheta} \ x(\vtheta)$. A key difference of \NLLS respect to \OLS is that the matrix of exogenous attributes of the regression equation is replaced by the Jacobian matrix $D_{\vtheta} \ x(\vtheta)$ (\ref{ssec:equivalence-ols-nlls}). Under this realization, the statistical properties of the \NLLS problem can be derived analogously to the \OLS case.

		\subsection{Assumptions}
		\label{ssec:assumptions-statistical-inference}

		The proof of consistency of the \NLLS estimator relies on a series of assumptions that has been well-established in Econometrics literature \citep{Wu1981}. An accessible reference on these assumptions is found in \citet{Cameron2005a}. Below are the assumptions reformulated in the context of our problem:

		\begin{assumption}[Model specification]
			\label{assumption:correct-model-specification} 
			The response function is well-specified, i.e. $\forall \vtheta_0 \in \mathbb{R}^{|D|}, \ \rvx = \vx(\vtheta_0, \mZ, \bar{\vt}) + \vu$
		\end{assumption}
		
		\begin{assumption}[Orthogonality between errors and regressors]
			\label{assumption:no-correlation-regressors-errors} 
			In the data generating process, $\E[\vu|\mZ, \bar{\vt}] = \vzero$ and $\E[\vu \vu^{\top}|\mZ, \bar{\vt}] = \mOmega_0$, where $\mOmega$ is the covariance matrix of the errors terms
		\end{assumption}
		
		\begin{remark}
			Assumption \ref{assumption:no-correlation-regressors-errors} allows for heterocedastic errors in the \NLLS problem but it requires knowing or estimating a functional form of the covariance matrix $\mOmega_0$ \citep{Cameron2005a}.  To ease the analysis, we introduce an additional assumption:
		\end{remark}
		
		\begin{assumption}[Spherical errors]
			\label{assumption:spherical-link-flows-errors} 
			Errors are homocedastic and non-autocorrelated, that is, their covariance matrix is diagonal $\mOmega_0 = \sigma^2\mI$, for $\sigma \in \sR_{\geq 0}$
		\end{assumption}

		\begin{remark}
			\label{remark:spherical-link-flows-errors} 
			Assumption \ref{assumption:spherical-link-flows-errors} of spherical errors is satisfied if the errors are independent and identically distributed (\textit{iid}) and with a constant variance
		\end{remark}
		
		\begin{assumption}[Identifiability]
			\label{assumption:identifiability} 
			Each traffic flow function $x_i(\cdot)$ satisfies that: $\forall i \in N, \forall \vtheta_1, \vtheta_2 \in \mathbb{R}^{|D|}$, $x_i(\vtheta_1, \mZ, \bar{\vt})$ = $x_i(\vtheta_2, \mZ, \bar{\vt})$ iff $\vtheta_1 = \vtheta_2$
		\end{assumption}

		\begin{remark}
			Assumption \ref{assumption:identifiability} is directly satisfied when $\theta \in \sR$ and the traffic flow functions are (coordinate-wise) monotonic (see Definition \ref{assumption:monotonocity-traffic-count-functions}, Section  \ref{ssec:pseudo-convexity-small-network}).
		\end{remark}

		\begin{assumption}[Rank of the limit distribution of the Hessian matrix approximation]
			\label{assumption:limit-distribution-hessian-approximation}

			The matrix
			
			$$
			\mA_0 = \textmd{plim} \frac{1}{N}  \left[D_{\vtheta} \ x(\vtheta) \Big|_{\vtheta_{\vzero}}\right]^\top  \left[D_{\vtheta} \ x(\vtheta)\Big|_{\vtheta_{\vzero}}\right]
			$$
			
			exists and is finite and nonsingular $\forall \vtheta_0 \in \mathbb{R}^{|D|}$.
			
		\end{assumption}

		\begin{remark}
			
			Assumption \ref{assumption:limit-distribution-hessian-approximation} formalizes the third rule of identifiability discussed in Section \ref{sssec:identifiability-illustrative-example}. Proposition \ref{prop:full-rank-jabobian} shows that this assumption can be tested via the rank of the Jacobian of the traffic functions:
			
			\begin{prop}[Full rank Jacobian and non-singular Hessian matrix approximation]
				\label{prop:full-rank-jabobian}
				Suppose the Jacobian matrix of the traffic flow functions given by $D_{\vtheta} \phantom{'}\vx(\vtheta) \in \sR^{n \times d}$, with $n >d$, is full rank at $\vtheta_{\vzero}\in \mathbb{R}^{|D|}$. Then, the square matrix $ \left[D_{\vtheta = \vtheta_{\vzero}} \phantom{'} \vx(\vtheta)\right]^{\top} D_{\vtheta= \vtheta_{\vzero}} \phantom{'}
				\vx(\vtheta) \in \sR^{d \times d}$ is positive definite.
			\end{prop}
			
			\begin{proof}
				Consider $\forall \vu \in \sR^{m} \neq \vzero$:
				\begin{align*}
					\vu^T \left(2 \left[D_{\vtheta} \phantom{'}\vx(\vtheta)\right]^{\top} D_{\vtheta} \phantom{'}\vx(\hat{\vtheta})\right) \vu
					= 2  (D_{\vtheta= \vtheta_{\vzero}} \phantom{'}\vx(\vtheta) \vu)^\top \left(D_{\vtheta= \vtheta_{\vzero}} \phantom{'}\vx(\vtheta) \vu \right)
					= \|D_{\vtheta= \vtheta_{\vzero}} \vx(\vtheta) \vu\|_2^2
				\end{align*}
				
				By assumption, the Jacobian matrix $D_{\vtheta = \vtheta_{\vzero}} \ \vx(\vtheta)$ is full rank at $\vtheta_{\vzero}$ and thus $D_{\vtheta = \vtheta_{\vzero}} \ \vx(\vtheta) \vu = \vzero$ iff $\vu = \vzero$. Since $\vu \neq 0$, $\|D_{\vtheta= \vtheta_{\vzero}} \vx(\vtheta) \vu\|_2^2 > 0$ which proves that the matrix $\left[D_{\vtheta = \vtheta^{\star}} \phantom{'} \vx(\vtheta)\right]^{\top} D_{\vtheta= \vtheta_{\vzero}} \phantom{'}
				\vx(\vtheta)$ is positive definite and it completes the proof.
			\end{proof}
			
			
			Thus, if the Jacobian of the traffic flow functions $D_{\vtheta} \ x(\vtheta)$ has full rank, by Proposition \ref{prop:full-rank-jabobian}, the matrix specified in the argument of the probability limit defined in Assumption \ref{assumption:limit-distribution-hessian-approximation} is positive definite and hence non-singular. 
			
			%
		\end{remark}
		
			
			\begin{assumption}[Central limit theorem]
				\label{assumption:convergence-distribution-first-ordercondition} 
				$N^{-1/2} \sum_{i = 1}^N [D_{\vtheta} \ x_i(\vtheta)]^\top \ru_i \xrightarrow[]{d} \mathcal{N}[0, \mB_0]$, where 
				
				$$
				\mB_0 
				= \textmd{plim} \frac{1}{N} \left[D_{\vtheta} \ x(\vtheta) \Big|_{\vtheta_{\vzero}}\right]^\top  \mOmega_0 \  \left[D_{\vtheta} \ x(\vtheta)\Big|_{\vtheta_{\vzero}}\right]
				= \textmd{plim} \frac{\sigma^2}{N} \left[D_{\vtheta} \ x(\vtheta) \Big|_{\vtheta_{\vzero}}\right]^\top \left[D_{\vtheta} \ x(\vtheta)\Big|_{\vtheta_{\vzero}}\right]
				$$
				
				and $\mOmega_0 = \sigma^2 \mI$ by Assumption \ref{assumption:spherical-link-flows-errors} of homocedasticity.

			\end{assumption}

			\begin{remark}
				When the errors $\vu$ are Gaussian and independently distributed, the linear combination of errors given by $[D_{\vtheta} \ x(\vtheta)]^\top \ru$ is a Multivariate Gaussian. As a result, Assumption \ref{assumption:convergence-distribution-first-ordercondition} holds even in small samples. Alternatively, if each term $[D_{\vtheta} \ x_i(\vtheta)]^\top \ru_i$ follows an arbitrary distribution, the central limit theorem (\CLT) must be invoked to ensure the asymptotic normality of the terms in the sequence  $\{D_{\vtheta} \ x_i(\vtheta)\}_{i= 1 \in N}$.  Since $D_{\vtheta} \ x_i(\vtheta)$ is assumed exogenous for statistical inference, the distribution of each term $i$ of the previous sequence is determined by the distribution of $\ru_i$.
			\end{remark}

			\subsection{Statistical properties of the NLLS estimator}
			
			Suppose $\vtheta_0 \in \mathbb{R}^{|D|}$ satisfies the first order necessary optimality condition defined in Eq. \ref{eq:first-order-optimality-condition-NLLS}. Then, under the assumptions stated in Section \ref{ssec:assumptions-statistical-inference}, we can establish the consistency of the \NLLS estimator $\hat{\vtheta}$ at $\vtheta_0$ (see Proposition 5.6, \citet{Cameron2005a}) and that
			
			\begin{equation}
				\label{eq:convergence-distribution-nlls-estimator}
				\sqrt{N}(\hat{\vtheta}-\vtheta_{\vzero}) \xrightarrow[]{d} \mathcal{N}[0, \mA_{\vzero}^{-1}\mB_0\mA_{\vzero}^{-1}]
			\end{equation}
			
			where
			$$
			\mA_{\vzero} = \textmd{plim} \frac{1}{N}  \left[D_{\vtheta} \ x(\vtheta) \Big|_{\vtheta_{\vzero}}\right]^\top  \left[D_{\vtheta} \ x(\vtheta)\Big|_{\vtheta_{\vzero}}\right]
			, 
			\quad 
			\mB_0 
			= \textmd{plim} \frac{\sigma^2}{N} \left[D_{\vtheta} \ x(\vtheta) \Big|_{\vtheta_{\vzero}}\right]^\top \left[D_{\vtheta} \ x(\vtheta)\Big|_{\vtheta_{\vzero}}\right]
			= \sigma^2\mA_{\vzero}
			$$

			By the non-singularity of $\mA_{\vzero}$ (Assumption \ref{assumption:limit-distribution-hessian-approximation}, Section \ref{ssec:assumptions-statistical-inference}):
			
			$$
			\mA_{\vzero}^{-1}\mB_0\mA_{\vzero}^{-1}
			= \mA_{\vzero}^{-1}\sigma^2 \mA_{\vzero}\mA_{\vzero}^{-1}
			= \sigma^2 \mA_{\vzero}^{-1}
			$$
			
			By rearranging terms in Eq. \ref{eq:convergence-distribution-nlls-estimator}, the asymptotic distribution of $\hat{\vtheta}$ becomes:

			$$
			\hat{\vtheta} \xrightarrow[]{a} \mathcal{N}\left[\vtheta_{\vzero}, \Var(\hat{\vtheta})\right]
			$$
			
			where

					\begin{align}
						\label{eq:asymptotical-distribution-nlls-estimator}
						\Var(\hat{\vtheta}) 
						&= \sigma^2 \left[D_{\vtheta} \ x(\vtheta) \Big|_{\vtheta_{\vzero}}\right]^\top \left[D_{\vtheta} \ x(\vtheta)\Big|_{\vtheta_{\vzero}}\right]
					\end{align}
			
			is the asymptotic variance matrix of the \NLLS estimator. Note that the distribution of the \NLLS estimator resembles \OLS, except that the design matrix $\mX$ is replaced by $\tilde{\mX} = D_{\vtheta = \vtheta_{\vzero}} \ x(\vtheta)$ (\ref{ssec:equivalence-ols-nlls}). Also, similar to the \OLS case, the statistical behavior of the \NLLS estimator and the variance of the error terms in large samples can be expressed in terms of $\rvu$ and approximated as \citep{Gallant1975}:
					\begin{align}
						\
						\hat{\vtheta} 
						&= \vtheta + (\tilde{\mX}^T \tilde{\mX})^{-1}\tilde{\mX}^T \rvu \\
						s^2 &= (n-p)^{-1} \rvu^{\top}(\mI-\tilde{\mX}(\tilde{\mX}^{\top}\tilde{\mX})^{-1}\tilde{\mX}^{\top})\rvu
					\end{align}
					
			Note that the variance $\sigma^2 \in \mathbb{R}_{\geq 0}$ of the error terms defined in Eq. \ref{eq:asymptotical-distribution-nlls-estimator} is unknown in real applications but it can be consistently estimated as:

			\begin{equation}
				\label{eq:estimated-variance-nlls}
				{\hat{\sigma}}^2
				= \frac{\RSS(\hat{\vtheta})}{N-|K|}
				= \frac{\|\bar{\vx} - x(\hat{\vtheta})\|_2^2}{N-|K|}
			\end{equation}
		
			where the numerator represents the residual sum of squares (\RSS) of the model and $|K| \in \mathbb{Z}_{+}$ in the denominator represents the dimension of the $\NLLS$ estimator.

			\subsection{Hypothesis testing and confidence intervals}
			\label{ssec:hypothesis-testing}
			
			Based on the asymptotic normality of the \NLLS estimator $\hat{\vtheta}$ proved in the previous section, we can derive confidence intervals and hypothesis tests of $\hat{\vtheta}$. The  $1-\alpha$\% confidence interval of $\hat{\vtheta}$ can be found by using the percentile $\1-\alpha/2$ of a t-variate with $n-|K|$ degrees of freedom $t_{n-|K|, 1-\alpha/2}$ and the variance of the \NLLS estimator as shown in the following formula \citep{Gallant1975}: 
			
			\begin{equation}
				\texttt{CI}(\hat{\vtheta})_{1-\alpha/2} = \hat{\vtheta} \pm t_{n-|K|, 1-\alpha/2}\sqrt{\diag(\Var(\hat{\vtheta}|\mX))}
			\end{equation}
			
			The hypothesis $H_0: \theta_d = \theta_{H_0}$ can be contrasted at the $\alpha$\% confidence level by computing the statistics $\bar{T}_{d, H_0}$:
			
			
			\begin{equation}
				\bar{T}_{d, H_0}
				= \frac{\hat{\theta}_d - \theta_{H_0} }{\sqrt{\diag(\Var(\hat{\vtheta}))_d } }
			\end{equation}
			
			\medskip
			and $H_0$ is rejected when $|\bar{T}_{d, H_0}| > |t_{\1-\alpha/2}|$.

			\subsection{F-test}
			\label{ssec:f-test}
			
			The F-test is another useful statistic for model comparison in non-linear regression and that allows to contrast hypotheses on multiple model parameters \citep{Gallant1975a}. Let's define $\vtheta_{1}, \vtheta_{2} \in \mathbb{R}^{|D|}$ as the vector of coefficients in the restricted (or nested) and unrestricted models. Then, the F-statistic $\bar{F}_{1,2}$ is:
			
			\begin{equation}
				\label{eq:f-statistic-model-comparison}
				\bar{F}_{1,2} = \frac{\dfrac{\RSS(\vtheta_{1})-\RSS(\vtheta_{2})}{\|\vtheta_{2}\|_0-\|\vtheta_{1}\|_0}}{\dfrac{\RSS(\vtheta_{2})}{N-\|\vtheta_{2}\|_0}}
			\end{equation}
			
			\medskip
			where $\|\vtheta_{1}\|_0$ and $\|\vtheta_{2}\|_0$ are the number of non-zero entries in the vectors of parameters of the non-restricted and restricted models. The null hypothesis is that the models are statistically equivalent and it is rejected when $\bar{F}_{1,2} > F_{\|\vtheta_{2}\|_0-\|\vtheta_{12}\|_0, N-\|\vtheta_{2}\|_0,\alpha} \ $, where $F_{\|\vtheta_{2}\|_0-\|\vtheta_{12}\|_0, N-\|\vtheta_{2}\|_0, \alpha}$ is the value of a central $F$-distribution with $\|\vtheta_{2}\|_0-\|\vtheta_{1}\|_0$ numerator degrees of freedom and  $N-\|\vtheta_{2}\|_0^{\phantom{1^2}}$ denominator degrees of freedom and at the $\alpha$-th upper percentile.
		
		\begin{remark}
			An important consideration for hypothesis testing concerns the identifiability of the parameters under the null hypothesis \citep{Gallant1975}. Assume the travelers' utility function is defined as $U = \theta_t t^\alpha$ where $\theta_t,\alpha \in \sR$  are the free parameters and $t$ is the travel time in a given path. If one tests the hypothesis $\theta_t = 0$, $\alpha$ becomes not identifiable. The latter violates Assumption \ref{assumption:limit-distribution-hessian-approximation}, Section \ref{ssec:assumptions} on the rank condition and which is required to prove the consistency and asymptotic normality of the \NLLS estimator. Conversely, under a homogeneous linear-in-parameters utility function (Assumption \ref{assumption:utility}, Section \ref{ssec:assumptions}), the identification of the parameters on a restricted model can be guaranteed under mild conditions (Section \ref{sssec:identifiability-illustrative-example}).

		\end{remark}

%% file: sections/results.tex
\section{Numerical experiments}
\label{sec:numerical-experiments}

To study the performance of the proposed algorithm and some of the mathematical properties of the optimization problem reviewed in Section \ref{sec:mathematical-properties-learning-problem}, we conduct experiments with synthetic data generated from four small networks. Subsequently, we study the statistical properties of the non-linear least squares estimators (\NLLS) of the utility function coefficients via a series of Monte Carlo experiments conducted in the Sioux Falls, SD network. For the experiments, we employ a validation framework where a hypothetical O-D matrix and a vector of the utility function coefficients are assumed as the ground truth and then used them to generate synthetic traffic count measurements consistent with \SUE-\logit. In line with the \LUE literature, the utility function coefficients are estimated under the assumption of an exogenous O-D matrix. All the experiments in this section are conducted on a MacBook Pro with Intel Core i5 CPU 2.7 GHz × 2, 1867 MHz 8 GB RAM, 256 GB SSD. The computation time to run all the experiments in this section was approximately 18 hours.

\subsection{Data generating process (DGP)}
\label{ssec:dgp}

The \DGP to conduct the experiments assumes the modeler has perfect information about the link performance functions, the network topology, the paths sets and the ground truth coefficients and attributes of the travelers' utility function. For each network and replicate of the experiment, we generate a set of synthetic traffic counts under an ideal scenario where link flows perfectly matches \SUE-\logit. To introduce randomness, we add \textit{i.i.d} errors to each set of traffic counts measurements. This \textit{i.i.d} sampling strategy (see Remark \ref{remark:spherical-link-flows-errors}) ensures that Assumptions 	\ref{assumption:no-correlation-regressors-errors} and \ref{assumption:spherical-link-flows-errors}, Section \ref{ssec:assumptions-statistical-inference} are satisfied and thus, the statistical inference for the \NLLS estimator remains consistent. For the sake of convenience, we choose Gaussian errors but other choices of probability distribution are also reasonable, including the Poisson or Negative Binomial distributions. The Gaussian errors are generated with a standard deviation equals to 10\% of the average value of the traffic counts. As a reference, the standard deviation of the error term used in the Monte Carlo experiments conducted by \citet{Gallant1975} was equal to the 3\% of the average value of the dependent variable reported in his synthetic dataset. 

\subsection{Small networks}
\label{ssec:small-networks}

All networks except for the toy network shown in Figure \ref{subfig:toy-network} have been analyzed in prior literature \citep{Wang2016,Yang2001,Lo2003}. For the toy network, the number of trips between origin-destination pairs is assumed equal to $q_{1,4} = 50, \ q_{2,4} = 100, \ q_{3,4} = 150$ and $q_{i,j} = 0,\  \forall (i,j) \neq \{(1,4), (2,4), (3,4)\}$. For the remaining networks, the O-D matrices are equal to those provided by \cite{Wang2016,Yang2001,Lo2003}. The travelers' consideration sets in the four networks include all acyclic paths connecting every O-D pair. The links' performance functions are BPR functions with parameters $\alpha = 0.15$ and $\beta = 4$ and the link capacities and free flow travel times are identical to those reported by \cite{Wang2016,Yang2001,Lo2003}. The utility function is dependent on travel time only, and with a ground truth coefficient $\theta^{\star}_t = -1 < 0$. Note that prior research defines $\theta = -\theta_t >0$ as a positive dispersion coefficient that weights the link/path costs.

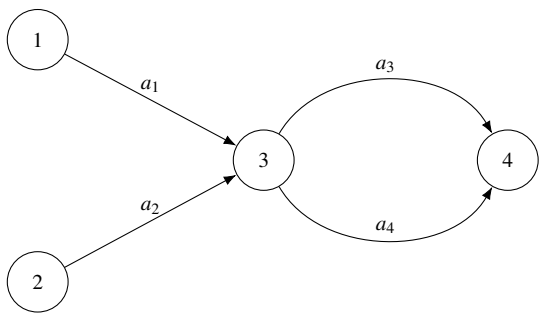
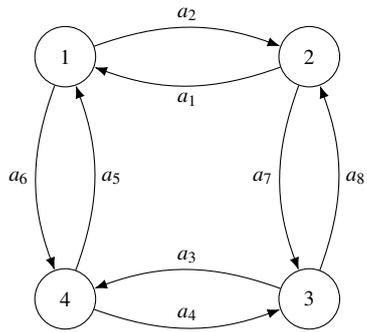
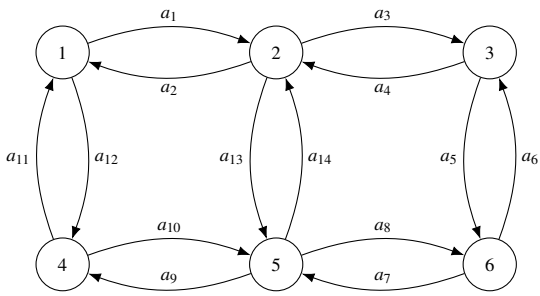
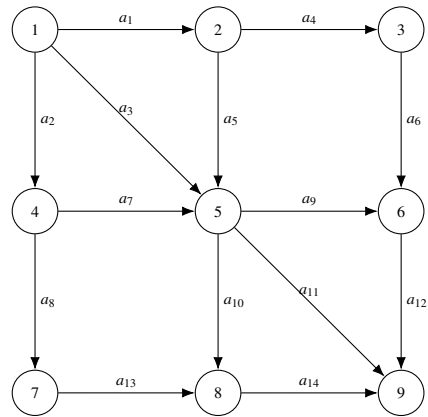
\begin{figure}[!htbp]
	\centering
	
	\begin{subfigure}[b]{0.45\columnwidth}
		
		\centering
		
		\begin{tikzpicture}[>={Latex[scale=1]}, scale = 0.8, transform shape]
			
			\tikzset{main node/.style={circle,draw,minimum size=1cm,inner sep=3pt,node distance=3cm}}
			
			\node[main node] (C1) [align = center] {$1$};
			\node[main node, below = of C1] (C2)  [align = center]  {$2$};
			\node[main node, above right = 1.3cm and 3cm of C2] (C3) [align = center]  {$3$};
			\node[main node, right = of C3] (C4) [align = center]  {$4$};
			
			\path[draw, ->] (C1) edge node[above]{$a_1$}  (C3) ;
			
			\path[draw, ->] (C2) edge node[above]{$a_2$}  (C3) ;
			
			\path[draw, ->] (C3) edge[bend left= 60] node[above]{$a_3$}  (C4) ;
			\path[draw, ->] (C3) edge[bend right= 60] node[above]{$a_4$}  (C4) ;

		\end{tikzpicture} 
		
		\caption{Toy network}
		\label{subfig:toy-network}
		
	\end{subfigure}%
	\begin{subfigure}[b]{0.45\columnwidth}
		
		\centering

		\begin{tikzpicture}[>={Latex[scale=1]}, scale = 0.8, transform shape]
			
			\tikzset{main node/.style={circle,draw,minimum size=1cm,inner sep=3pt,node distance=3cm}}
			
			\node[main node] (C1) [align = center] {$1$};
			\node[main node, right = of C1] (C2)  [align = center]  {$2$};
			\node[main node, below = of C1] (C4) [align = center]  {$4$};
			\node[main node, right = of C4] (C3) [align = center]  {$3$};
			
			\path[draw, ->] (C1) edge[bend left = 20] node[above]{$a_2$}  (C2) ;
			\path[draw, ->] (C2) edge[bend left = 20] node[below]{$a_1$}  (C1) ;
			
			\path[draw, ->] (C1) edge[bend right = 20] node[left]{$a_{6}$}  (C4) ;
			\path[draw, ->] (C4) edge[bend right = 20] node[right]{$a_{5}$}  (C1) ;
			
			\path[draw, ->] (C2) edge[bend right = 20] node[left]{$a_{7}$}  (C3) ;
			\path[draw, ->] (C3) edge[bend right = 20] node[right]{$a_{8}$}  (C2) ;
			
			\path[draw, ->] (C3) edge[bend right = 20] node[above]{$a_{3}$}  (C4) ;
			\path[draw, ->] (C4) edge[bend right = 20] node[above]{$a_{4}$}  (C3) ;

		\end{tikzpicture}

		\caption{\citet{Wang2016} network}
		\label{subfig:Wang-network}
	\end{subfigure}
	
	\vspace{0.5cm}
	
	\begin{subfigure}[b]{0.45\columnwidth}
		
		\centering

		\begin{tikzpicture}[>={Latex[scale=1]}, scale = 0.7, transform shape]
			
			\tikzset{main node/.style={circle,draw,minimum size=1cm,inner sep=3pt,node distance=3cm}}
			
			\node[main node] (C1) [align = center] {$1$};
			\node[main node, right = of C1] (C2)  [align = center]  {$2$};
			\node[main node, right = of C2] (C3) [align = center]  {$3$};
			\node[main node, below = of C1] (C4) [align = center]  {$4$};
			\node[main node, right = of C4] (C5) [align = center]  {$5$};
			\node[main node, right = of C5] (C6) [align = center]  {$6$};
			
			\path[draw, ->] (C1) edge[bend left = 20] node[above]{$a_1$}  (C2) ;
			\path[draw, ->] (C2) edge[bend left = 20] node[below]{$a_2$}  (C1) ;
			
			\path[draw, ->] (C1) edge[bend left = 20] node[right]{$a_{12}$}  (C4) ;
			\path[draw, ->] (C4) edge[bend left = 20] node[left]{$a_{11}$}  (C1) ;
			
			\path[draw, ->] (C2) edge[bend left = 20] node[above]{$a_3$}  (C3) ;
			\path[draw, ->] (C3) edge[bend left = 20] node[below]{$a_4$}  (C2) ;
			
			\path[draw, ->] (C2) edge[bend right = 20] node[left]{$a_{13}$}  (C5) ;
			\path[draw, ->] (C5) edge[bend right = 20] node[right]{$a_{14}$}  (C2) ;
			
			\path[draw, ->] (C3) edge[bend right = 20] node[left]{$a_5$}  (C6) ;
			\path[draw, ->] (C6) edge[bend right = 20] node[right]{$a_6$}  (C3) ;
			
			\path[draw, ->] (C5) edge[bend left = 20] node[above]{$a_{9}$}  (C4) ;
			\path[draw, ->] (C4) edge[bend left = 20] node[above]{$a_{10}$}  (C5) ;

			\path[draw, ->] (C5) edge[bend left = 20] node[above]{$a_8$}  (C6) ;
			\path[draw, ->] (C6) edge[bend left = 20] node[above]{$a_7$}  (C5) ;
			
		\end{tikzpicture}    
		
		\vspace{1cm}

		\caption{\citet{Lo2003} network}
		\label{subfig:Lo-Chan-network}
		
	\end{subfigure}%
	\begin{subfigure}[b]{0.45\columnwidth}
		\centering

		\begin{tikzpicture}[>={Latex[scale=1]}, scale = 0.6, transform shape]
			
			\tikzset{main node/.style={circle,draw,minimum size=1cm,inner sep=3pt,node distance=3cm}}
			
			\node[main node] (C1) [align = center] {$1$};
			\node[main node, right = of C1] (C2)  [align = center]  {$2$};
			\node[main node, right = of C2] (C3) [align = center]  {$3$};
			\node[main node, below = of C1] (C4) [align = center]  {$4$};
			\node[main node, right = of C4] (C5) [align = center]  {$5$};
			\node[main node, right = of C5] (C6) [align = center]  {$6$};
			\node[main node, below = of C4] (C7) [align = center]  {$7$};
			\node[main node, right = of C7] (C8) [align = center]  {$8$};
			\node[main node, right = of C8] (C9) [align = center]  {$9$};
			
			\path[draw, ->] (C1) edge node[above]{$a_1$}  (C2) ;
			\path[draw, ->] (C1) edge node[right]{$a_2$}  (C4) ;
			\path[draw, ->] (C1) edge node[above]{$a_3$}  (C5) ;
			\path[draw, ->] (C2) edge node[above]{$a_4$}  (C3) ;
			\path[draw, ->] (C2) edge node[right]{$a_5$}  (C5) ;
			\path[draw, ->] (C3) edge node[right]{$a_6$}  (C6) ;
			\path[draw, ->] (C4) edge node[above]{$a_7$}  (C5) ;
			\path[draw, ->] (C4) edge node[right]{$a_8$}  (C7) ;
			\path[draw, ->] (C5) edge node[above]{$a_9$}  (C6) ;
			\path[draw, ->] (C5) edge node[right]{$a_{10}$}  (C8) ;
			\path[draw, ->] (C5) edge node[above]{$a_{11}$}  (C9) ;
			\path[draw, ->] (C6) edge node[right]{$a_{12}$}  (C9) ;
			\path[draw, ->] (C7) edge node[above]{$a_{13}$}  (C8) ;
			\path[draw, ->] (C8) edge node[above]{$a_{14}$}  (C9) ;
			
		\end{tikzpicture}    
		
		\caption{\citet{Yang2001} network}
		\label{subfig:Yang-network}
		
	\end{subfigure}
	
	\caption{Topologies of small networks}
	\label{fig:topologies-small-networks}

\end{figure}

\subsubsection{Monotonocity of traffic flow functions}
\label{sssec:monotonicity-pseudoconvexity-small-networks}

Figure \ref{fig:monotonicity-yang-network} shows the output of the traffic flow functions for four arbitrary links in each network and for $\theta_t \in [-15,15]$. The first two values of the tuples shown in the legend of each subfigure correspond to the origin and destination nodes of a link. The last element of the tuple is an index that distinguish parallel links connecting the same O-D pair, e.g $a_3, a_4$ in Figure \ref{fig:monotonicity-yang-network}(a). Note that most traffic flow functions are monotonic and some exhibit convex and concave regions that resemble sigmoidal functions. In the toy network, the traffic flow functions of links $(2,3)$ and $(1,3)$ are constant because their link flows are invariant to the value of the travel time coefficient. Therefore, all individuals traveling from node 1 or 2 are forced to traverse those links and thus, the traffic count measurements are not contributing to the identifiability of additional coefficients of the travelers' utility function.

\begin{figure}[H]
	\centering
	
	\begin{subfigure}{0.45\columnwidth}
		
		\centering
		
		\includegraphics[width=\columnwidth]{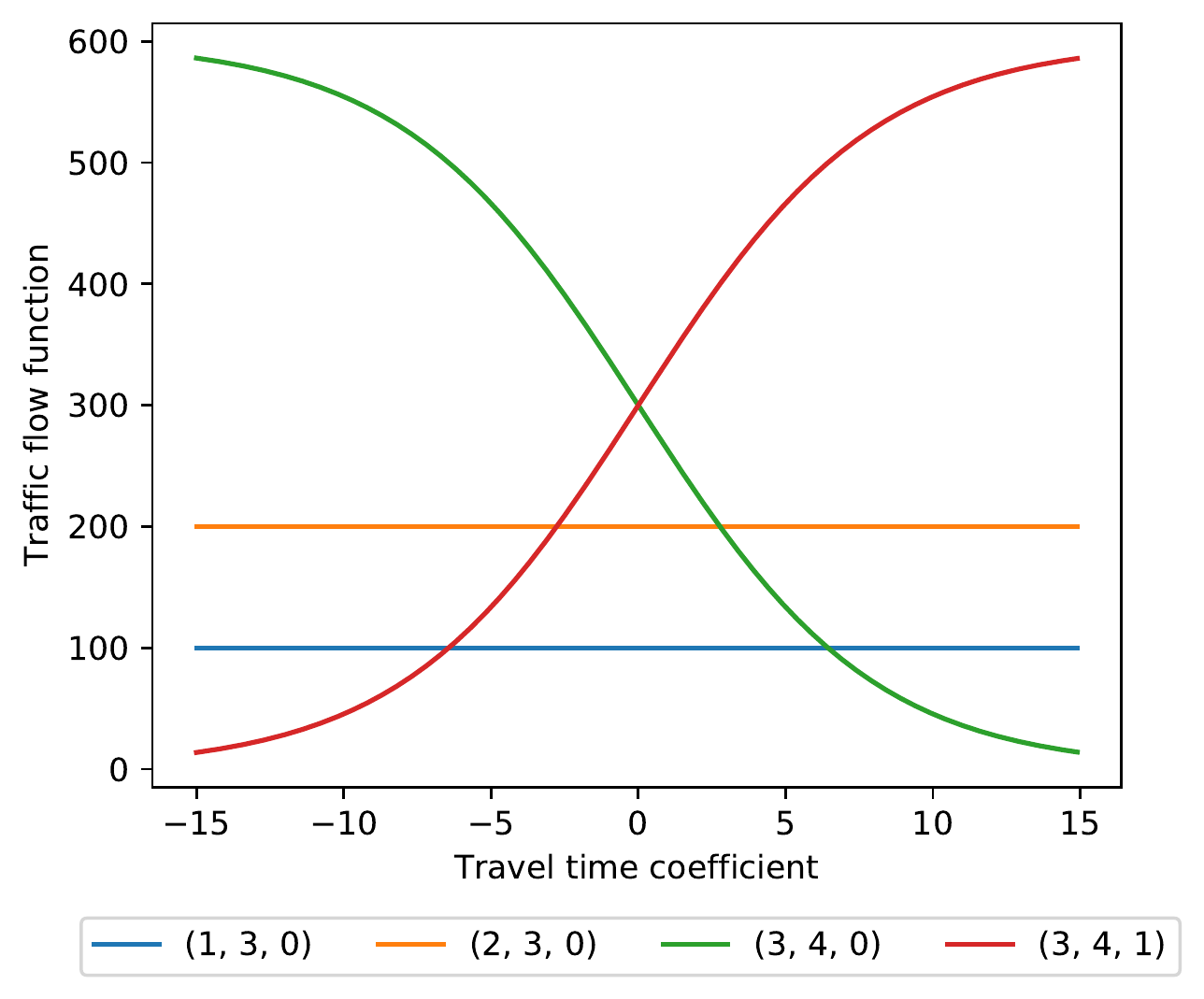}
		
		\caption{Toy network}
		\label{fig:monotonocity-toy-network}
		
	\end{subfigure}	
	\begin{subfigure}{0.45\columnwidth}
		
		\centering
		
		\includegraphics[width=\columnwidth]{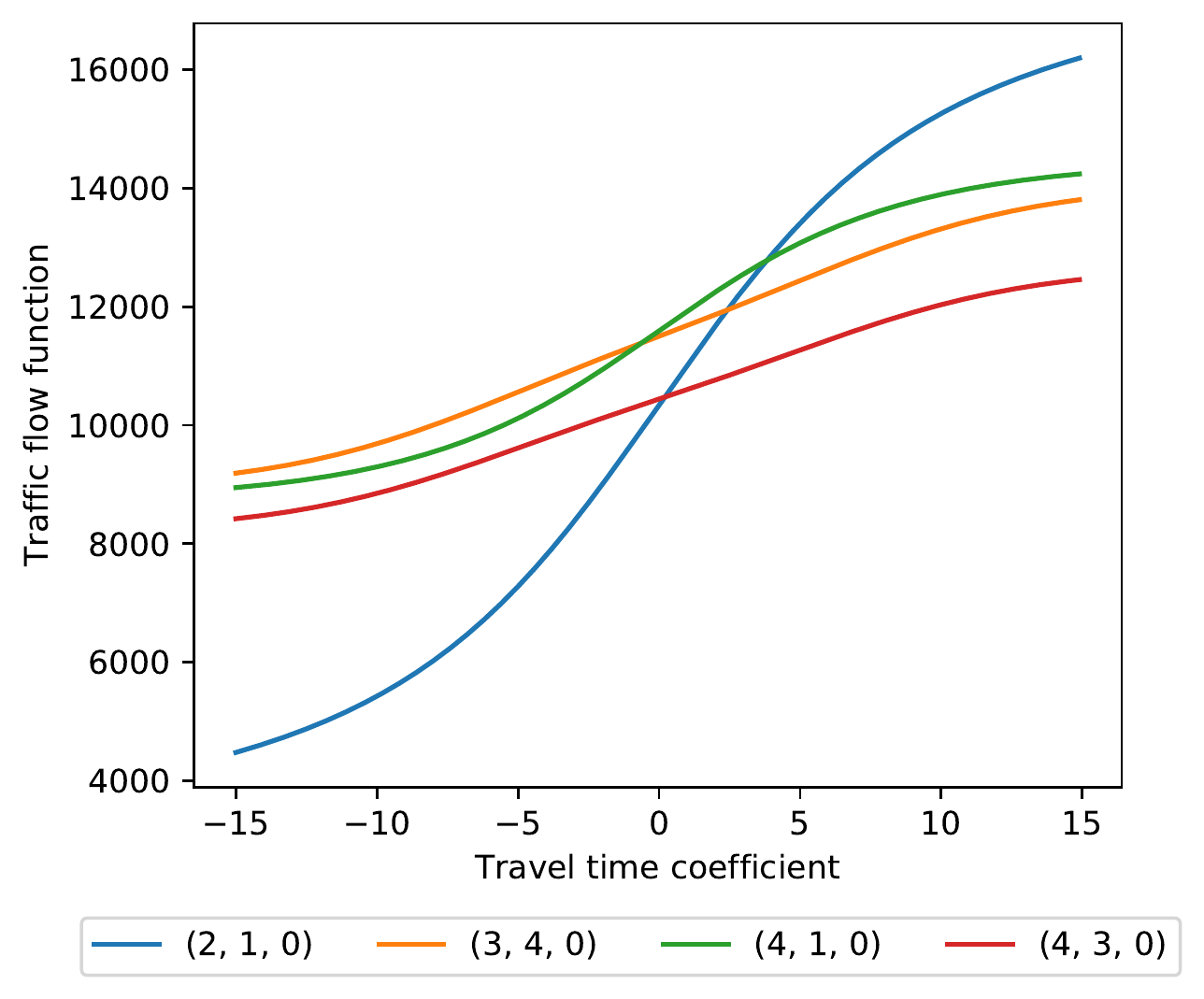}
		
		\caption{Wang network}
		\label{fig:monotonocity-wang-network}
		
	\end{subfigure}
	\begin{subfigure}[t]{0.45\columnwidth}
		
		\centering
		
		\includegraphics[width=\columnwidth]{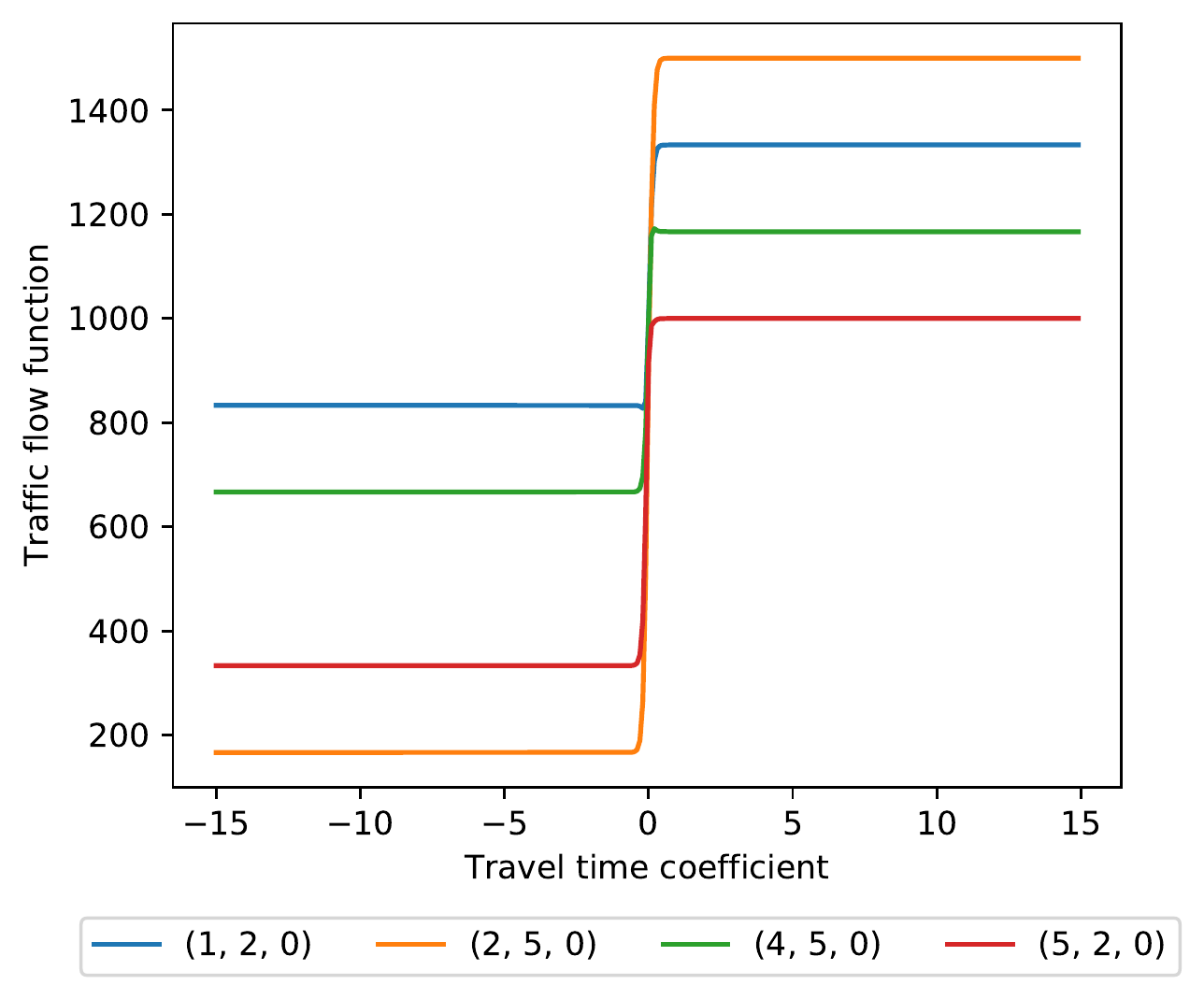}
		
		\caption{Lo \& Chan network}
		\label{fig:monotonocity-lochan-network}
		
	\end{subfigure}%
	\begin{subfigure}[t]{0.45\columnwidth}
		
		\centering
		
		\includegraphics[width=\columnwidth] {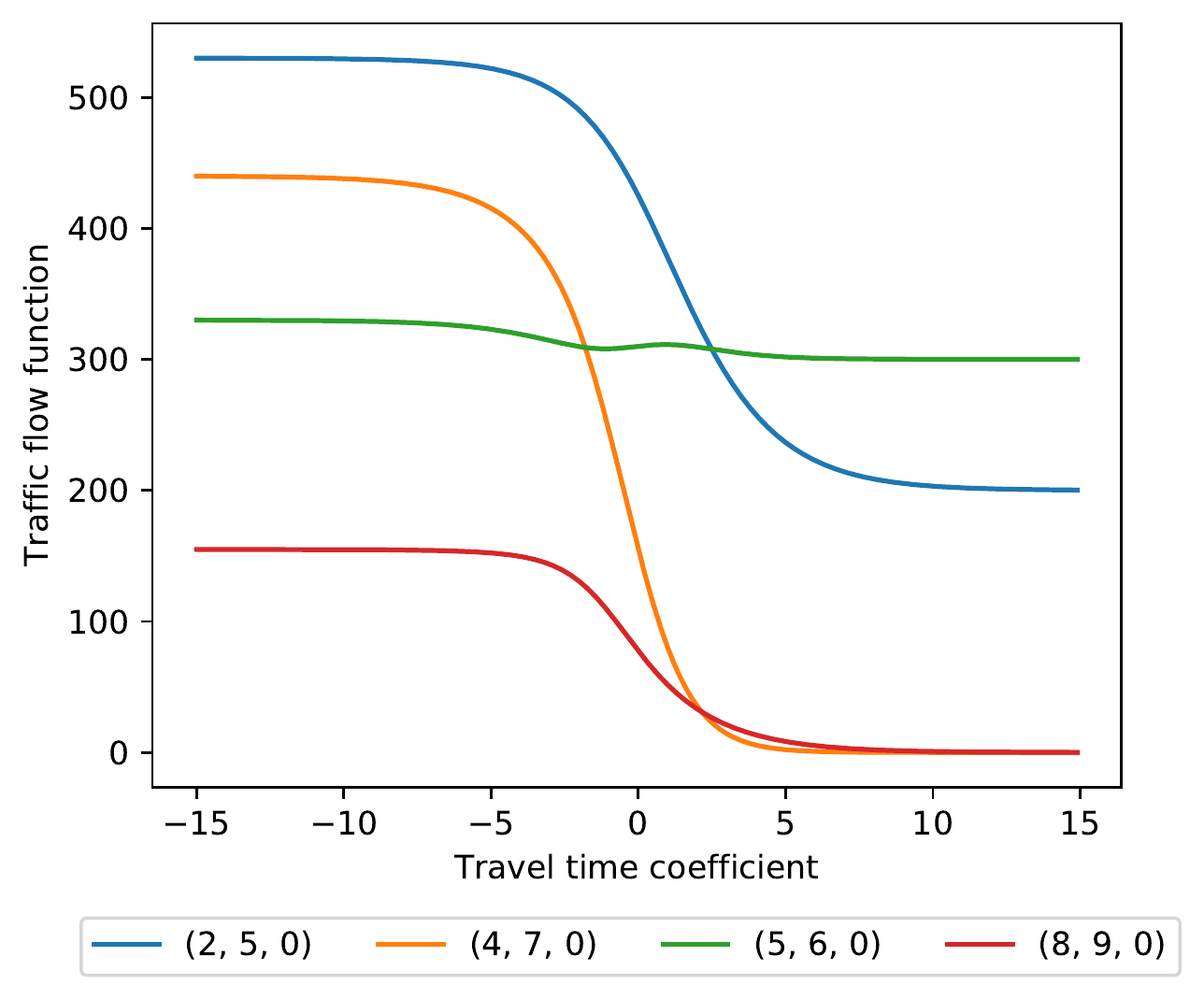}
		
		\caption{Yang network}
		\label{fig:monotonocity-yang-network}
		
	\end{subfigure}%
	
	\caption{Monotonicity of traffic flow functions in small networks}
	\label{fig:monotonicity-yang-network}
	
\end{figure}

\subsubsection{Pseudo-convexity of the objective function}
\label{appendix:sssec:pseudo-convexity-small-networks}

Figure \ref{fig:pseudo-convexity-small-networks} illustrates the coordinate-wise pseudo-convexity of the \LUE problem. The dashed vertical lines represents the true value of the travel time coefficient $\theta_t^{\star}$. Note that since noise was introduced in the traffic count measurements, the error is not necessarily minimized at $\theta_t^{\star}$. In the four networks, the objective function $\ell$ is monotonically decreasing or increasing for the ranges of values that are lower or higher than $\theta^{\star}$, respectively, meaning that the negative slope of the first derivatives are pointing toward the region where the global optima is located (Figure \ref{subfig:sign-first-derivative-pseudo-convexity-small-networks}). As expected, the changes in curvature of $\ell$ within the feasible domain suggest that $\ell$ is not convex but (coordinate-wise) pseudo-convex.

\begin{figure}[H]
	\centering

	\begin{subfigure}[t]{0.49\columnwidth}
		\centering
		\includegraphics[width=\columnwidth] {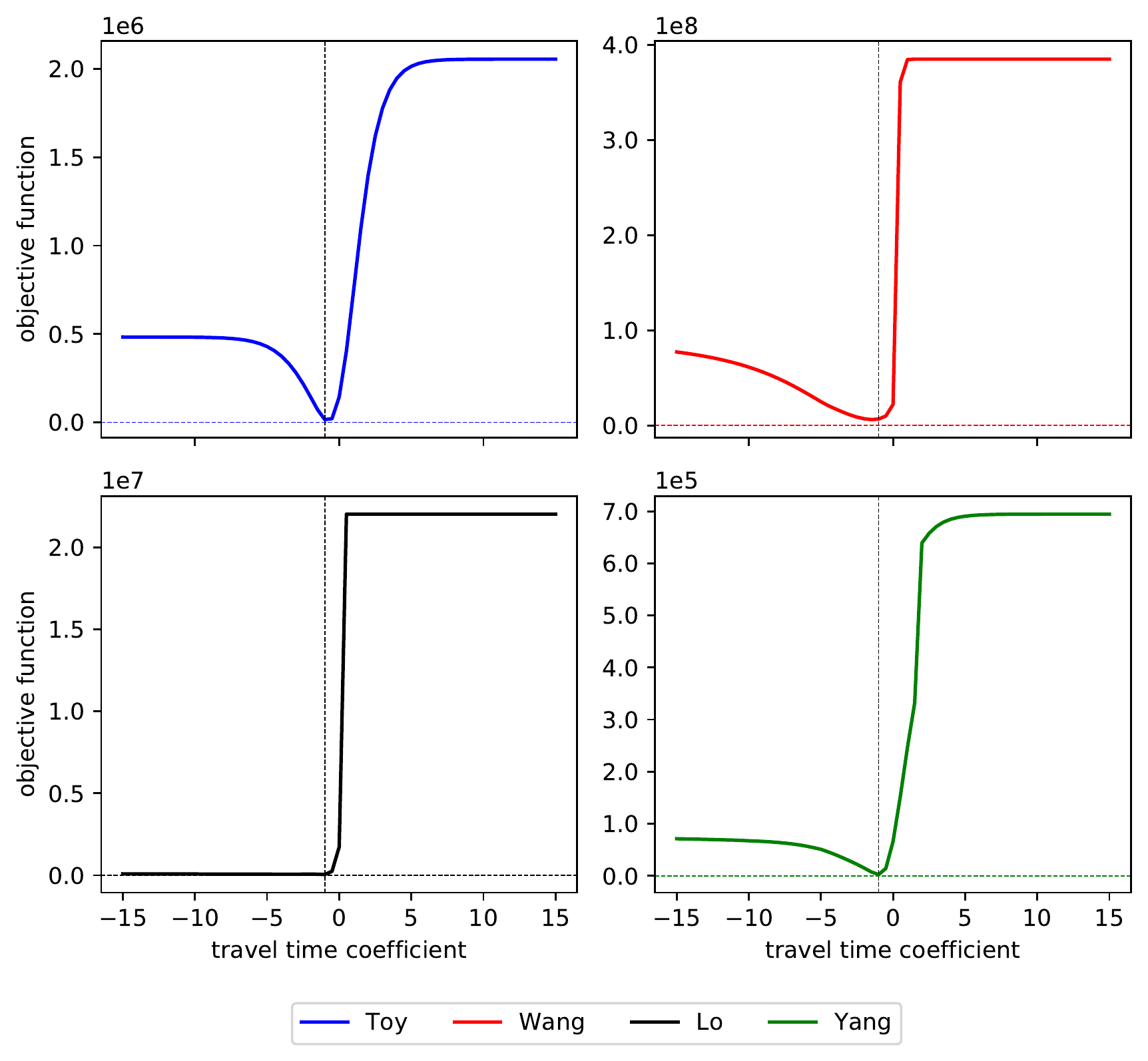}
		\caption{Objective function}
		\label{subfig:objective-function-pseudo-convexity-small-networks}
	\end{subfigure}
	\hspace{0.1cm}
	\begin{subfigure}[t]{0.49\columnwidth}
		\centering
		\includegraphics[width=0.98\columnwidth] {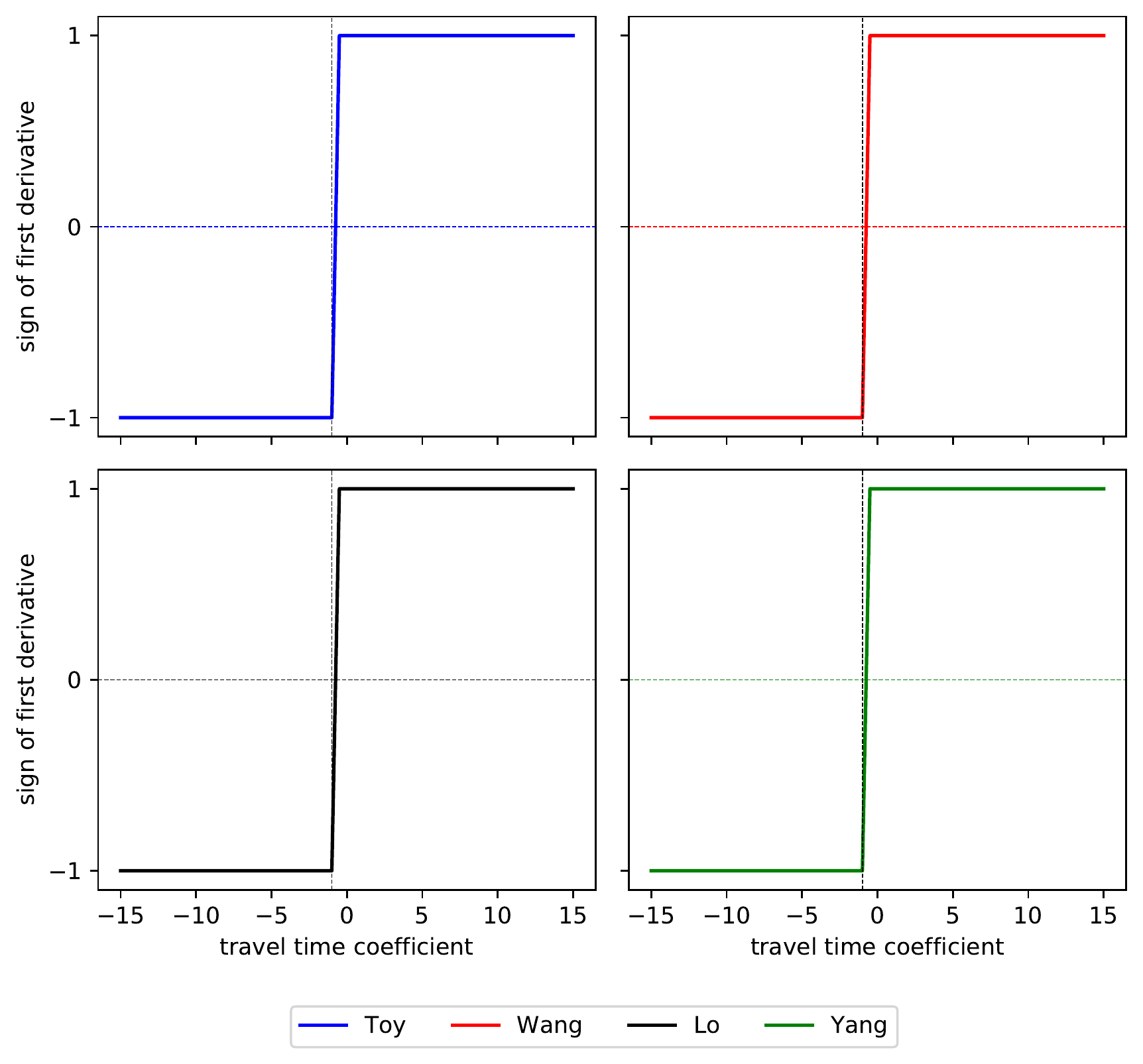}
		\caption{Sign of first derivative of objective function}
		\label{subfig:sign-first-derivative-pseudo-convexity-small-networks}
	\end{subfigure}
	
	\caption{Pseudo-convexity of objective functions in small networks}
	\label{fig:pseudo-convexity-small-networks}
	
\end{figure}

The proof of coordinate-wise pseudo-convexity of the \LUE problem presented in Section \ref{sec:mathematical-properties-learning-problem} requires both coordinate-wise monotonocity of the traffic flow functions and a utility function with exogenous attributes only. Our empirical results suggest that the exogeneity of travel time and the monotonocity of traffic flow function are not necessary but only sufficient conditions for the pseudo-convexity of the objective function.

\subsubsection{Convergence and consistency in parameters' recovery}

To study the advantages of using \NGD instead of a second order method in the non-refined stage of the optimization, we interchange the use of the methods between stages. Figure \ref{subfig:convergence-consistency-small-networks-gn-ngd} presents the convergence results obtained by applying \LM and \NGD in the no-refined and refined stages, respectively. Figure \ref{subfig:convergence-consistency-small-networks-ngd-gn} shows the converse case. To ensure that the global optima at $\theta_t = -1$ is attainable, no noise is introduced into the traffic count measurements. The learning rate for \NGD was set to $\eta =2$ which, for our unidimensional optimization problem on $\theta_t$, equates to solution updates with magnitude equal to 2. The starting point for optimization of the travel time coefficient is set to $\theta_t=-14$, which is the most challenging scenario of convergence analyzed by \citet{Yang2001}. 
\begin{figure}[H]
	\centering
	
	\begin{subfigure}[t]{0.49\columnwidth}
		\centering
		\includegraphics[width=\columnwidth] 
		{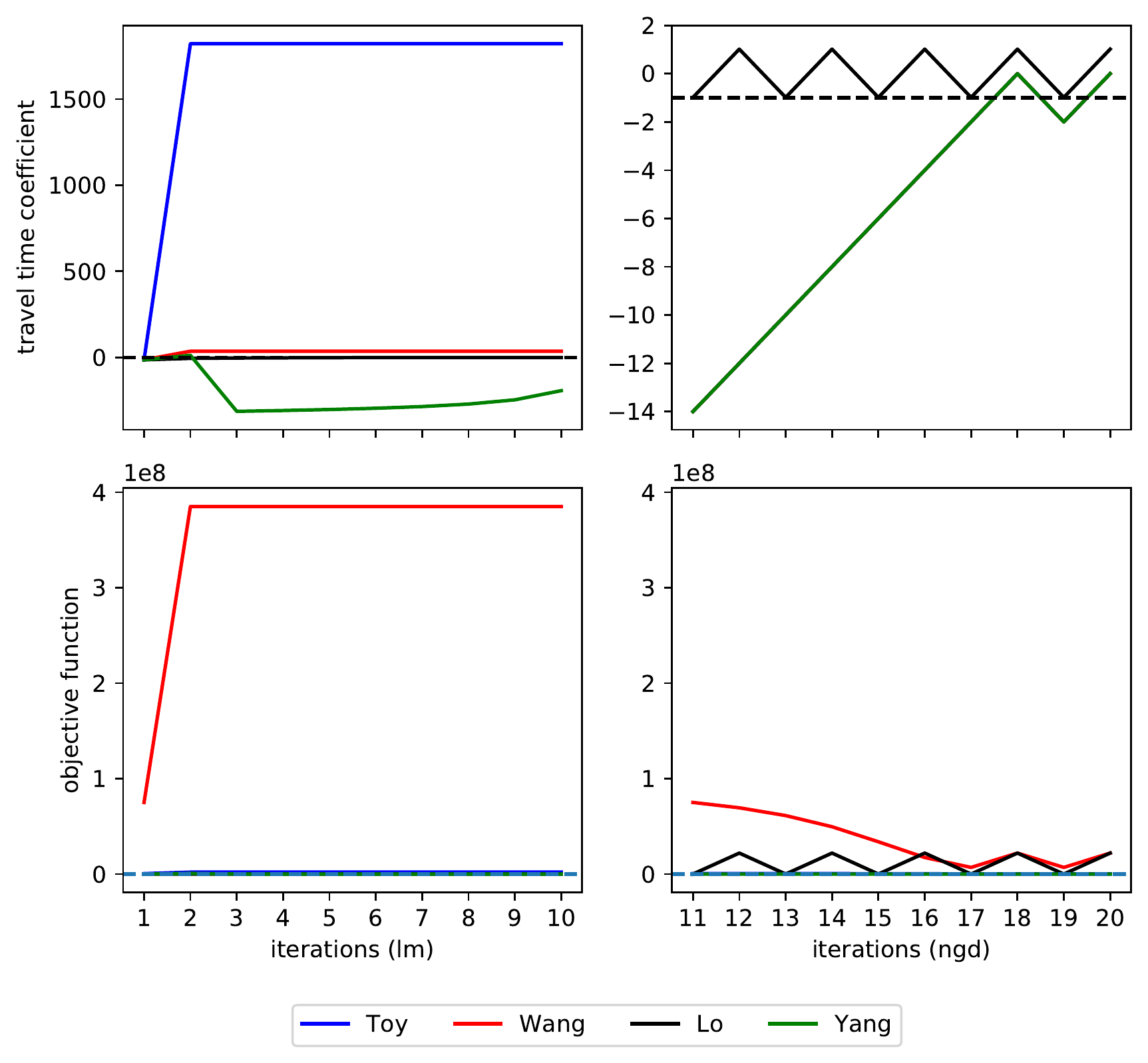}
		\caption{No-refined stage: \LM. Refined stage: \NGD}
		\label{subfig:convergence-consistency-small-networks-gn-ngd}
	\end{subfigure}
	\hspace{0.1cm}
	\begin{subfigure}[t]{0.49\columnwidth}
		\centering
		\includegraphics[width=\columnwidth] {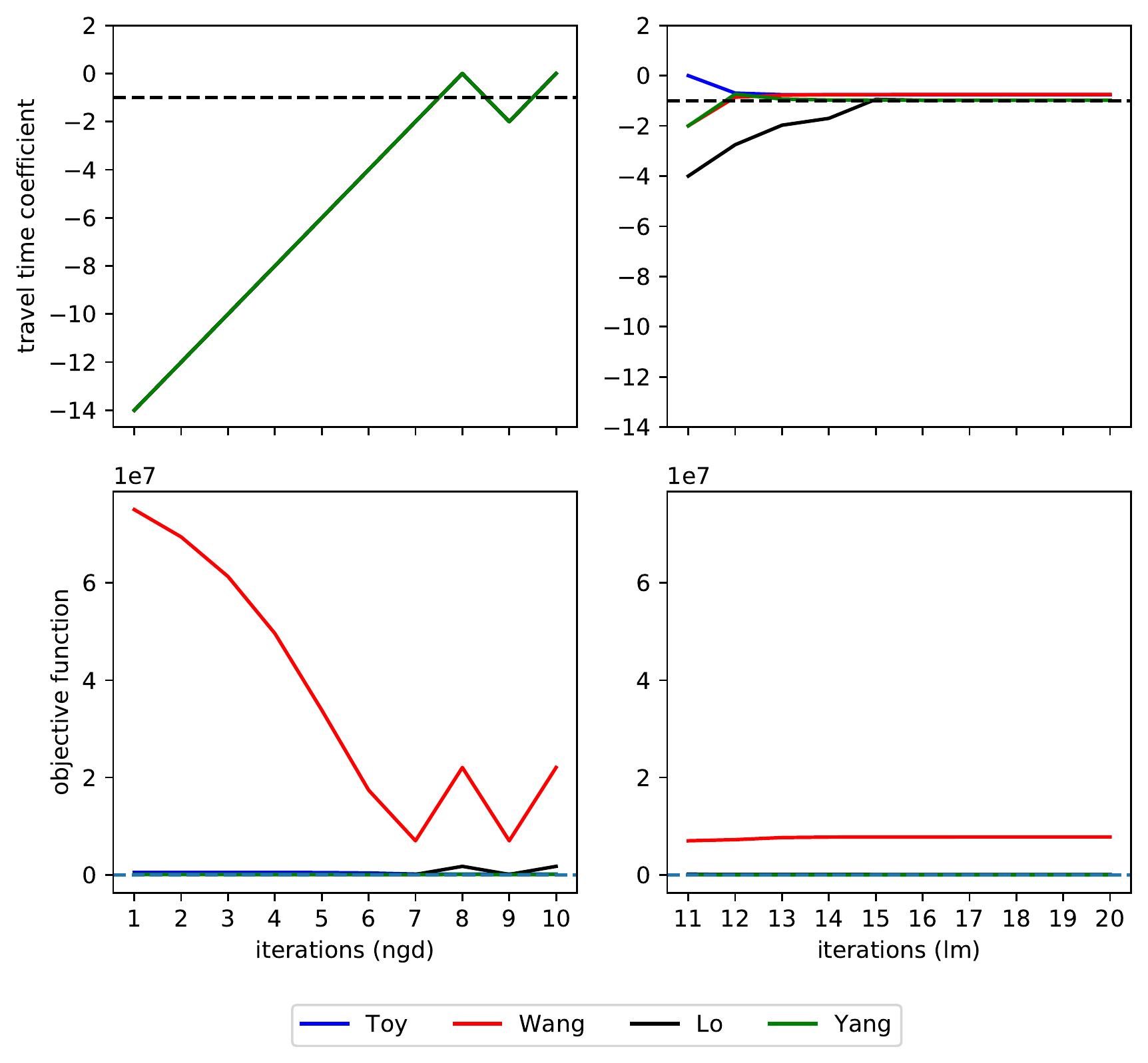}
		\caption{No-refined stage: \NGD. Refined stage: \LM}
		\label{subfig:convergence-consistency-small-networks-ngd-gn}
	\end{subfigure}
	
	\caption{Convergence and consistency in parameter recovery with synthetic data from small networks}
	\label{fig:convergence-consistency-small-networks}

\end{figure}

We observe that the use of \NGD in the no-refined stage makes the convergence faster and stable in the four networks. In contrast, the use of \LM in the no-refined stage results into unstable updates of $\theta_t$. This behavior in second order optimization methods is expected due to the constant change of curvature sign and the existence of flat regions within the range of pseudo-convex functions. In fact, for all networks the objective function is concave and flat at the initial value for optimization $\theta_t = -14$ and it is convex and sharp around the global optima at $\theta_t=-1$  (Figure	\ref{fig:pseudo-convexity-small-networks}).

\subsubsection{Hypothesis testing}
\label{ssec:hypothesis-testing-results}

Table \ref{table:small-networks-inference} shows the parameters estimates and t-tests obtained for the four small networks and under four scenarios; (i) only \NGD, (ii) only \LM, (iii) \NGD in no-refined stage and \LM in refined stage, (iv) \LM in no-refined stage and \NGD in refined stage. Overall, this evidence confirms the advantages of the integration of first and second order optimization methods to perform statistical inference. The integration of \NGD and \LM in scenarios (iii) and (iv) generate solutions closer to the global optima than the standalone application of each optimization method. The hypothesis tests correctly reject the null at a 1\% confidence level but only when \NGD and \LM are used in the no-refined and refined stages of the optimization, respectively.

\begin{table}[H]
	\centering 
	
	\caption{Point estimates and t-tests for travel time coefficient} 
	\label{table:small-networks-inference}
	
	\vspace{-0.2cm}
	
	\begin{adjustbox}{width=\textwidth}  
		
		\begin{threeparttable}

			\begin{tabular}{@{\extracolsep{4pt}}lcccc} 
				\\[-1.8ex]\hline 
				\hline \\[-1.8ex] 
				& \multicolumn{4}{c}{\begin{tabular}{c} Network\end{tabular}} \\
				\cmidrule(ll){2-5}
				\multicolumn{1}{l}{\begin{tabular}{l} \hspace{-0.3cm}\textbf{Parameter (t-test)} \end{tabular}} 
				& Toy ${\!}$ & \citet{Wang2016} & \citet{Lo2003}  & \citet{Yang2001} \\ 
				\hline \\[-1.8ex] 
				\NGD & 0.000 (0.0) & $-$2.000$^{***}$ ($-$5.6)  & $-$2.000$^{***}$ ($-$6.9) & $-$4.000$^{***}$ ($-$3.6)  \\ 
				\LM & $-$14.000 ($-$0.0) & $-$14.000 ($-$0.6)  & $-$0.984$^{***}$ ($-$7.3) & $-$14.000 ($-$1.1)  \\ 
				\NGD$+$\LM & $-$0.761$^{***}$ ($-$7.1) & $-$2.000$^{***}$ ($-$5.6)  & $-$0.976$^{***}$ ($-$15.7) & $-$0.982$^{***}$ ($-$7.3) \\ 
				\LM$+$\NGD & 0.000 (0.0) &  $-$2.000$^{***}$ ($-$5.6)  & $-$0.984$^{***}$ ($-$7.3)  & $-$2.000$^{***}$ ($-$6.8)  \\ 
				\hline \\[-1.8ex] 		
				Links (coverage) & 4 (100\%)  & 8 (100\%) & 14 (100\%)  & 14 (100\%) \\ 
				Paths & 6  & 24  & 44 & 28 \\ 
				O-D pairs & 3  & 12  & 12 & 9 \\ 
				\hline 
				\hline \\[-1.8ex]

			\end{tabular}

			\begin{tablenotes}
				\vspace{-0.2\baselineskip}
				\scriptsize
				\item Note: Significance levels: $^{*}$p$<$0.1; $^{**}$p$<$0.05; $^{***}$p$<$0.01
			\end{tablenotes}

		\end{threeparttable}
	\end{adjustbox}
	
\end{table}

\subsubsection{Impact of error in reference OD matrix}
\label{sssec:error-reference-od-small-networks}

To study the impact on convergence and inference of assuming an inaccurate reference O-D matrix, the experiment and testing network used by \cite{Yang2001} were used as a baseline. The ground truth O-D matrix was defined as $\vq^{\textmd{true}} = (q_{1,6}, q_{1,8}, q_{1,9}, q_{2,6}, q_{2,8}, q_{2,9}, q_{4,6}, q_{4,8}, q_{4,9}) = (120,150,100,130,200,90,80,180,110)$ whereas the reference O-D matrix was $\vq^{\textmd{distorted}} = (100,130,120,120,170,140,110,170,105)$. In line with \cite{Yang2001}, it is assumed that traffic counts are only observed in the following links of the network:  $(3,6), (5, 6), (5, 8), (5, 9), (7, 8)$.

Figure \ref{fig:convergence-yang-network-od-experiment} shows the convergence toward the optimal solution when the reference O-D matrix was assumed equal to (i) the true O-D matrix $\vq^{\textmd{true}}$ (blue curve) or to (ii) the distorted O-D matrix $\vq^{\textmd{distorted}}$ (red curve). Note that in the two scenarios, the algorithm converges close to the global optima ($\theta_t^{\textmd{true O-D}} = -1.058, \ \theta_t^{\textmd{distorted O-D}} = -1.054$) and the objective function ($\ell_\star^{\textmd{true O-D}} = 758.8, \ \ell_\star^{\textmd{distorted O-D}} = 12301.6$) is minimized over iterations. Note that when $\vq = \vq^{\textmd{true}}$, a zero objective function is not attainable because the introduction of random error in the traffic count measurements (see Section \ref{ssec:dgp}). Despite the bias in the reference O-D matrix, the null was correctly rejected under the two scenarios at a 95\% level of confidence (true O-D: $T_{H_0:\theta_t = 0} = -10.3, \ p < 0.01$,  distorted O-D: $T_{H_0:\theta_t = 0} = -4.1, \ p < 0.05)$)

\begin{figure}[H]
	\centering
	\includegraphics[width=0.7\columnwidth] {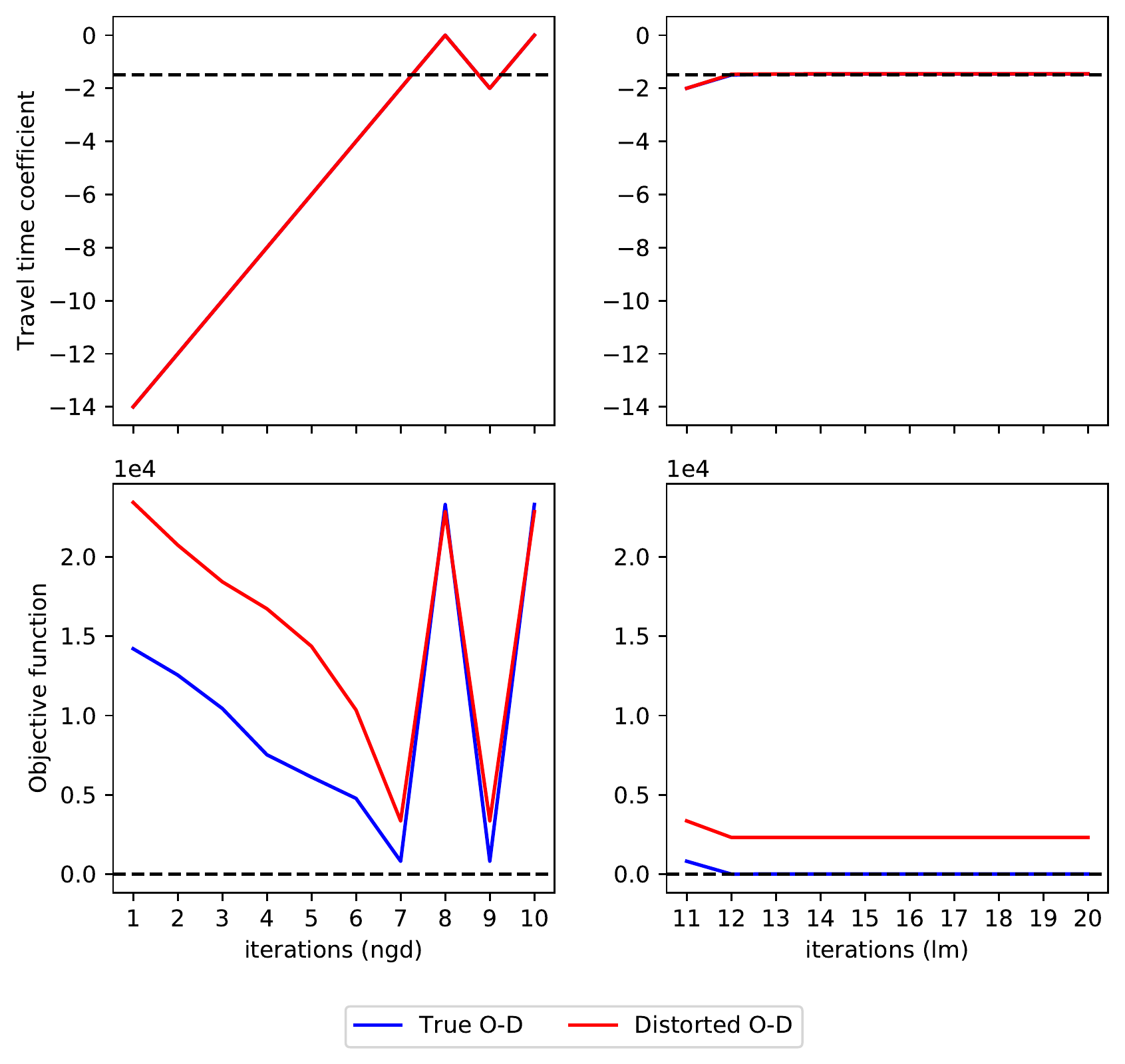}
	\caption{Convergence under true and distorted reference OD matrix}
	\label{fig:convergence-yang-network-od-experiment}
\end{figure}

\subsection{A medium-scale network: Sioux Falls}
\label{ssec:sioux-falls-experiments}

The Sioux Falls network is a standard testing bed used by transportation researchers in network modeling studies \citep{Shen2012}. The network comprises 24 nodes and 76 links (Figure \ref{subfig:sioux-falls-network}). We generate 1,584 paths corresponding to the three paths with the shortest distance to travel between each O-D pair. The O-D matrix is obtained from \citet{TNTP2016} and a heatmap with its values is presented in Figure \ref{subfig:sioux-falls-od-matrix}. The setup to generate the synthetic data is the same as for the small networks except that the exogenous attributes for the monetary cost $c$and the number of intersections $s$ at each link/path are incorporated in the utility function. The attributes $c$ and $s$ are generated as continuous and discrete random variables, respectively, and the utility function coefficients are set to $\theta _t = -1, \ \theta_c = -6, \ \theta_s = -3$. If the units of the travel time and monetary cost attributes are minutes and US dollars, the value of time (\texttt{VOT}) become equal to 10 USD\$ per hour. By Assumption \ref{assumption:utility}, Section \ref{ssec:assumptions}, the set of coefficients $\theta_t, \theta_c, \theta_s \in \sR$ are common among travelers and linearly weighting the attributes $t$, $c$, $s$ in the utility function. The true path set is assumed to be known and thus, the experiments do not perform the column generation phase of the inner level optimization algorithm (Step 1, Algorithm \ref{alg:inner-level-optimization}, \ref{appendix:ssec:implementation-inner-level-optimization})

\begin{figure}[H]
	\centering
	
	\begin{subfigure}[t]{0.45\columnwidth}
		\centering
		\includegraphics[width=0.7\columnwidth, trim= {7cm 0cm 10cm 0cm},clip]{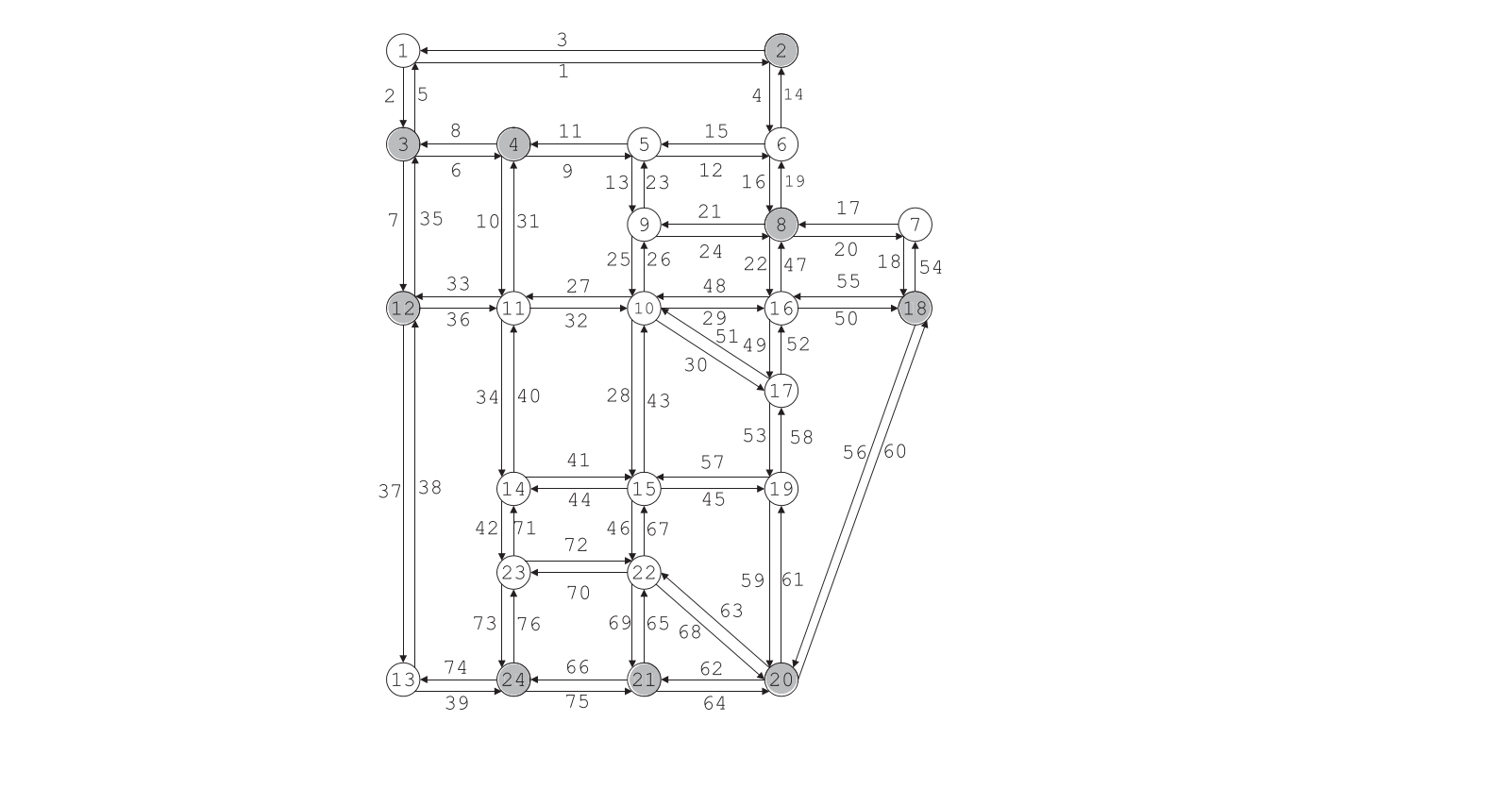}
		\caption{Network topology}
		\label{subfig:sioux-falls-network}
	\end{subfigure}
	\begin{subfigure}[t]{0.45\columnwidth}
		\centering
		\includegraphics[width=\columnwidth, trim= {0cm 0cm 0cm 0cm},clip]{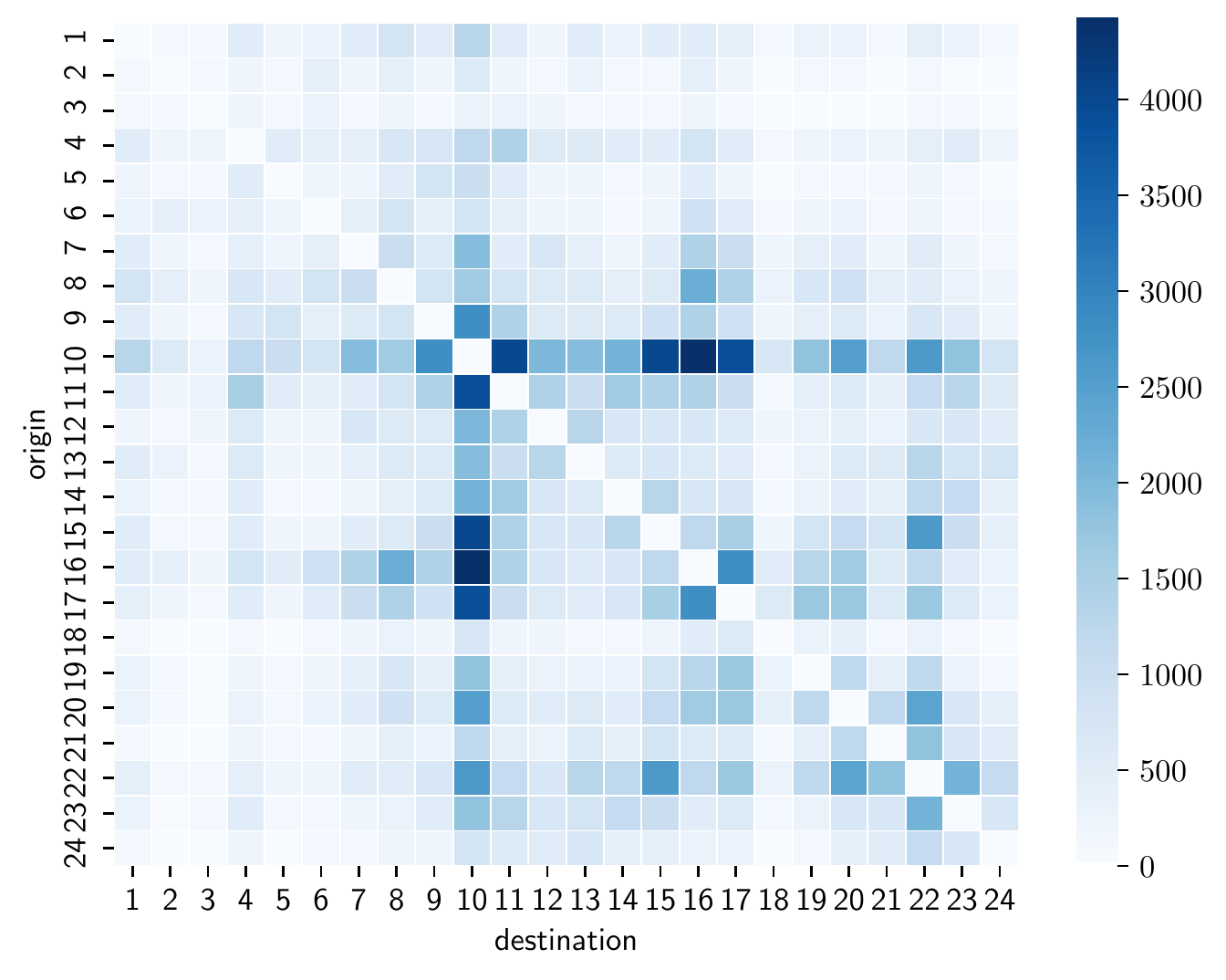}
		\caption{O-D matrix}
		\label{subfig:sioux-falls-od-matrix}
	\end{subfigure}
	
	\caption{Topology and O-D matrix of Sioux Falls network}
	\label{fig:sioux-falls}
	
\end{figure}

\subsubsection{Pseudo-convexity of the objective function}
\label{sssec:pseudo-convexity-sioux-falls}

To have a detailed characterization of the coordinate-wise pseudo-convexity of the outer level optimization problem, Figure \ref{fig:pseudo-convexity} shows plots of the objective function $\ell$(top left), the first of derivative of $\ell$ (top right), the sign of the first derivative of $\ell$ (bottom left) and the sign of the second derivative of $\ell$ (bottom right) respect to a coefficient of the utility function. The vectorized expressions to compute the first and second derivatives of the objective function are presented in \ref{appendix:sssec:gradients} and \ref{appendix:sssec:second-derivatives}. To avoid an excess of overlap of multiple curves when included in the same figure, we only analyze the curves associated to the travel time and cost coefficients. The true values of the coefficients are represented with vertical dashed lines in each figure and the curve associated to each coefficient was generated by fixing the value of the other coefficient to its true value. To satisfy the exogeneity assumption required for coordinate-wise pseudo-convexity in the objective function (Section \ref{ssec:exogenous-utility-attributes}), travel times are assumed known and equal to the travel times obtained when generating the synthetic traffic counts at \SUE-\logit. 

Similar to the results obtained for the small networks (Section \ref{sssec:monotonicity-pseudoconvexity-small-networks}), we observe that the objective function is locally convex and minimized at a point close to $\theta_t = -1, \theta_c = -6$. In addition, the sign of the first derivatives are always pointing toward the global optima, which illustrate this key property of pseudo-convex functions. In contrast to the results obtained in small networks, the sign of the second derivative changes more frequently. The latter may be associated with the higher complexity in larger networks which tends to increase the non-linearity of the traffic flow functions and hence, of the objective function, respect to the utility function coefficients.

\begin{figure}[h]
	\centering
	\includegraphics[width=0.7\columnwidth]{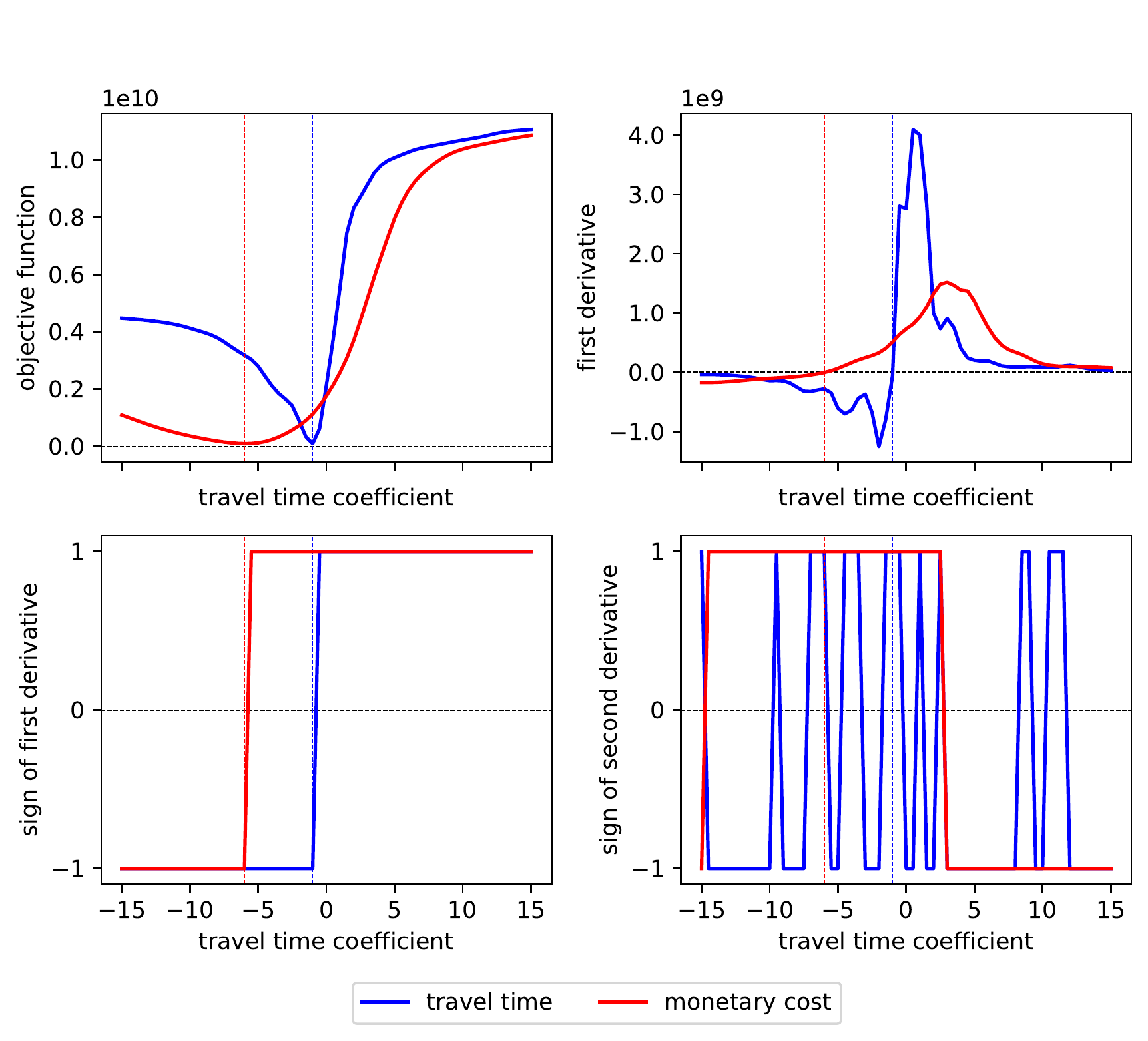}
	\caption{Pseudoconvexity of the objective function in Sioux Falls network}
	\label{fig:pseudo-convexity}
\end{figure}


\subsubsection{Convergence with a multi-attribute utility function}

The coordinate-wise pseudo-convexity of the outer level objective function provides theoretical guarantees for the convergence of \NGD toward global optima but only when the utility function depends on a single exogenous attribute (Section \ref{sssec:convergence-ngd-uncongested-network}). Interestingly, our experimental results suggests the ground truth coefficients of a multi-attribute utility function can be also recovered in the endogenous case. The top and bottom plots in Figure \ref{fig:convergence-sioux-falls-ngd-lm} show the values of the ratio between the travel time and cost coefficients, i.e. the value of time, and the value of the objective function over iterations, respectively. The non-refined and refined stages of the optimization performed 10 iterations of \NGD and  \LM, respectively. The starting points for optimization of the utility function coefficients were set to zero. In the exogenous and endogenous cases, we observe that the value of time is perfectly recovered and that the value of the objective function is close to zero. In the exogenous case, travel times are assumed to be known and thus, each iteration reduces to minimize the outer level objective function only. In the endogenous case, the convergence is less stable due to the additional computation of \SUE-\logit at each iteration of the bilevel optimization.

\begin{figure}[H]
\centering
\includegraphics[width=0.7\columnwidth] {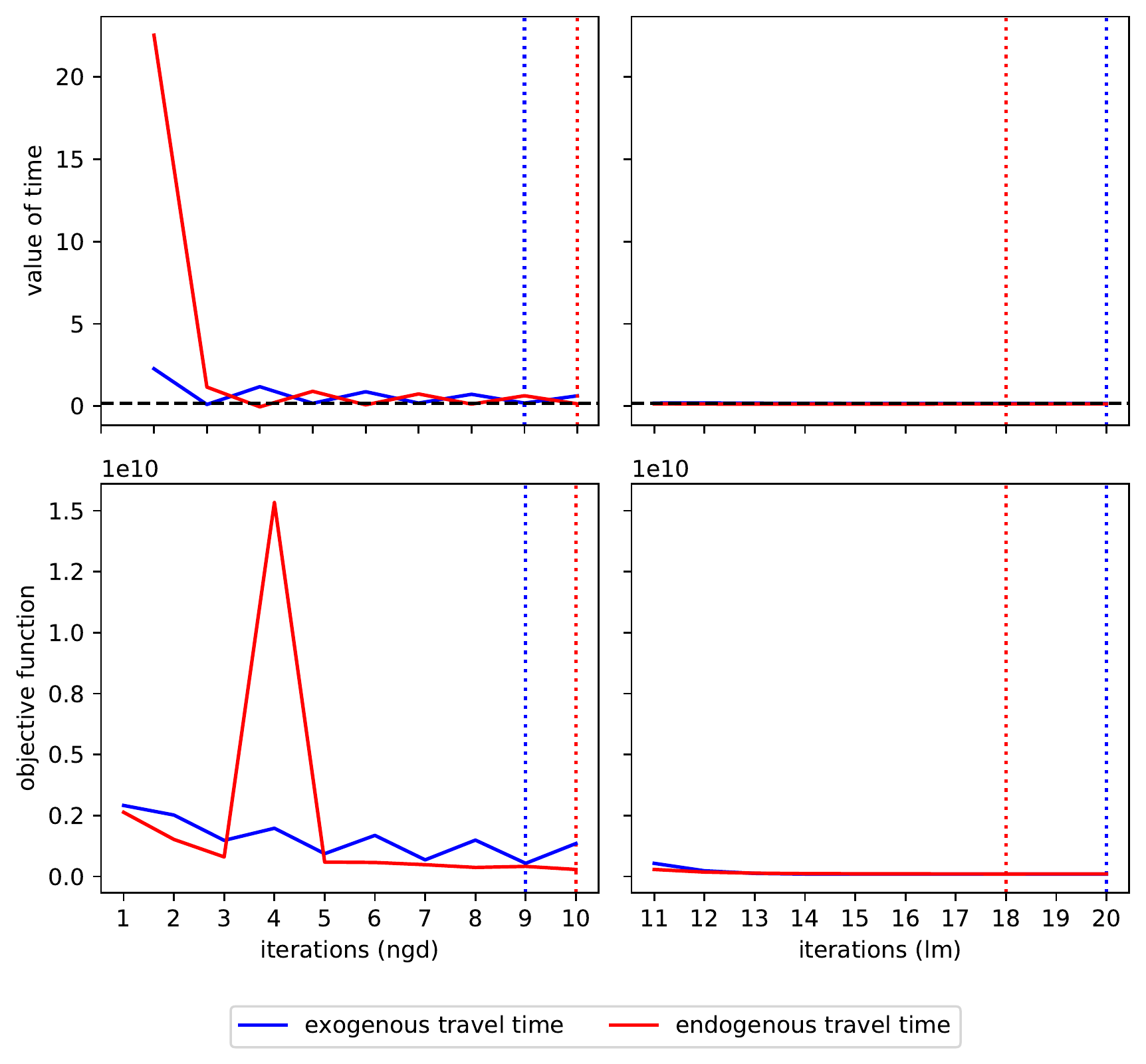}
\caption{Convergence to ground truth value of time in Sioux Falls network}
\label{fig:convergence-sioux-falls-ngd-lm}
\end{figure}

\subsubsection{Impact of endogeneity of travel times on statistical inference}
\label{sssec:impact-congestion-statistical-inference}

We conduct 100 replicates of two Monte Carlo experiments to study the impact of the endogeneity of travel times on the statistical inference of the utility function coefficients. The significance level for all hypotheses tests is set at $\alpha = 0.1$. Monte Carlo experiments are often used by the discrete choice modeling community to empirically study the properties of estimators in travel behavior models \citep{Train2002a}. We also study the impact on statistical inference of three different setups of the optimization algorithm; (i) \NGD is used in the no-refined stage (ii) \LM is used in the no-refined stage, (iii) \NGD and \LM are used in the no-refined and refined stages, respectively. In scenarios (i) and (ii), no iterations of the optimization methods are performed in the refined stage. 

Each replicate of the experiment draws a new sample of errors from a Gaussian distribution (Section \ref{ssec:dgp}) and hence of new traffic counts. For each replicate, we compute the bias of the utility function coefficients and of the value of time, the normalized root mean squared error (\texttt{NRMSE}) and the amount of false negatives. The \texttt{NRMSE} is defined as the ratio between the root mean squared error (\texttt{RMSE}) and the mean of the observed traffic counts. Because the standard deviation $\sigma$ of the random error is defined as a proportion of the mean of true traffic counts (Section \ref{ssec:dgp}), the \texttt{NRMSE} is expected to be close to $\sigma = 0.1$. Since the ground truth coefficients of the utility function are all set to values different than zero, a false negative equates a non-rejection of the null hypothesis $H_{0} = 0$. The initial values for optimization of the utility function coefficients are set with a uniform distribution centered at the true values of the coefficients and with a width of two. Thus, the initial value of the travel time coefficient $\theta_t$ is randomly chosen within the interval $[-3,1]$. 

\begin{figure}[H]
\centering

\begin{subfigure}[t]{0.49\columnwidth}
	\centering
	\includegraphics[width=\columnwidth]{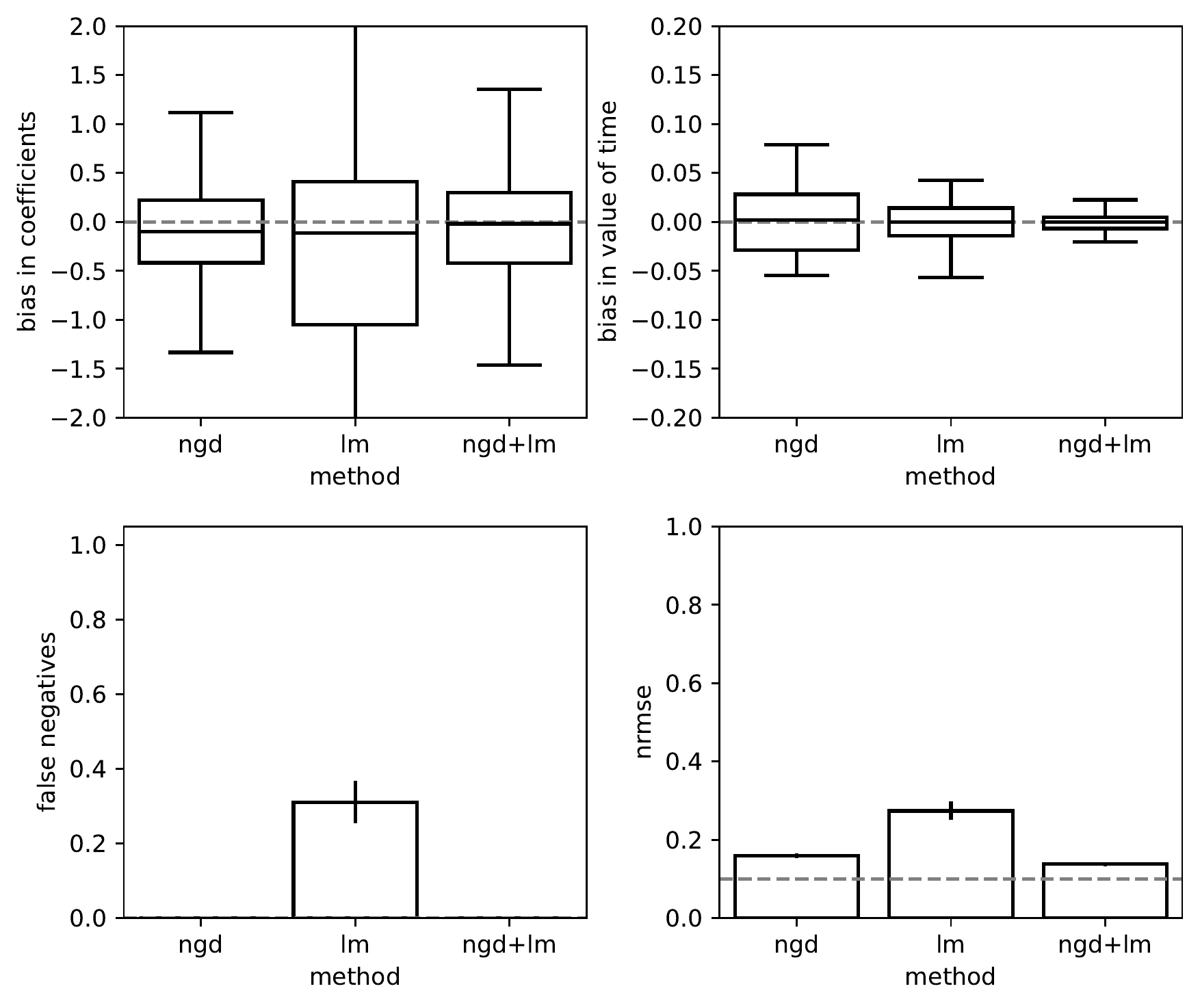}
	\caption{Exogenous travel times}
	\label{subfig:inference-uncongested-network-sioux-falls}
\end{subfigure}
\begin{subfigure}[t]{0.49\columnwidth}
	\centering
	\includegraphics[width=\columnwidth]{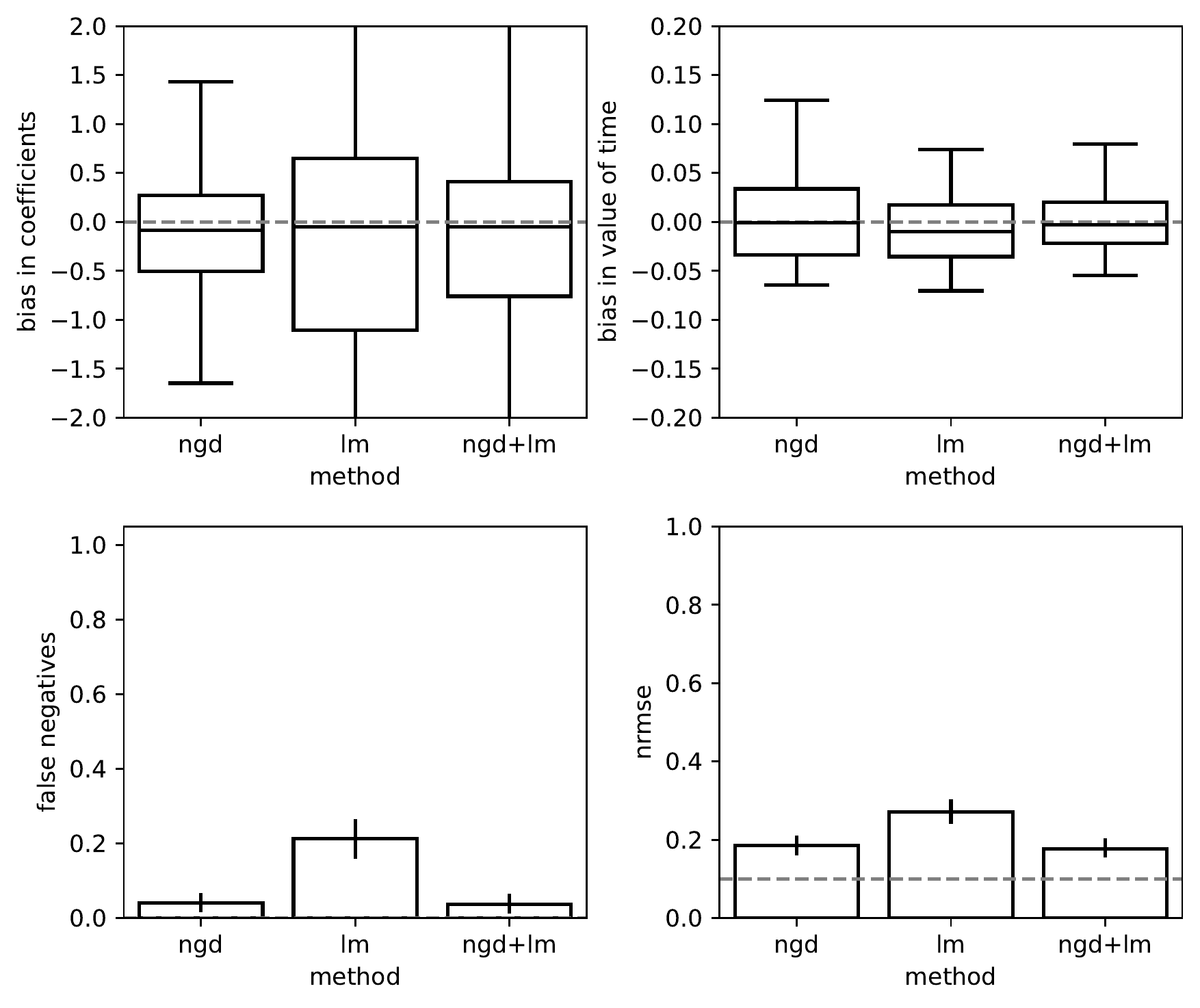}
	\caption{Endogenous travel times}
	\label{subfig:inference-congested-network-sioux-falls}
	
\end{subfigure}

\caption{Consistency in coefficients recovery}

\label{fig:inference-sioux-falls}
\end{figure}

The experiment results show that the estimates were roughly unbiased in all optimization setups and under exogenous and endogenous travel times. This result is interesting because the \NLLS assumptions prove only consistency but not unbiasedness of the \NLLS estimators in large samples. A potential explanation is related to the pseudo-convexity of the outer level objective and to the existence of a unique global optima in the neighborhood around the initial values for optimization. Furthermore, our evidence strongly supports the use of \NGD+\LM over the standalone application of first or second order optimization methods. The higher standard error of the estimates obtained with \LM unveils the problems of numerical stability of second order methods when the initial values for optimization are in regions where the objective function is not convex. The latter is also associated with the large proportion of false negatives reported by \LM in both experimental scenarios. In contrast, the standalone application of \NGD reports less than 5\% of false negatives in both scenarios, which equates to a statistical power of 95\%. This is a encouraging result given the reference threshold of 80\% typically used for experimental studies. 

In the endogenous case, we observe an increase of the difference between the \NRMSE and its expected value, given by the dashed bar shown in the bottom right plots of Figures \ref{subfig:inference-uncongested-network-sioux-falls} and \ref{subfig:inference-congested-network-sioux-falls}. Similarly, the bias of the estimated coefficients and of the value of time increase consistently in the three setup of optimization methods compared to the exogenous case. The computation of equilibria in the endogenous case may result into inaccurate approximations of the travel times at equilibria. From an econometric standpoint, this introduce measurement error in the travel time attribute and it can add bias in the parameter estimates. Another problem in the endogenous case is the significant increase of computational burden due to solving traffic equilibria at each replicate of the experiment. 

To ensure good convergence over replicates and that the assumptions for consistency of the \NLLS estimator are satisfied, the following experiments in this section are performed under the scenario of exogenous travel times only. The initial values for optimization of the utility function coefficients are also set with a uniform distribution centered at the true values of the coefficients and with a width of two. This choice of width allows to start the optimization in a region close of the ground truth values of the utility function coefficients and where the objective function is at best convex or at worst pseudo-convex. A tighter width makes difficult to assess the reduction of the objective function with different setups of optimization method. The significance level for all hypotheses tests is also set at $\alpha = 0.1$.

\subsubsection{Irrelevant attributes}  
\label{sssec:irrelevant-attributes-experiment}

The \NLLS theory suggests that the \NLLS estimator is asymptotically consistent. Therefore, the coefficients associated to irrelevant attributes included in the utility function should converge to zero and they should not impact the statistical inference on coefficients weighting relevant attributes. However, with finite samples, the inclusion of irrelevant attributes in the utility function can reduce the efficiency of coefficients of relevant attributes. As a consequence, the t-tests of other coefficients may decrease and this could artificially increase the number of false negatives, i.e. lower rejections of the null hypothesis for coefficients different than zero. 

Figure \ref{fig:irrelevant-attributes-inference} shows the impact of the inclusion of irrelevant attributes on the bias of the estimated coefficients and of the value of time, and on hypothesis testing. In contrast to the Monte Carlo experiments conducted in Section \ref{sssec:impact-congestion-statistical-inference}, the utility function used to generate synthetic traffic counts includes six irrelevant attributes. Each irrelevant attribute is generated as a standard Gaussian random variable and it is weighted by a coefficient equal to zero in the utility function. Hence, the hypothesis tests are expected to not reject the null for these coefficients. To quantify the power of the hypothesis tests to detect irrelevant attributes, we compute the proportion of false positives, namely, the proportion of times that the null hypotheses for coefficients weighting the set of irrelevant attributes are rejected. By classic statistical theory, p-values are uniformly distributed under the null hypothesis, hence, a significance level of $\alpha = 0.1$ should result into a 10\% of false positives. To analyze the impact of the inclusion of irrelevant attributes on the statistical inference of relevant attributes, we also report the proportion of false negatives.

\begin{figure}[H]
\centering

\includegraphics[width=0.7\columnwidth] {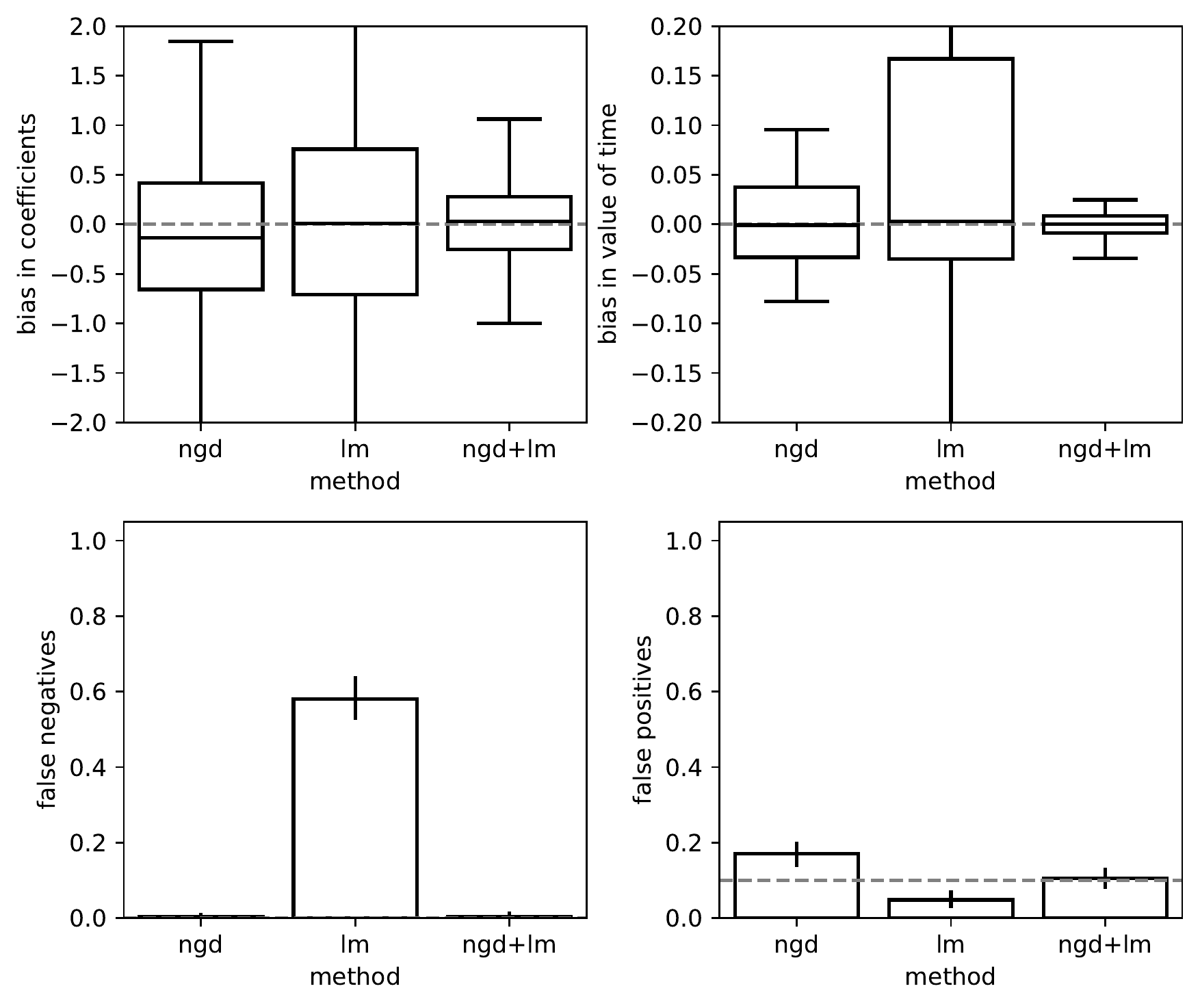}

\caption{Impact of irrelevant attributes on statistical inference}

\label{fig:irrelevant-attributes-inference}
\end{figure}

Similar to the results obtained in Section \ref{sssec:impact-congestion-statistical-inference}, the integration of \NGD and \LM significantly reduces the amount of false negatives and the bias of the estimated coefficients and of the value of time. The inclusion of irrelevant attributes in the utility function increases the amount of false negative but only in the case of \LM. Notably, the integration of \NGD and \LM perfectly matches the theoretical 10\% level of false positives that is expected at a significance level of $\alpha = 0.1$. Overall, these results reaffirm that the statistical inference is more reliable with the integration of first order and second order optimization methods. Given this evidence, the following experiments are conducted using \NGD+\LM only.

\subsubsection{Sensor coverage}
\label{sssec:sensor-coverage-experiment}

A lower sensor coverage reduces the sample size and this is expected to increase the standard error in the estimates of the utility function coefficients. As a consequence, t-tests may be overestimated and thus, the amount of false negative may increase. Figure \ref{fig:link-coverage-inference} shows the impact on hypothesis testing when varying the link coverage at three levels; 25\%, 50\% and 75\%. As expected, the amount of false negatives and the standard error of the coefficients decrease with the level of sensor coverage. Notably, for coverages of only 50\% (N = 38), the statistical power already surpasses the desired threshold of 80\%. For all coverage levels, the proportion of false positives does not surpass the theoretical level of 10\% and the \NRMSE closely matches its expected value at 10\%. 

\begin{figure}[H]
\centering
\includegraphics[width=0.7\columnwidth] {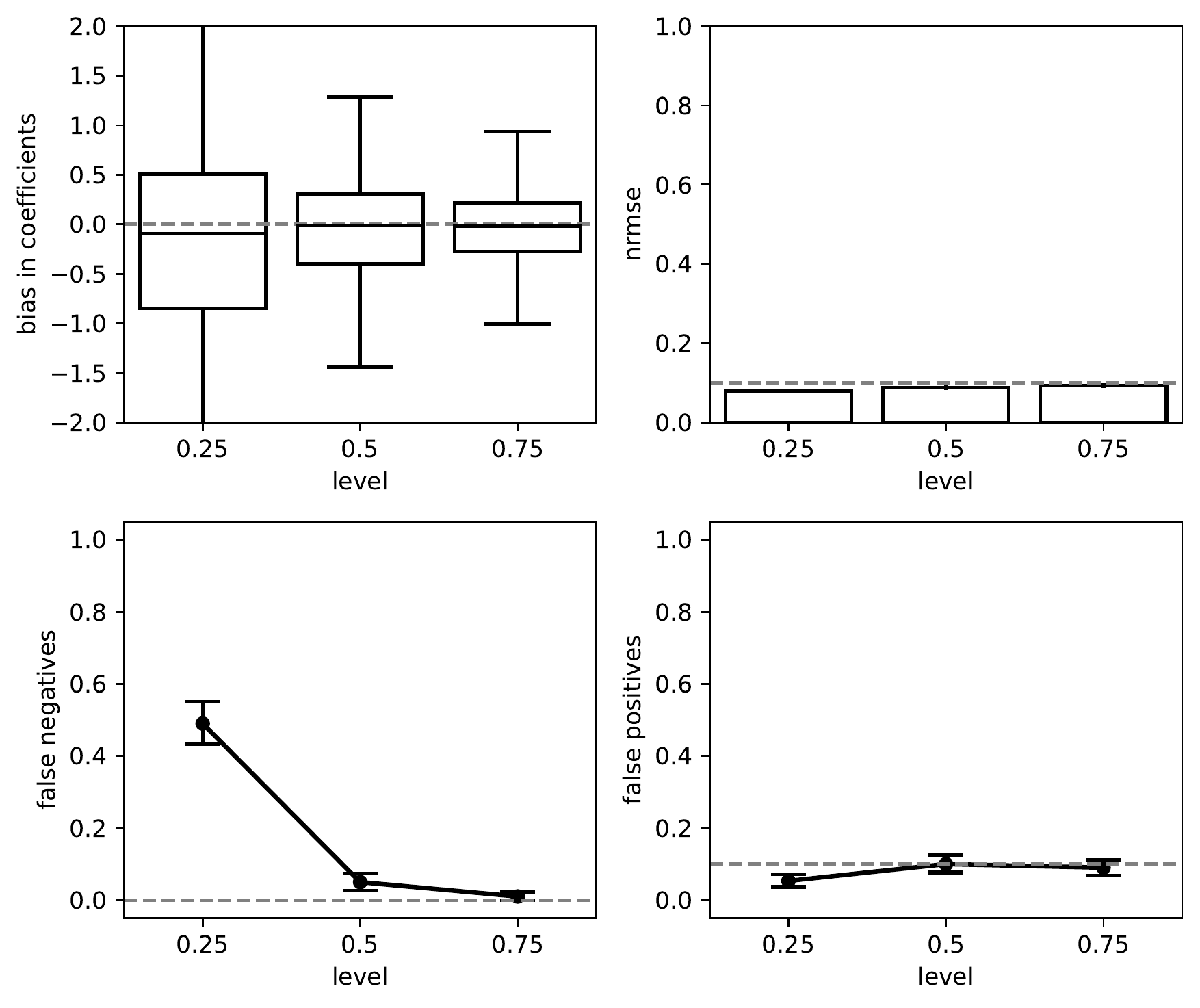}

\caption{Impact of sensor coverage on statistical inference}
\label{fig:link-coverage-inference}

\end{figure}

\subsubsection{Error in traffic count measurements}
\label{sssec:error-link-counts-experiment} 

The standard error of the \NLLS estimators is expected to increase when the variance of the random perturbation is higher. Thus, hypothesis testing should tend to reject less frequently the null for coefficients associated to relevant attributes and hence, the amount of false negatives could artificially increase. Figure \ref{fig:noisy-counts-inference} shows the impact of increasing the standard deviation of the Gaussian perturbation on statistical inference. The standard deviation is defined as a percentage of the average value of the synthetic traffic counts (see Section \ref{ssec:dgp}). As expected, the amount of false negatives and the standard error of the estimated coefficients increases with the level of variability in the random perturbation. The values of the \NRMSE closely match the error levels and the proportion of false positives does not surpass the theoretical level of 10\%. In real applications and in line with the sensor coverage experiment (Figure \ref{fig:link-coverage-inference}), a larger sample size should compensate for an increase in the variance of the random perturbation.  

\begin{figure}[H]
\centering
\includegraphics[width=0.7\columnwidth] {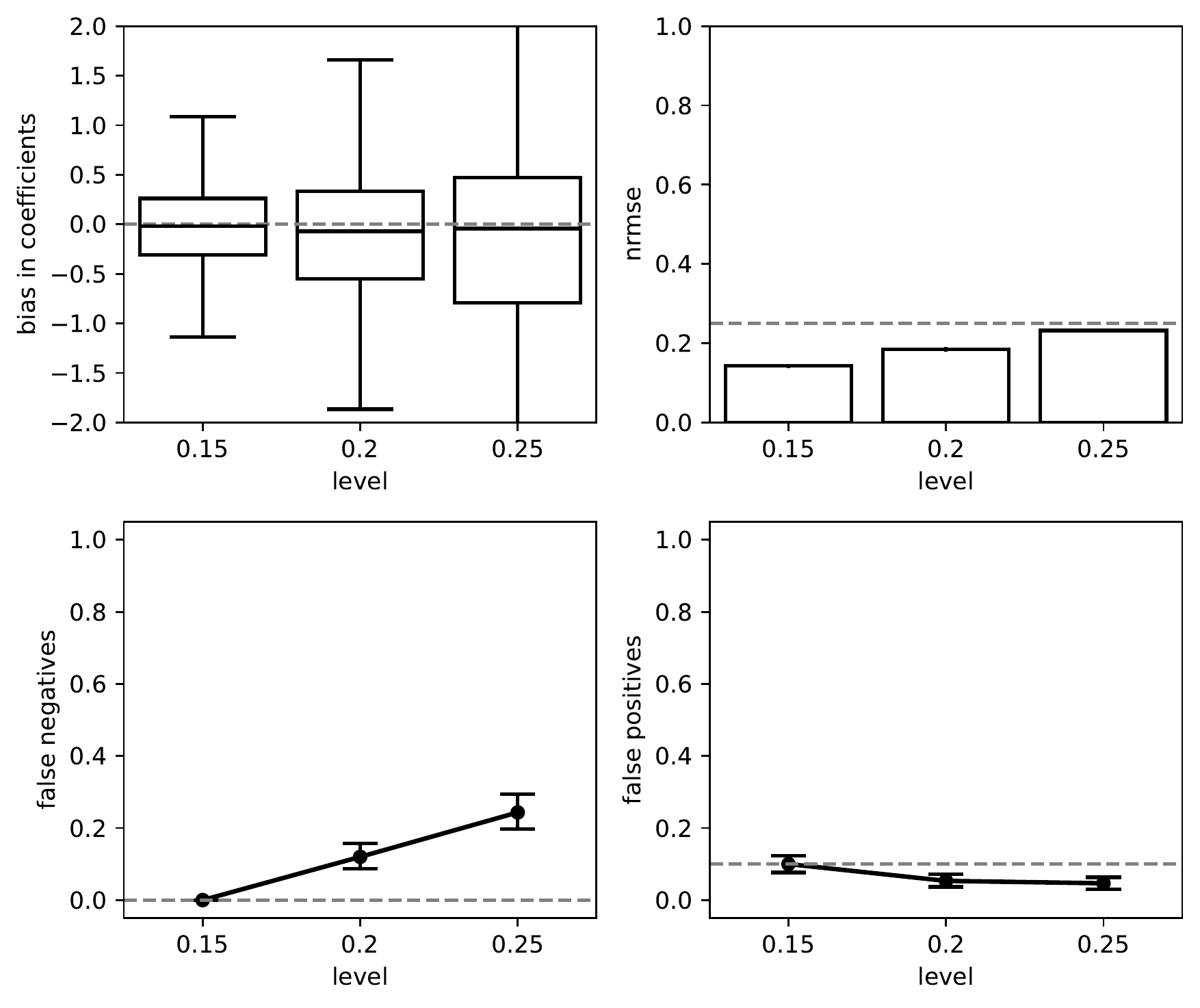}
\caption{Impact of noise in traffic counts on statistical inference}
\label{fig:noisy-counts-inference}
\end{figure}

\subsubsection{Non-deterministic O-D matrix}
\label{sssec:non-deterministic-od}

Our methodology assumes that the O-D matrix is exogenous and deterministic (Assumption \ref{assumption:od}, Section \ref{ssec:assumptions}). While it is feasible to have an accurate reference O-D matrix available, its cells may be subject to random perturbations. Any gap between the true and reference O-D matrix may add measurement error in the model and this could bias or distort the t-tests of the coefficients. Figure \ref{fig:noisy-od-inference} shows the impact of introducing different level of noise in the O-D matrix. The reference O-D matrix is assumed to be a Gaussian random variable with expected value $\vmu_{\mQ} \in \mathbb{R}^{V \times V}_{\geq 0}$ equal to the values of the true O-D matrix $\mQ$ and with standard deviation $\sigma_{\mQ} \in \mathbb{R}_+$ defined as a percentage of the average value of $\mQ$. To compare the impact of adding noise in either the traffic counts (Section \ref{sssec:error-link-counts-experiment}) or the O-D matrix, we set the same levels of the standard deviation of the random perturbation in both experiments. Surprisingly, our results show that the noise in the O-D matrix does not significantly impact statistical inference. Therefore, the standard error of the estimated coefficients and the \NRMSE remains invariant for all levels of random noise. 

\begin{figure}[H]
\centering
\includegraphics[width=0.7\columnwidth] {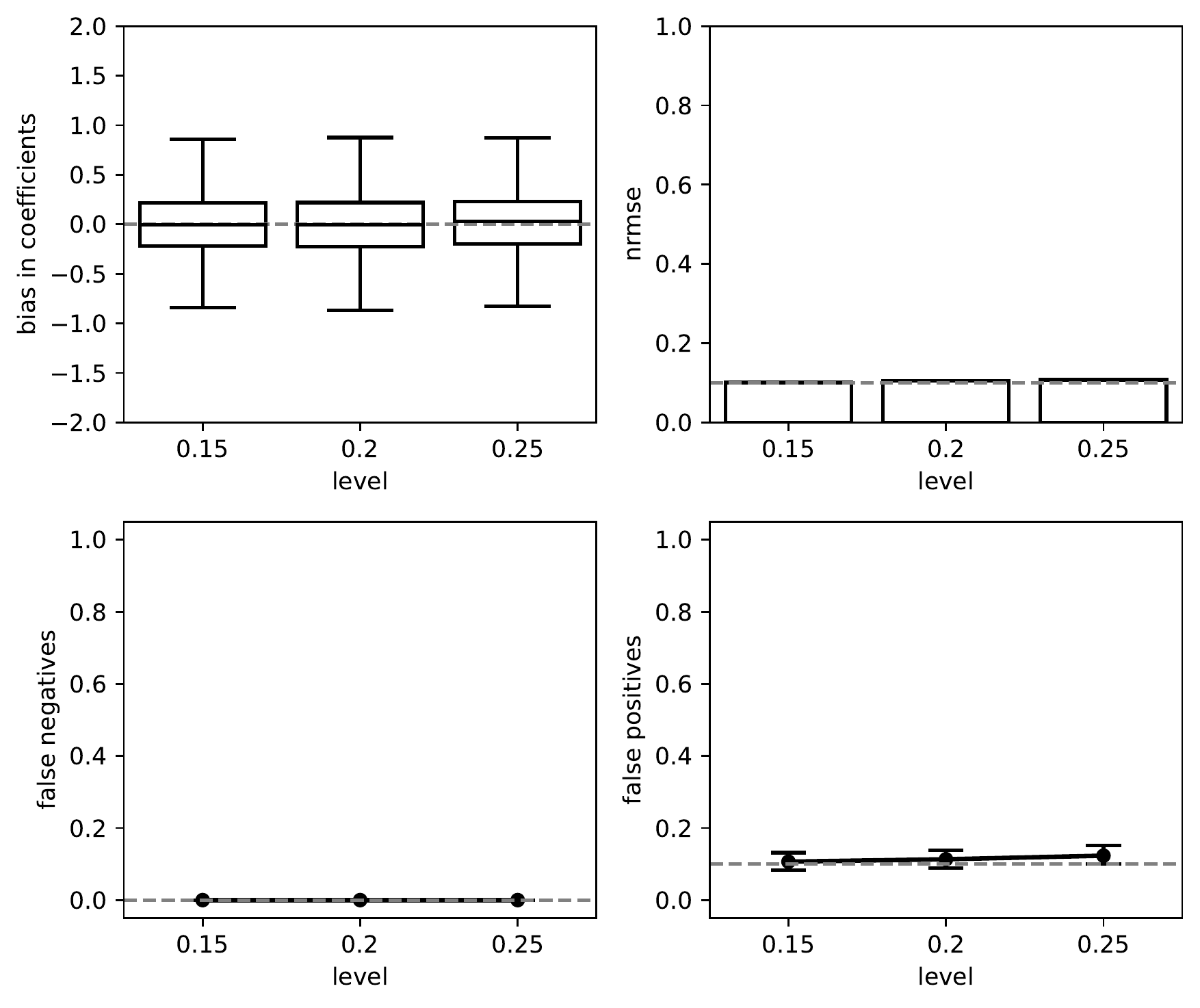}
\caption{Impact of noise in O-D matrix on statistical inference}
\label{fig:noisy-od-inference}
\end{figure}

\subsubsection{Ill scaled O-D matrix}
\label{sssec:scale-od}

An ill scaled O-D matrix is expected to increase the gap between predicted and observed traffic counts and this may induce an overestimation of the variance of the random error. As a consequence, the t-tests and the false negatives could increase. Figure \ref{fig:bad-scaled-od-inference} shows the impact of a ill-scaled the O-D matrix. The level represents the true value of the scale or factor that should multiply the reference O-D matrix to recover the true O-D matrix. Therefore, values lower or higher than 1 represent cases of overestimation or underestimation of the true O-D matrix, respectively. Compared to the previous experiments that assume a well-scaled O-D matrix, we observe that the \NRMSE is far from its expected value, which suggest that there is no a good convergence toward the ground truth values of the utility function coefficients. The estimated coefficients are biased upwards and downwards when the true O-D matrix is overestimated or underestimated, respectively. The latter is directly associated with the underestimation and overestimation of the false positives respect to the theoretical level of 10\%. The false negatives surpasses the 20\% when the true O-D matrix is underestimated, which significantly reduce the statistical power to detect the effect of relevant attributes in the utility function. 

\begin{figure}[H]
\centering
\includegraphics[width=0.7\columnwidth] {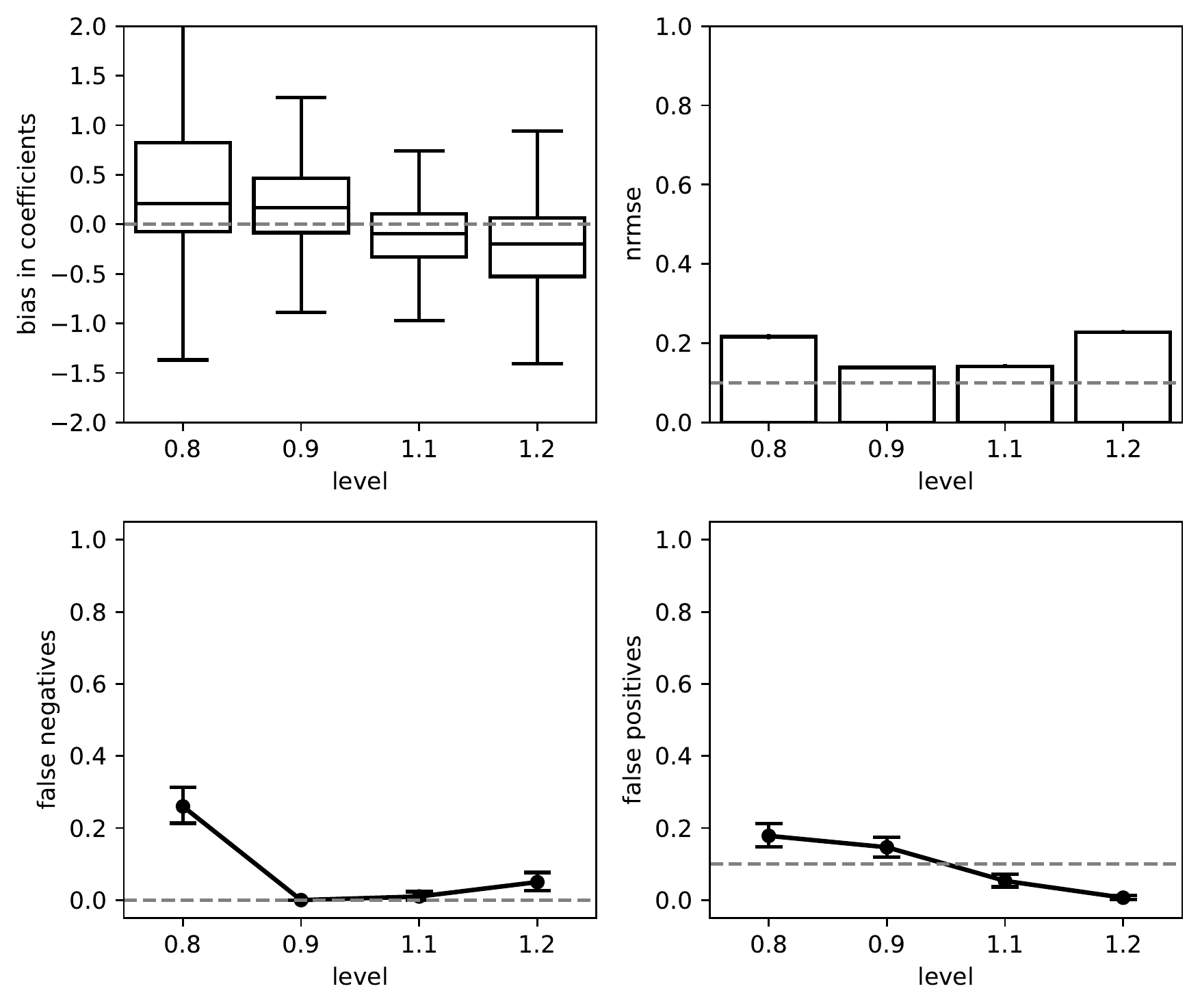}
\caption{Impact of a bad scaled O-D matrix on statistical inference}
\label{fig:bad-scaled-od-inference}
\end{figure}

\section{A large-scale network with multiple attributes: California SR-41 corridor}

The proposed methodology was also applied to a large scale transportation network located in the City of Fresno, California. This network primarily covers major roads and highways around the SR-41 corridor and it comprises 1789 nodes and 2413 links \citep{Ma2017, Ma2018}. The following sections describe the use of various data sources to estimate our model in the Fresno network. Subsequently, we describe the attributes of the utility function, the models' specifications, the estimation procedure and the results.  

\begin{figure}[h]
\centering
\includegraphics[width=0.5\columnwidth, trim= {0cm 0cm 0cm 0cm},clip ]{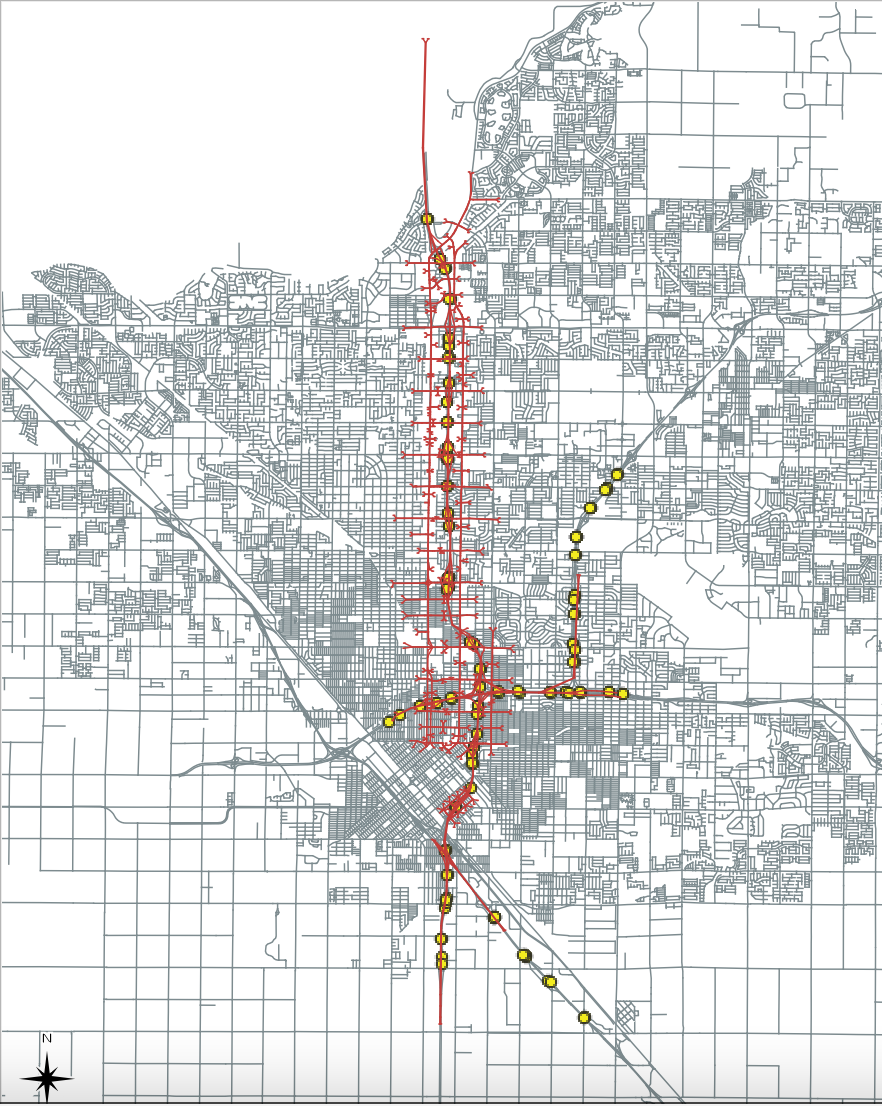}
\caption{SR-R1 corridor in Fresno, CA}
\label{fig:fresno-network}
\end{figure}

\subsection{Traffic counts}

The Caltrans Performance Measurement System \citep{PeMS} provides open source data with georeferenced stations that record information on traffic counts, speeds and travel times in major highways in California. Using geoprocessing tools, a total of 141 links of the Fresno network were matched to PeMS stations, which equates to a 5.8\% of sensor coverage. We only use traffic count data collected during the first Tuesday of October 2019 and October 2020 between 4pm and 5pm and which correspond to periods before and during COVID-19. The average traffic counts for the selected periods in 2019 and 2020 are 2213.6 and 2113.3 vehicles per hour, respectively. 

\subsection{Travel demand}
\label{ssec:travel-demand}

A dynamic O-D demand matrix calibrated with support of the PARAMICS estimator tool for the SR-R1 corridor is adopted in our study \citep{Ma2017,Zhang2008, Liu2006}. This O-D demand is provided with a 15 minutes resolution and for the period between 4PM and 6PM on a typical weekday. To have consistency in the temporal resolutions of the O-D demand and the traffic counts data, both data sources are aggregated in a one hour window starting at 4pm. The aggregation in vehicles per hour is also consistent with the temporal resolution of the link performance functions. 

Figures \ref{fig:cumulative-demand-fresno} show the cumulative demand respect to the number of O-D pairs. There are 6970 O-D pairs that reported trips, which gives a total 66,266.3 trips. The dashed bars in the figures suggest that approximately the 30\% of the O-D pairs with the highest demand covers a 85\% of all trips. To account for the reduction in travel demand in the selected periods of 2019 and 2020, we scale the O-D matrix in 2019 by a uniform factor of 0.9546, which is equal to ratio between the average traffic counts in each year (2113.3/2213.6). Note that this factor is similar but more conservative than the ratio between the traffic volumes of October 2020 and October 2019 in all roads and streets in the United States \citep{FHA2020}. 


\begin{figure}[h]
\centering
\includegraphics[width=0.9\columnwidth, trim= {0cm 0cm 0cm 0cm},clip ]{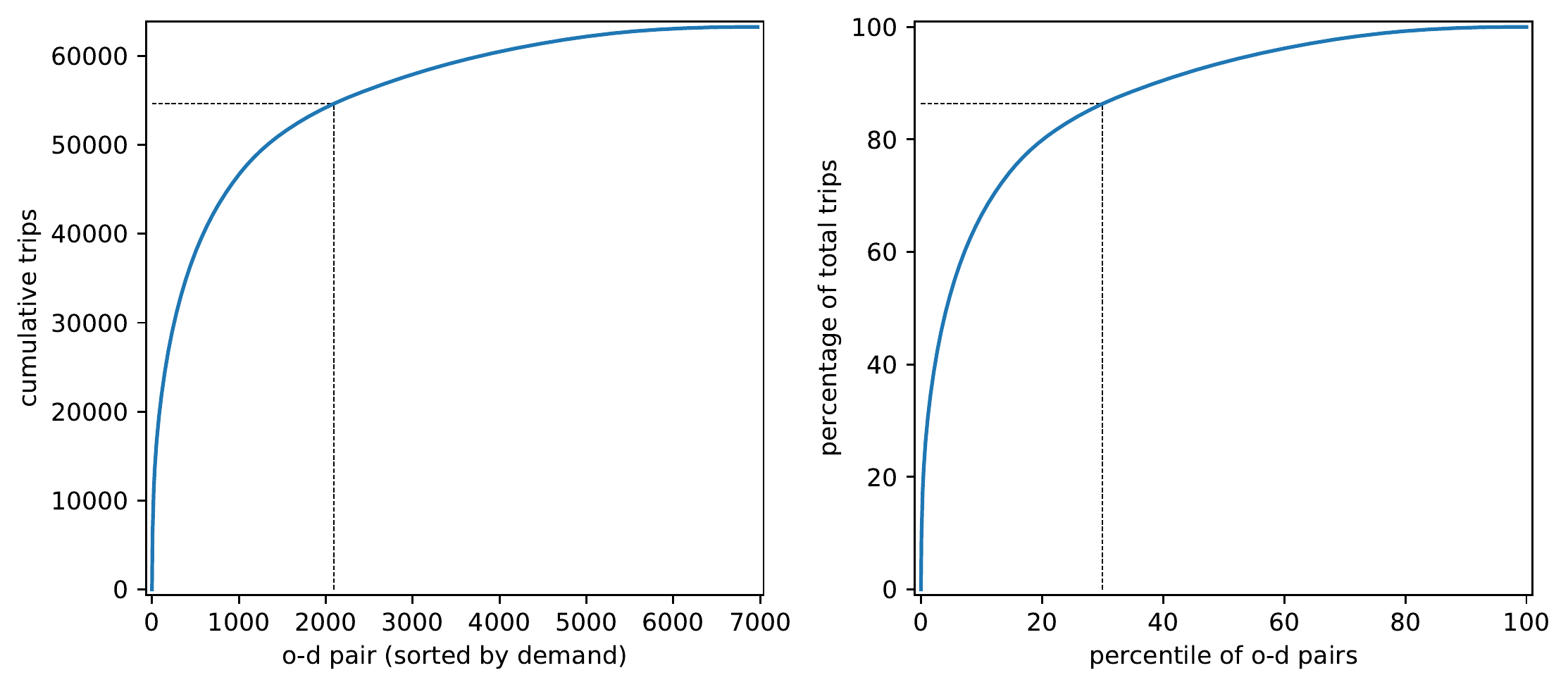}
\caption{Cumulative travel demand in Fresno, CA}
\label{fig:cumulative-demand-fresno}
\end{figure}

\subsection{System level data}
\label{ssec:data-sources}

Below are the sources of system level data used to compute the attributes of the travelers' utility function in every link of the Fresno network:

\begin{itemize}[leftmargin=*]

\item  Speed and travel times

INRIX traffic time and traffic speed data have been successfully used in previous transportation studies an it is considered a reliable data source for our analyses \citep{Yao2020}. Our raw data includes a shapefile with a collection of georeferenced line segments representing the major roads in Fresno, CA. It also includes \texttt{csv} files with traffic speeds for every line segment and with a 5 minutes resolution. Attributes in the dataset include a unique identifier for each line segment, a time stamp, observed speed, average speed and reference speed. Using the length of each line segment, the speed attributes were transformed into travel time attributes in minute units. The line segments of the INRIX shapefile and the links of the Fresno network are matched according to the orientation and proximity of the streets. For the 15.5\% links that were not successfully matched, we imputed the mean value of the attributes that were matched in the remaining links.

\item Traffic incidents

The Statewide Integrated Traffic Records System (SWITRS) is one of the two official sources of traffic incidents data in California \citep{Waetjen2021}. SWITRS is managed by the California Highway Patrol (CHP) and includes post-processed and georeferenced data on incidents causing human injury or death. A codebook of the dataset and the incidents reported between 2016 and 2021 are publicly available in \citet{SWITRS2021} and \cite{Kaggle2021}. A buffer of 50 feet is used to match the links in the Fresno network with the location of the incidents. In 2019 and 2020, 1601 and 2020 incidents are matched to 276 and 382 links of the Fresno network, respectively. 

\item Streets characteristics

Street characteristics are hypothesized to be a predictor of travelers' route choices. Using an open repository maintained by the \citet{GISDivisionCityofFresno2021}, we processed shapefiles with the georeferenced positions of bus stops and streets intersections in Fresno. A buffer of 50 feet is used to match the links in the Fresno network with the location of the stops and intersections. 362 and 2115 bus stops and streets intersections are matched to 300 and 1139 links, respectively. 

\item Socio-demographics from Census data

A shapefile with the most recent selected sociodemographic data at the block level was retrieved from \citet{USCensusBureau2019}. We only use the layers with data about the proportions of population by gender, age and income level at each block. There are other layers with relevant features to explain route choice preference, such as one related to commuting patterns, but this data is not available for the city of Fresno or it is only available at the tract and county levels. 

\end{itemize}

\subsection{Attributes of the utility function}
\label{ssec:attributes-utility-function}

\textit{Travel time} is the only endogenous attribute in the utility function and it is equal to the output of the links' performance functions. The free flow travel times defined in the link performance functions is obtained from INRIX data. The speed attribute used to compute free flow travel times corresponds to the speed driven on a road when it is wide open and it does not necessarily match the legal speed limit. Using the system level data described in Section \ref{ssec:data-sources}, we also compute the following set of exogenous attributes for every link of the Fresno network:

\begin{itemize}[leftmargin=*]
\item \textit{Std. Speed}: Standard deviation of speed [miles/hour]
\item \textit{Incidents}: Total incidents in the current year
\item \textit{Intersections}: Number of streets intersections
\item \textit{Bus stops}: Number of bus stops
\item \textit{Median income}: Median household income in the U.S. Census block [USD/year]
\end{itemize}

\textit{Std. Speed} is computed using historical data provided by INRIX for the period between 4pm and 5pm on a weekday and it is assumed to negatively impact the utility of choosing a path. There are alternative ways to capture the effect of travel time variability on travelers' route choices, such as the standard deviation or coefficient of variation of travel time. The main advantage of the standard deviation of speed is that it provides a measure of time variability that is independent of the length of the link segments. The rationale to aggregate the number of incidents in the current year is to have a proxy of the drivers' perception on the road safety of a link segment. \textit{Median income} is used as a proxy of the perception of the level of crime in surrounding neighborhoods and it is hypothesized to be positively associated with the utility of choosing a path. A larger number of bus stops or streets intersections in a link segment should encourage drivers to use alternative paths. Thus, all attributes, except for \textit{Median income}, are expected to have a negative coefficient in the utility function.

To test non-linear effects of the attributes on the utility function, we also binarized the exogenous attributes as follows. \textit{Reliable speed} and \textit{Low income} take the value 1 if the values of \textit{Std. speed} and \textit{Median income} are below the 25th percentile of their distributions, respectively, and 0 otherwise. \textit{Bus stop}, \textit{Intersection} and \textit{Incident} take the value 1 if \textit{Bus stops}, \textit{Intersections} and \textit{Incidents} are greater than zero, respectively, and 0 otherwise. Under this new specification of the attributes, all coefficients of the utility function, except for the coefficient weighting \textit{Reliable speed}, are expected to be negative.

Note that the network contains 656 links (38.8\%) that are centroid connectors. These connectors usually do not have physical counterparts in the real network \citep{SeanQian2012} and they should not impact travelers' route choices. Thus, we set the values of the attributes in these links to zero, regardless if the attributes are continuous or binary. The path size correction factor is set to its default value of one in all model specifications.

\subsection{Descriptive statistics}
\label{ssec:descriptive-statistics}

Figure \ref{fig:features-correlations-years} shows the correlations among traffic counts and system level attributes of the Fresno network during the first Tuesdays of October 2019 and 2020. We include free flow speed instead of free flow travel times which normalize the attribute by the link length. The Pearson correlations shown in the top row of the plots indicate that the free flow speed is negatively correlated with the number of intersections and bus stops in the two selected periods. Also, as expected, the number of incidents is positively associated with the free flow speed. An interesting result is the negative correlation between median income and free flow speed. The latter is probably explained by the higher rate of motorization in high income areas and which should be associated to lower speeds. The standard deviation of speed shows a positive correlation with free flow speed but only during 2020. The sign of the correlation between traffic flows and free flow speed is unclear. Overall, the correlational analysis shows sensitive associations between the attributes of the utility function and it suggests that the spatial matching of attributes conducted in \ref{ssec:data-sources} was reasonably accurate.

\begin{figure}[h]
\centering
\includegraphics[width=0.9\columnwidth, trim= {0cm 0cm 0cm 0cm},clip ]{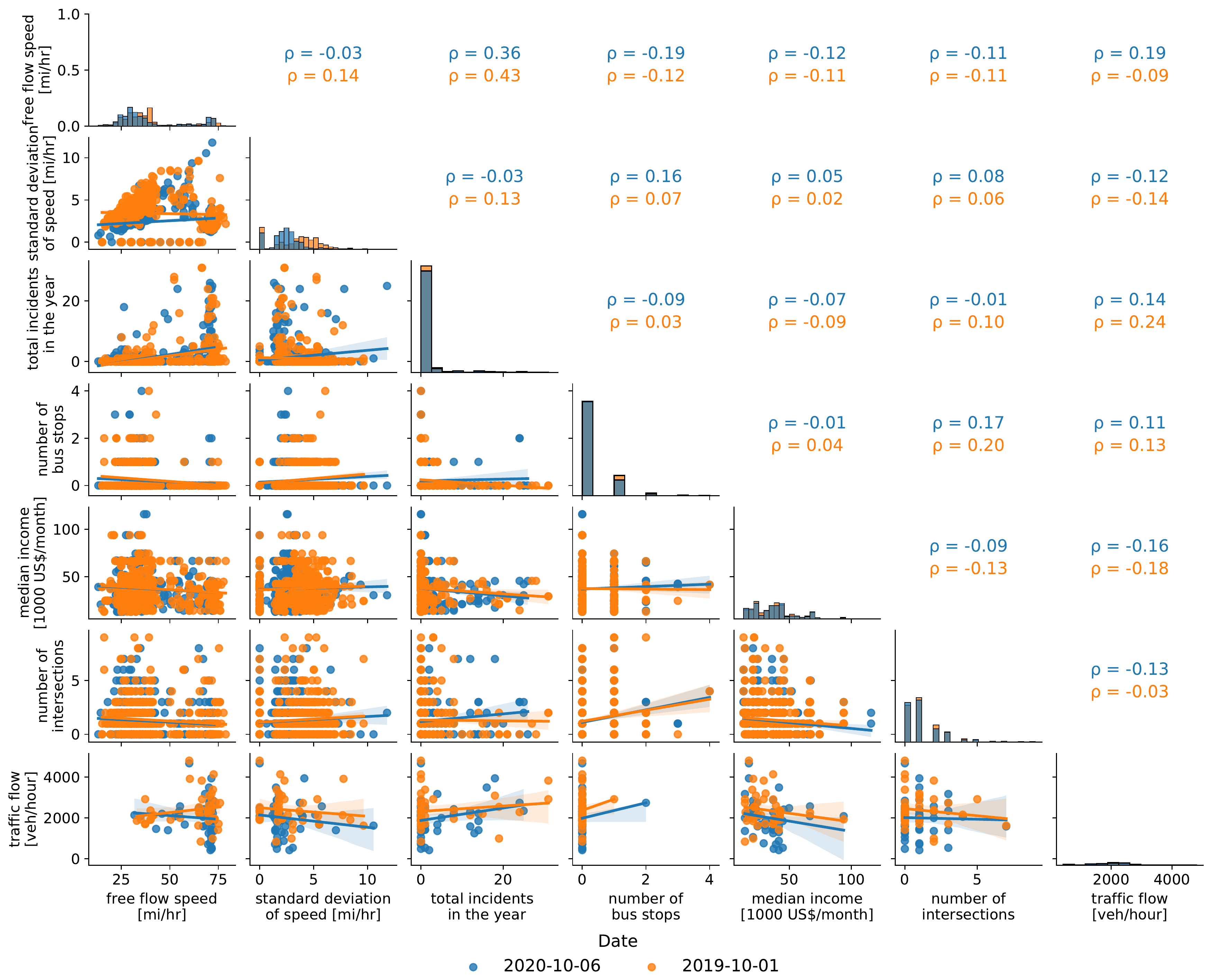}
\caption{Correlations among system level attributes in Fresno, CA}
\label{fig:features-correlations-years}
\end{figure}

\subsection{Model specifications}
\label{ssec:model-specifications}

We estimate three model specifications using the data collected in each time period. Following the state of the practice in previous literature, we first estimate a \textit{Baseline model} including \textit{travel time} only. We also estimate a \textit{Full model} including the exogenous attributes in a continuous scale; \textit{Std. Speed}, \textit{Median income}, \textit{Incidents}, \textit{Intersections}, \textit{Bus stops}. Finally, a \textit{Binarized model} includes the exogenous attributes in a binary scale; \textit{Reliable speed}, \textit{Low income}, \textit{Incident}, \textit{Intersection}, \textit{Bus stop} (Section \ref{ssec:attributes-utility-function}). To assess the gain in explanatory and predictive power of the \textit{Full model} and \textit{Binarized model} compared the \textit{Baseline model}, these models also include \textit{Travel time}. By assumption, the utility function is defined as a linear weight between attributes and their respective coefficients. \textit{Travel time}, the only endogenous attribute in all models, is iteratively updated during the bilevel optimization. 

\subsection{Estimation}
\label{ssec:estimation}
To prevent that numerical issues associated to the scale of the attributes impact the quality of the statistical inference, we normalize each exogenous attribute, including the free flow travel time defined in the link performance functions, by their maximum values. Because all these attributes are non-negative, their new ranges are between 0 and 1. Section \ref{appendix:ssec:descriptive statistics} includes figures with descriptive statistics of the attributes after this normalization and that exclude links that are centroid connectors. 

The models are estimated on Amazon Web Services on two \texttt{t2.2xlarge} instances equipped with Intel Xeon CPU 3.3GHz x 8, 32 GB RAM \citep{AWS2022}. The three model specifications are run in parallel using the data collected in either 2019 or 2020. Because \SUE-\logit is path-based, the computational and memory complexity is a function of the number of paths. The CPU processor and RAM memory of the \texttt{t2.2xlarge} instance allow to handle 33,430 of paths. In contrast, a \texttt{t2.xlarge} instance with Intel Xeon CPU 3.3GHz x 4, 16 GB RAM, handles 22,015 paths without reporting memory overflow. The memory bottleneck of the algorithm is on the computation of the analytical gradients performed to update the utility function coefficients during the outer level optimization. The computation time of each notebook is about 6 hours.

All models are estimated using \NGD in the no-refined stage with a learning rate of $\eta_1 = 0.5$, and \LM  in the refined stage, respectively. Each stage performs 10 iterations of the optimization methods. The starting points for optimization in the non-refined stage are set to zero for the three model specifications. The best estimate of the utility function coefficients in the non-refined stage is chosen as the starting point for optimization in the refined-stage. The coefficients are projected to zero during the iterations of the bilevel optimization if their sign is not consistent with our expectation  (Section \ref{ssec:attributes-utility-function}). If the coefficients of the attributes are zero at the end of the non-refined stage, they are excluded for the refined stage. This strategy is used to reduce computational burden and to accelerate the convergence of the optimization algorithm to a local optimal solution. 

\subsection{Path generation}
\label{ssec:path-generation}

We generate an initial set of 13,905 paths using the two shortest paths between the 6970 O-D pairs that report trips (Figure \ref{fig:cumulative-demand-fresno}, Section \ref{ssec:travel-demand}). The initial path set of the three model specifications are the same. The total number of paths is odd because some O-D pairs are only connected by a single path. The column generation method described in Step 1, Algorithm \ref{alg:inner-level-optimization}, \ref{appendix:ssec:implementation-inner-level-optimization} is used to update the path sets during the bilevel optimization. To guide the paths exploration, new paths are generated in the 30\% of O-D pairs with the highest demand  only, which covers approximately 85\% of the total trips in the Fresno network (Figure \ref{fig:cumulative-demand-fresno}, Section \ref{ssec:travel-demand}). Based on the current estimate of the utility function coefficients and the link attributes, the column generation step generates the 40 shortest paths among 3\% of the OD pairs at each iteration of the bilevel optimization. The selected O-D pairs at each iteration changes sequentially according to their level of demand and the process continue until the target of 30\% of OD-pairs is reached in the last iteration. The sequential selection of O-D pairs reduces the computational burden of generating new paths in all O-D pairs and at every iteration, without compromising significantly path exploration. Subsequently, the inner level optimization algorithm solves \SUE-\logit and the paths utilities are updated according to the new travel times at equilibria. Finally, it selects the 10 paths with the highest utility for every O-D pair.

\subsection{Indicators for model comparison}
\label{ssec:indicators-model-comparison}

For each model specification and time period, we compute the value of the objective function (i.e. sum of squared errors), the root mean squared error (RMSE), the normalized RMSE (NRMSE), the F-test (Section \ref{ssec:hypothesis-testing}) and the adjusted pseudo $R^2$. The F-test compares the sum of squared errors of the model with its null version where all coefficients of the utility function were set to zero and that represents an scenario where travelers' make equilikely choices among paths. F-tests with p-values lower than $\alpha$ in rejects the null hypothesis that two models are statistically equivalent at the $1-\alpha$ \% level of confidence and hence, this is evidence that supports the selection of the augmented model. 

Our adjusted pseudo $R^2$ is analogous of the McFadden Pseudo-$R^2$ \citep{McFadden1973} and it is adjusted by the number of coefficients of the model \citep{Ben-Akiva1985}. It is defined as 1 minus the ratio between the SSE of a model minus the number of coefficients and the SSE of the model with all coefficients set to zero. Since traffic counts are an aggregate of individual path choices, the adjusted pseudo $R^2$ of our model is directly comparable with those values obtained from discrete choice models. This indicator is an absolute measure of the predictive ability of discrete choice models and it is useful to generate benchmark values against which researchers can evaluate \citep{Parady2021}. It tends to be considerably lower than the $R^2$ index used in ordinary regression analysis, and values of 0.2 to 0.4 represent an excellent fit \citep{McFadden1977}. 
 
\subsection{Results}

Table \ref{table:fresno-estimation} shows the estimation results obtained with traffic count data collected during the first Tuesday of October 2019 (before COVID-19) and October 2020 (during COVID-19) and for three model specifications (Section \ref{ssec:model-specifications}). Besides the point estimates and t-tests of each coefficient of the utility function, the table also includes multiple indicators for model comparison (Section \ref{ssec:indicators-model-comparison}). Figure \ref{fig:distribution-errors-fresno} shows histograms with the distribution of errors obtained after performing the non-refined and refined stages of the optimization. Figures \ref{fig:convergence-baseline-and-full-models-fresno} and \ref{fig:convergence-baseline-and-feature-engineering-models-fresno}, \ref{appendix:ssec:estimation-results} show a comparison of the convergence of the full and binarized models against the baseline model. The top plots of the figures shows the number of paths that are generated during column generation and that are added after performing the paths selection step during the inner level optimization. The significant decrease in the number of paths added in the bilevel iterations of the refined stage suggests that the path exploration over the selected set of O-D pairs in the non-refined stage was reasonably exhaustive.

\begin{table}[h] 
\centering 

\caption{Point estimates and summary statistics of models fitted with data collected between 4pm and 5pm during the first Tuesdays of October 2019 and October 2020 in Fresno, CA} 
\label{table:fresno-estimation}

\vspace{-0.2cm}

\begin{adjustbox}{width=\textwidth}  
	
	\begin{threeparttable}
		
		\begin{tabular}{@{\extracolsep{1pt}}lcccccc} 
			\\[-1.8ex]\hline 
			\hline \\[-1.8ex] 
			& \multicolumn{3}{c}{\begin{tabular}{c} First Tuesday of October 2019 (Before COVID-19) \end{tabular}} 
			& \multicolumn{3}{c}{\begin{tabular}{c} First Tuesday of October 2020 (During COVID-19) \end{tabular}} \\
			\cmidrule(lr){2-5} 
			\cmidrule(ll){5-7}
			\multicolumn{1}{l}{\begin{tabular}{l} \vspace{0.2cm} \hspace{-0.5cm}\textbf{Attribute  (t-test)} \end{tabular}} 
			& Baseline model & Full model & Binarized model & Baseline model & Full model & Binarized model \\ 
			\hline \\[-1.8ex] 
			Travel time  $(\theta_{t})$ 
			& $-$2.000$^{**}$ ($-$2.5) & $-$1.904$^{**}$ ($-$2.3)  & $-$0.760$^{***}$ ($-$3.1)  
			& $-$2.500$^{**}$ ($-$2.1) & $-$1.889$^{*}$ ($-$1.9)  & $-$0.458$^{*}$ ($-$1.8)  \\ 
			Std. speed  & $\minus$ & 0.000 (0.0)  & $\minus$ & $\minus$ & 0.000 (0.0) & $\minus$  \\ 
			Incidents  & $\minus$ & $-$2.409$^{**}$ ($-$2.2) & $\minus$  & $\minus$ & $-$2.597 ($-$1.3) & $\minus$\\ 
			Median income  & $\minus$  & 0.864$^{**}$ (2.0) & $\minus$ & $\minus$ & $-$0.447 ($-$0.8)  & $\minus$ \\ 
			Intersections  &  $\minus$ & 0.000 (0.0) & $\minus$ &  $\minus$ & $-$0.708 ($-$0.7) & $\minus$   \\ 
			Bus stops & $\minus$   & 0.000 (0.0)  & $\minus$ & $\minus$   & 0.000 (0.0)  & $\minus$  \\ 
			Reliable speed  & $\minus$ & $\minus$  & 0.000 (0.0) & $\minus$ & $\minus$   & 0.000 (0.0)  \\ 
			Incident  & $\minus$ & $\minus$ & $-$0.118 ($-$0.7)  & $\minus$ & $\minus$ & $-$0.249 ($-$1.5)\\  
			Low income  & $\minus$  & $\minus$ & 0.000 (0.0) & $\minus$ & $\minus$   & 0.000 (0.0) \\  
			Intersection  & $\minus$ & $\minus$  & 0.000 (0.0) & $\minus$ & $\minus$ & 0.000 (0.0)   \\ 
			Bus stop & $\minus$ & $\minus$  & 0.000 (0.0) & $\minus$   & $\minus$   & 0.000 (0.0)  \\ 
			\hline \\[-1.8ex] 		
			Obs. (coverage) & 141 (5.8\%)  & 141 (5.8\%) & 141 (5.8\%) & 141 (5.8\%)  & 141 (5.8\%) & 141 (5.8\%) \\ 
             Initial objective & 278,924,807 & 278,924,807 & 278,924,807  & 273,069,736 &  273,069,736 & 273,069,736 \\
			Objective function & 192,568,627 & 168,917,573 & 186,985,714  & 188,694,982 &  158,156,498 & 178,166,886 \\
			RMSE & 1168.6 & 1094.5 & 1151.6  & 1094.5 &  1059.1 & 1124.1 \\ 
   			NRMSE & 0.528 & 0.494 & 0.520  & 0.547 &  0.501 & 0.532 \\
            Adjusted pseudo $R^2$ & 0.311 & 0.395 & 0.331  & 0.316 &  0.422 & 0.345 \\
			F-test (p-value) & 71.954$^{***}$ (0.000) & 11.224$^{**}$ (0.000) & 8.047$^{***}$ (0.000) & 99.482$^{***}$ (0.000) & 4.157$^{***}$ (0.002) & 9.166$^{***}$ (0.000)\\
			\hline 
			\hline \\[-1.8ex] 
			
		\end{tabular}

		\begin{tablenotes}
			\vspace{-0.2\baselineskip}
			\tiny
			\scriptsize
			\item Note: Significance levels: $^{*}$p$<$0.1; $^{**}$p$<$0.05; $^{***}$p$<$0.01
		\end{tablenotes}

	\end{threeparttable}
\end{adjustbox}

\end{table}

\begin{figure}[h]
\centering


\begin{subfigure}[t]{0.3\columnwidth}
	\centering
	\includegraphics[width=\columnwidth, trim= {0cm 2cm 0cm 0cm},clip]{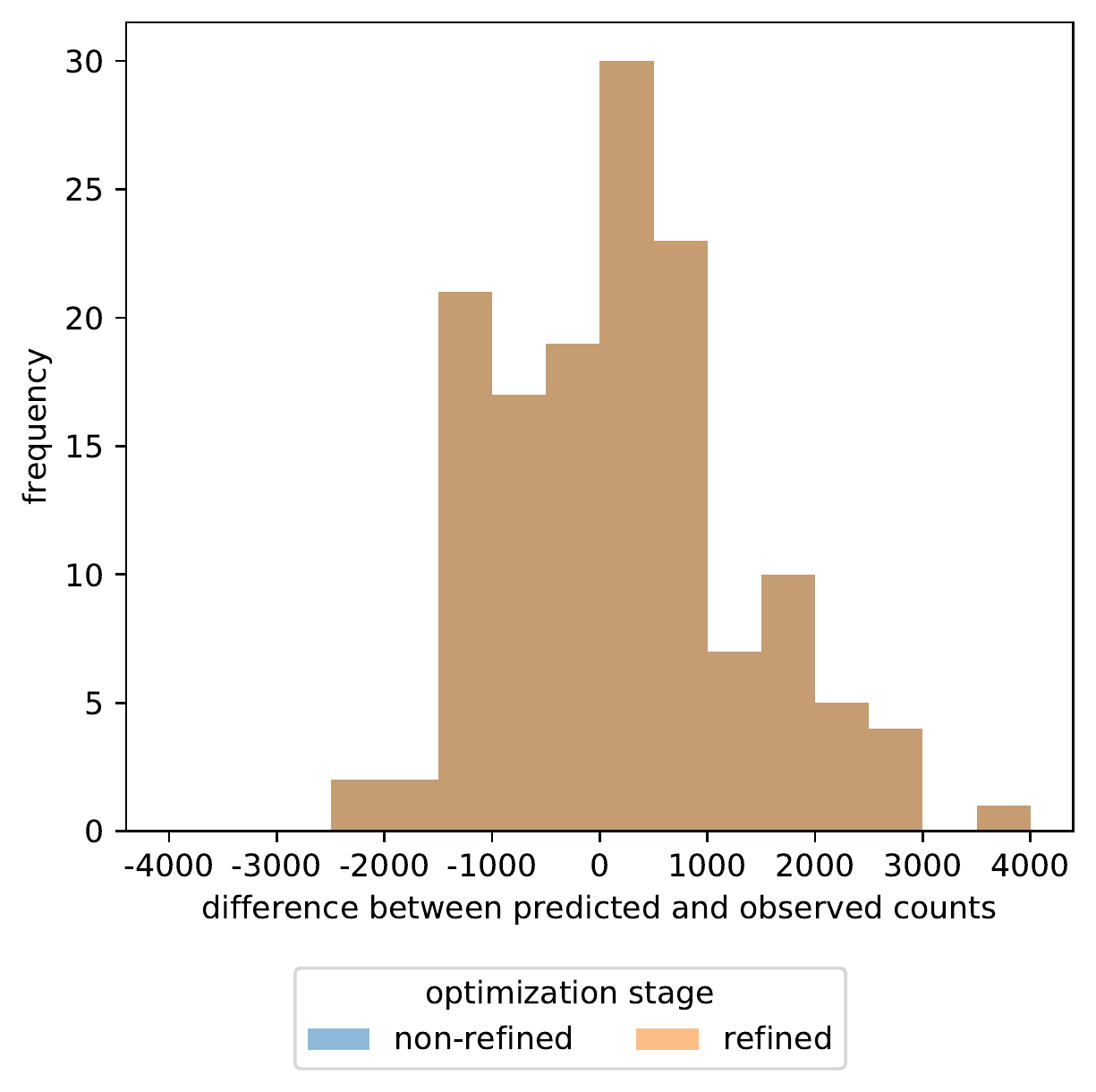}
	\caption{Baseline model 2019}
	\label{subfig:distribution-errors-baseline-model-2019}
\end{subfigure}
\begin{subfigure}[t]{0.3\columnwidth}
	\centering
	\includegraphics[width=\columnwidth, trim= {0cm 2cm 0cm 0cm},clip]{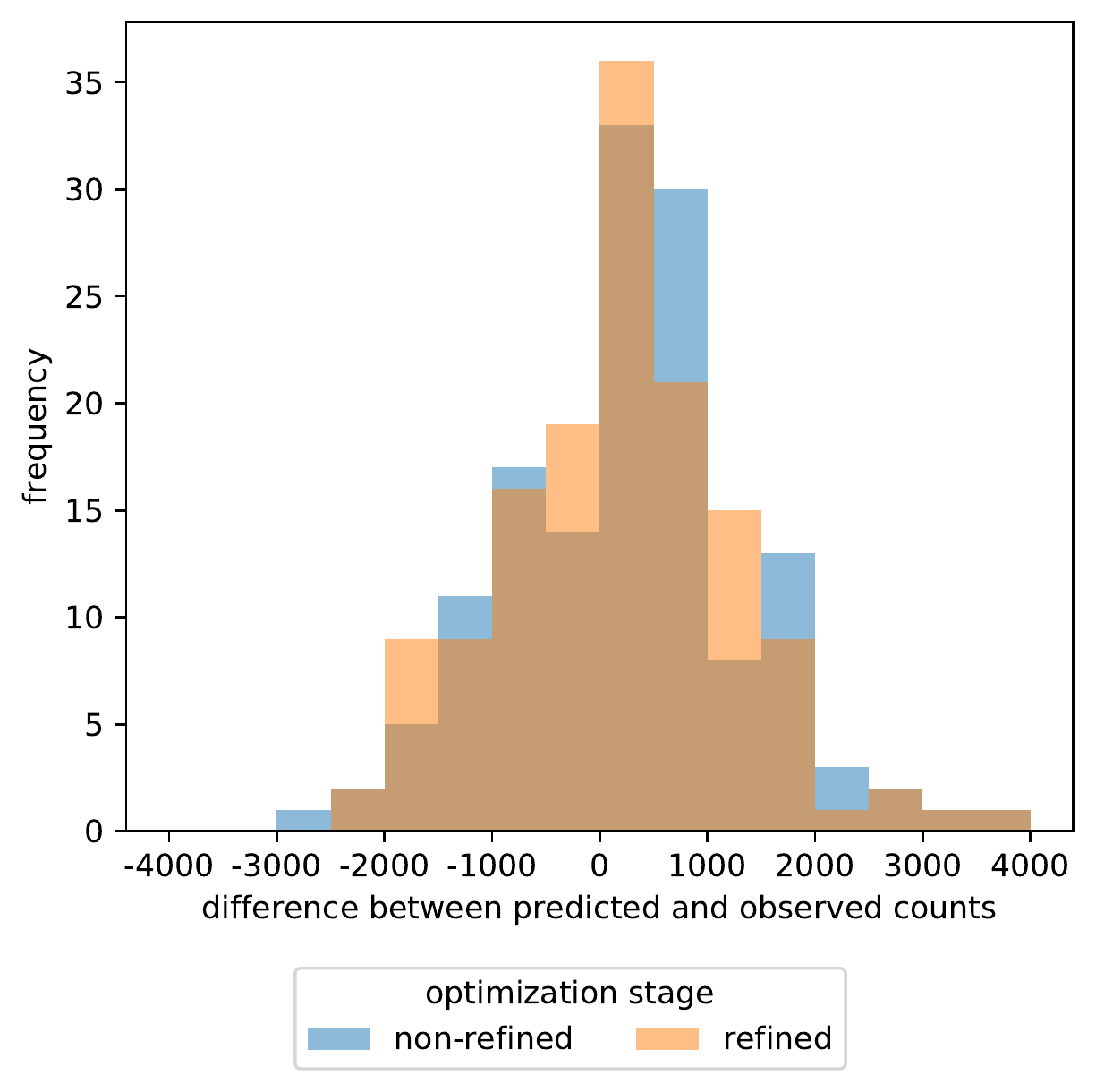}
	\caption{Full model 2019}
	\label{subfig:distribution-errors-full-model-2019}
	
\end{subfigure}
\begin{subfigure}[t]{0.3\columnwidth}
	\centering
	\includegraphics[width=\columnwidth, trim= {0cm 2cm 0cm 0cm},clip]{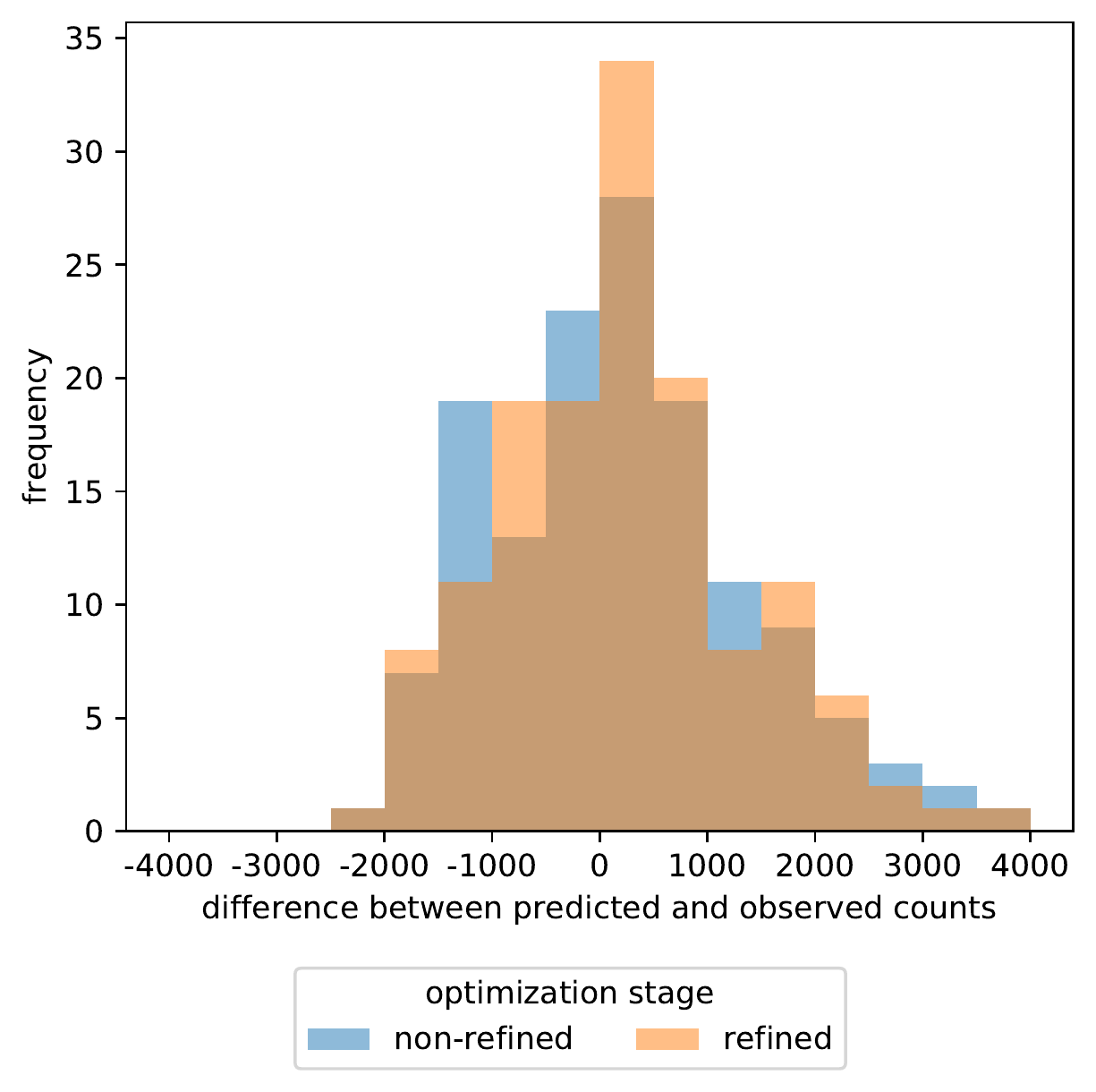}
	\caption{Binarized model 2019}
	\label{subfig:distribution-errors-feature-engineering-model-2019}
	
\end{subfigure}

\begin{subfigure}[t]{0.3\columnwidth}
	\vskip 0pt
	\centering
	\includegraphics[width=\columnwidth, trim= {0cm 2cm 0cm 0cm},clip]{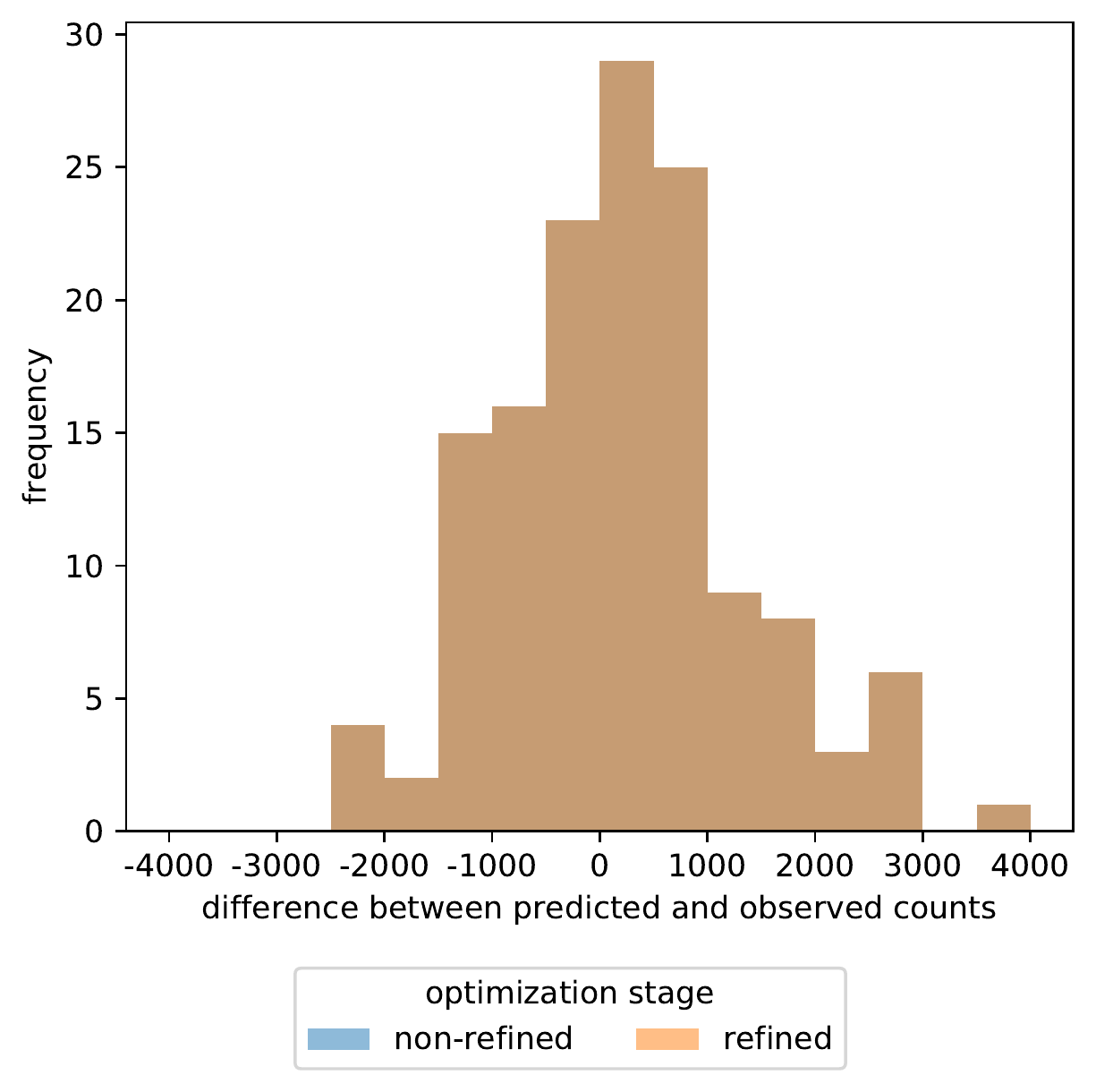}
	
	\caption{Baseline model 2020}
	\label{subfig:distribution-errors-baseline-model-2020}
	
\end{subfigure}
\begin{subfigure}[t]{0.3\columnwidth}
	\vskip 0pt
	\centering
	\includegraphics[width=\columnwidth, trim= {0cm 2cm 0cm 0cm},clip]{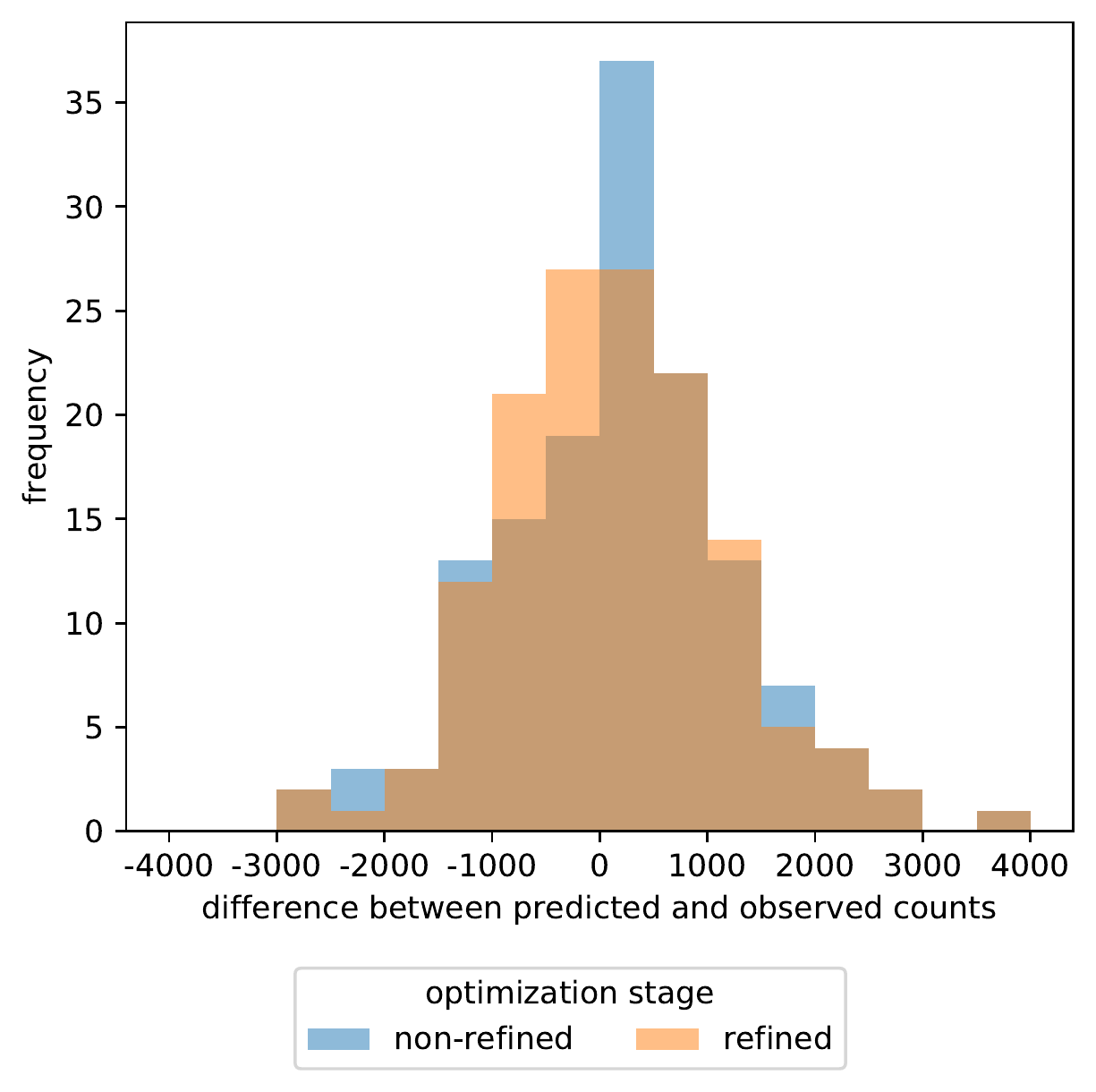}
	\caption{Full model 2020}
	\label{subfig:distribution-errors-full-model-2020}
	
\end{subfigure}
\begin{subfigure}[t]{0.3\columnwidth}
	\centering
	\vskip 0pt
	\includegraphics[width=\columnwidth, trim= {0cm 2cm 0cm 0cm},clip]{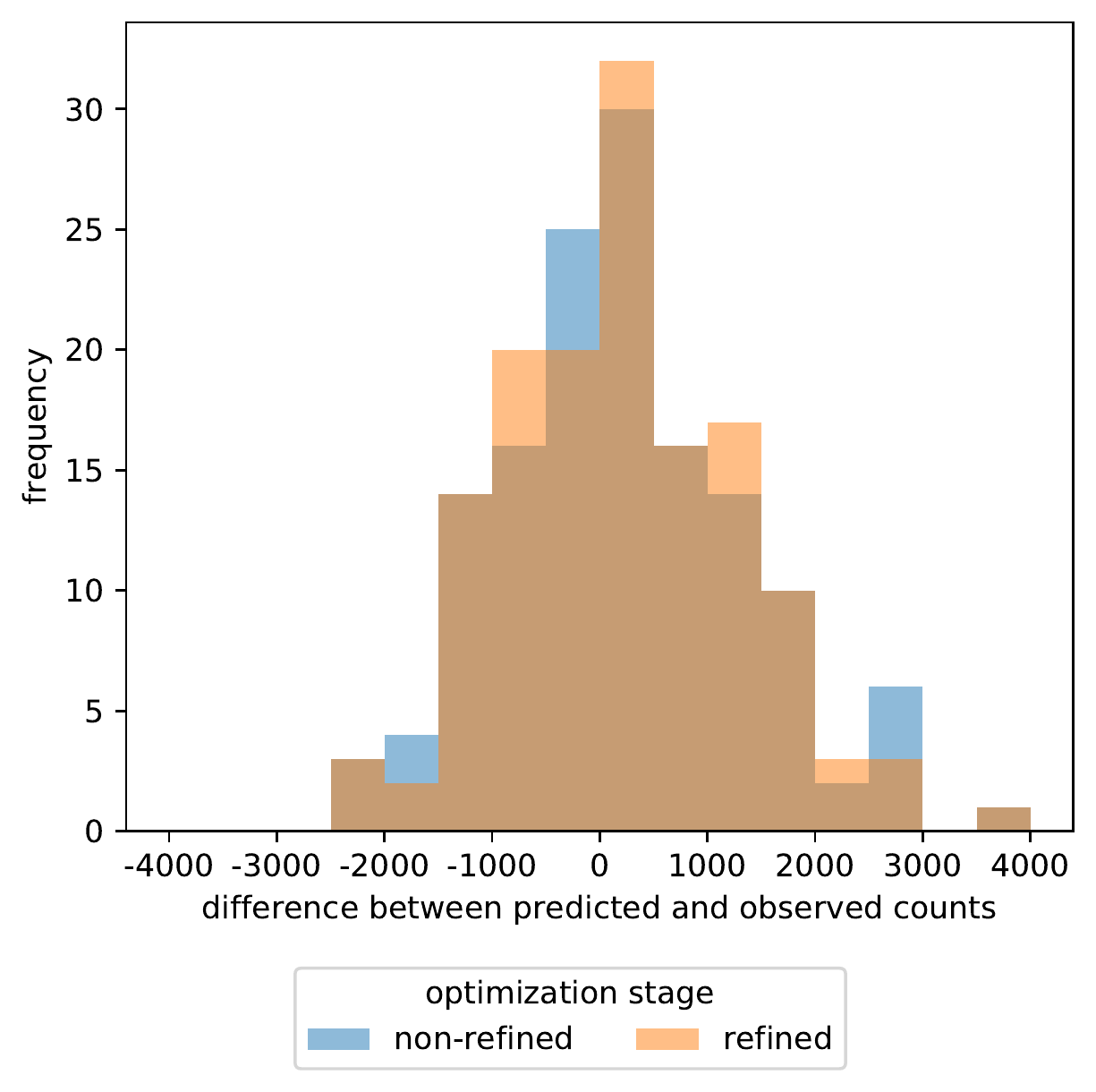}
	\caption{Binarized model 2020}
	\label{subfig:distribution-errors-feature-engineering-model-2020}
	
\end{subfigure}

\vspace{0.2cm}

\begin{subfigure}[b]{0.3\columnwidth}
	\centering
	\vskip 0pt
	\includegraphics[width=\columnwidth, trim= {0cm 0cm 0cm 11cm},clip]{figures/fresno/distribution-errors/distribution_predicted_count_error_feature_engineering_model_2020.pdf}
\end{subfigure}

\caption{Distribution of errors in non-refined and refined stages of the optimization of the baseline, full and binarized models}

\label{fig:distribution-errors-fresno}
\end{figure}

Notably, our results shows that the adjusted pseudo $R^2$ obtained in all models are in the order of 0.3-0.56, which are considered an excellent fit in travel behavior studies \citep{McFadden1977, Bovy2008, Bierlaire2008,Bwambale2019a}. The F-tests included in the table suggest that all models are statistically different than the null model at the 99\% confidence level. The F-tests comparing the \textit{Full model} with the \textit{Baseline model} are equal to 3.7804 and 4.3697 when using data collected before and during COVID, respectively. Because they are lower than their critical value of $F_{6-1,141-6 3.1557, 99\%}$, we conclude that the \textit{Full model} is not statistically equivalent to the \textit{Baseline model}. The indicators for model comparison of the \textit{Binarized model} were lower than the \textit{Full model} in all cases, thus, the latter is preferred over the former. 

The NRMSE of the \textit{Full model} is in the order of 0.5, meaning that the standard deviation of the traffic counts is approximately a 50\% of their sample mean. The experiments conducted in the Sioux Falls network find that a NRMSE of 0.25 results into more than 20\% false negatives. Although the sample size in the Fresno network is two times higher than the number of links in the Sioux Falls network, it is unlikely to compensate for the difference in NRMSE. Other factors that could contribute to the decrease of the t-test of the coefficients in the \textit{Full model} are associated to the high correlation of travel time with some exogenous attributes and also to numerical issues arising from the estimation of a richer utility function. Overall, this evidence anticipates a low statistical power to detect of the effect of attributes that significantly impact travelers' utility in the Fresno network.

As expected, the travel time coefficient is significant at the 90\% confidence level in all the models. Note that the estimated coefficients of travel times in the baseline models are exact multiples of the learning rate of \NGD because there was no improvement of the objective function in the refined stage. In contrast, the estimations of the full and binarized models report improvements during the refined stage. The estimation of the \textit{Full model} with the data collected before COVID-19 suggests that the total incidents during the year and the monthly household income of the neighborhoods nearby the link segments has a positive and negative effect on the travelers' utility, respectively. Both effects were significant at the 95\% level of confidence. The estimation results of the \textit{Full model} with the data collected during COVID-19 also found a negative effect of the number of streets intersection but the coefficients of all attributes, except for travel time, were not significant at the 90\% confidence level. As discussed earlier, the high level of noise in the Fresno network can explain the difficulty to detect significant effect of these attributes. Although it is not required in non-linear regression models, the distribution errors of the models resemble Gaussian distributions with zero mean (Figure \ref{fig:distribution-errors-fresno}). Finally, regarding the impact of COVID-19, we conclude that there are no significant differences on the attributes that determine travellers' route choices.

%% file: sections/conclusions.tex
\section{Conclusions}

The network modeling community has studied the problem of estimating the travelers' utility function coefficients using traffic counts and travel time measurements. However, research has been limited to utility functions dependent on travel time only and methods have been tested on networks of relatively small size. Under the assumption of a known exogenous O-D demand matrix, we enhance existing methods to estimate the coefficients of utility functions with multiple attributes and using traffic counts consistent with stochastic user equilibrium with \logit assignment (\SUE-\logit). We refer to this problem as Logit Utility Estimation (\LUE). To perform attributes' selection, we conduct hypothesis tests on the coefficients of a multi-attribute utility function . Furthermore, a rigorous analysis of the non-convexity and mathematical properties of the \LUE problem is conducted to inform the design of our solution algorithm and to derive some theoretical guarantees about convergence toward local optima. 

The realization of the pseudo-convexity of the optimization problem motivates the use of normalized gradient descent (\NGD), a first order method developed in the machine learning community that is suitable for pseudo-convex optimization. The integration of \NGD with Levenberg–Marquardt (\LM) algorithm outperforms the standalone application of second order optimization methods used in previous literature in many ways. First, the estimates of the utility function coefficients become less sensitive to the starting points for optimization, which is a common issue in non-convex problems. Second, \NGD improves the convergence toward global optima and this makes the statistical inference more reliable. Third, the use of first order methods reduces computational cost because it only requires to compute the gradient of the objective function respect to the utility function coefficients. To our knowledge, this is the first time that a paper presented vectorized expressions to perform the gradient computation and this is key to accelerate the optimization and to scale up our methodology to the largest transportation network studied to date within the \LUE literature. 

The analysis of mathematical properties of the problem identifies the coordinate-wise monotonicity of the traffic flow functions as a main driver of the coordinate-wise pseudo-convexity of the objective function. Experiments in networks used in previous studies show that the traffic flow functions are generally monotonic and that the objective function of the \LUE problem is pseudo-convex respect to the coefficient of a utility function dependent on travel time only. Results in the Sioux Falls network support the coordinate-wise pseudo-convexity of the objective function respect to the coefficients of a multi-attribute utility function. A series of Monte Carlo experiments show that statistical inference on the utility function coefficients is robust to traffic congestion and also to different levels of sensor coverage and noises in the O-D matrix and in the traffic counts. The amount of false negatives and false positives obtained in these experiments are generally well-aligned with statistical theory. For instance, the amount of false positives perfectly matches the value expected for an arbitrary significance level. A higher amount of noise in the traffic counts increase false negatives and hence it reduces the statistical power to identify effect of relevant attributes in the utility function. Surprisingly, the statistical inference is resilient to high level of noise in the cells of t he reference O-D matrix. 

Our solution algorithm is deployed on a large scale network in Fresno, CA and it gave reasonable results in terms of both the estimates and hypothesis tests of the utility function coefficients. Based on standard metrics of model comparison used in travel behavior research, the models report an excellent goodness of fit. As expected, travel time is identified as a main determinant of travelers' route choices in all model specifications, which partially supports the standard practice in the network modelling community of assuming utility functions dependent on travel time only. However, the incorporation of additional system level attributes in the utility function significantly increases the goodness of fit of the baseline model that uses travel time only. These conclusions are robust when using data collected before and during COVID-19. Among the exogenous attributes, the total incidents during the year and the median income are the most relevant predictors of travellers' route choices. The coefficients of these attributes are significant at the 95\% level of confidence but only when using data collected before COVID-19. 

The implementation of our methodology in a large scale network is challenging mainly due to the high computational cost and level of noise of real world data. On one hand, accounting for the interdependence between travelers' choices in transportation networks and for the endogeneity of travel times requires to compute traffic equilibria. On the other hand, estimating the coefficients of the travelers' utility function requires to solve a regression problem. The computational complexity of the inner and outer level problems is a function of the number of paths and thus, the alternating optimization of the problems become intractable with a large number of O-D pairs. We learned that the use of column generation methods for updating path sets in \SUE-\logit provides a good compromise between computational cost and prediction error. The strategy of selecting the O-D pairs according to their demand level is also key to control the exploration of new paths over iterations of the bilevel optimization and to also identify the paths that most decrease the prediction error. To our knowledge, this strategy is not integrated in previous column generation algorithms and we strongly encourage its use in further application of our methodology. To control for the effect of path correlation/overlapping, we correct path utilities according to the Path Size Logit model. 

\section{Further research}

Further research could enhance our methodology to incorporate multi-day data and to perform a joint estimation of the utility function coefficients and the O-D matrix. The increase in sample size should help to identify all parameters of interest and to improve the quality of the statistical inference. The use of cloud computing is important to handle a larger number of paths and thus, to achieve a larger reduction of prediction error. However, the cost of estimating a model with data from multiple time period can become high. We expect that use of deep learning models, computational graphs and automatic differentiation tools will help to handle large amounts of multi-day data while keeping computational cost reasonable. Regarding statistical inference, our study is lacking a more careful treatment of the endogeneity that arises from the computation of travel time over the iterations of the bilevel optimization. Therefore, the use of Two-Stage Least Squares (2SLS) and instrumental variables may contribute to correct bias and inconsistency in the coefficient estimates caused by the presence of endogeneity \citep{Gallant1979a}. 

We are also interested in relaxing the assumption of a homogeneous and linear-in-parameters utility function. The use of a non-homogeneous utility function and fixed effects may contribute to capture heterogeneity of preferences among individuals traveling between different O-D pairs. Besides, the effect of time variability may be better captured with a non-linear specification of the utility function as the one used in prospect theory route choice models. We will also look at estimating the coefficient weighting the utility term associated to the path size correction. We would also like to leverage the use of GPS data to have a better prior of the path sets among O-D pairs and to improve the estimation of the utility function coefficients. Here the integration of our methodology with the nested recursive logit model \citep{Mai2015} seems a promising avenue for further research. 

This study chooses traffic flows for the response function of the non-linear least objective functions. However, there are other choice of response functions that are also admissible such as link travel times or traffic densities. In static traffic assignment, both quantities are monotonic functions of the traffic counts and thus, they are suitable quantities to estimate the utility function coefficients. Finally, system level data is less subject to sampling bias than data collected from individual surveys but it is also more subject to measurement errors. Thus, we expect the joint use of travel surveys and system-level data can help leverage the strengths and weaknesses of each data source.

%% file: sections/acknowledgments.tex
\section{Model implementation and data}
\label{sec:model-implementation-data}

The Python package developed to implement our methodology and the system level data from the Fresno, CA network can be found at the following url: \url{https://github.com/pabloguarda/isuelogit}. The folder \texttt{notebooks} contains Jupyter notebooks that reproduce all the results presented in this paper.

\section{Acknowledgments}

This research is supported by a National Science Foundation grant CMMI-1751448

\section{Author contributions}

The authors confirm contribution to the paper as follows: study conception and design: Pablo Guarda, Sean Qian; data collection: Pablo Guarda, Sean Qian; programming and experiments: Pablo Guarda; analysis and interpretation of results: Pablo Guarda, Sean Qian; draft manuscript preparation: Pablo Guarda, Sean Qian. All authors reviewed the results and approved the final version of the manuscript.

%% file: sections/references.tex
\bibliographystyle{elsarticle-harv}\biboptions{authoryear}


\renewcommand{\refname}{}


%
%



\section*{References} \vspace{-1cm}
\bibliography{references.bib} 

%% file: sections/appendix.tex
\appendix

\section{Proofs and derivations}

\subsection{Notation}
\label{appendix:ssec:notation}

Tables \ref{table:notation1}, \ref{table:notation2} and \ref{table:notation3} present the notation used throughout the paper. 


\begin{table}[H]
	\caption{Network variables and parameters} 
	\begin{tabularx}{\textwidth}{ll@{}}
		\toprule
		Notations & Definitions \\
		\midrule
		$A$ & The set of all links \\
		$V$ & The set of all nodes \\
		$W$ & The set of O-D pairs\\
		$H$ & The set of all paths \\
		$H_{w}$ & The set of paths connecting O-D pair $w \in W$ \\
		$A^{o}, A^{u}$ & The sets of links with observed and unobserved traffic counts, respectively \\
		$\mIq \in \sR^{|H|\times |W|}$ & The path-demand incidence matrix \\
		$\mIx \in \sR^{|A|\times |H|}$ & The path-link incidence incidence matrix \\
		$\mQ \in \sR^{|V|\times |V|} $ & The O-D matrix \\
		$\vq \in \sR^{|W|}$ & The dense vector associated to the O-D matrix \\
		$q_w \in \sR$ & The demand in O-D pair $w \in W$ \\
		$\vx \in \sR_{\geq 0}^{|A|}$ & The vector of link flows \\
		$x_a  \in \sR_{\geq 0}$ & Link flow in link $a \in A$ \\
		$\vgamma \in \sR^{|A|}_{\geq 0}$ & The vector of capacities among links \\
		$\gamma_a \in \sR_{\geq 0}$ & The capacity of link a \\
		$\bar{\vt}^0 \in \sR^{|A|}_{+}$ & The vector of links' free flow travel times \\
		$\bar{t}_a^0 \in \sR_{+}$ & The free flow travel time at link $a \in A$ \\
		$\vf  \in \sR_{\geq 0}^{|H|}$ & The vector of path flows \\
		$f_{h} \in \sR_{\geq 0}$ & The path flow on path $h \in H$\\
		$\vpf  \in \sR_{]0,1[}^{|H|}$ & The vector of path choice probabilities \\
		$p_h \in \sR_{]0,1[}$ & The choice probability of path $h \in H$  \\
		\bottomrule
	\end{tabularx}
	\label{table:notation1}
\end{table}


\begin{table}[H]
	\caption{Behavioral variables and parameters} 
	\begin{tabularx}{\textwidth}{ll@{}}
		\toprule
		Notations & Definitions \\
		\midrule
		$K_{\mZ}$ & The set of exogenous attributes in the utility function \\
		$K$ & The set of attributes in the utility function \\
		$D$ & The set of utility function coefficients\\
		$L$ & The set of travelers in the network \\
		$J_{l}$ & Consideration set of traveler $l \in L$\\
		$\vt \in \sR_{\geq 0}^{|A|}$ & The vector of values of the endogenous travel times among links \\
		$\mZ  \in \sR^{|A| \times |K_{\mZ}|}$ & The matrix of values for the exogenous attributes among links  \\
		$\vz_k \in \sR^{|A|}$ & The vector of values for the exogenous attribute $k \in K_{\mZ}$ among links \\
		$Z_{ak} \in \sR $ & The value of the exogenous attribute $k \in  K_{\mZ}$ at link $a \in A$ \\
		$\bar{\vt} \in \sR_{\geq 0}^{|A|}$ & The vector of exogenous travel times among links \\
		$\vtheta \in \sR^{|D|}$ & The vector of true utility function coefficients \\
		$\vtheta_{\mZ}  \in \sR^{|K_{\mZ}|}$ & The vector of true utility function coefficients associated to the exogenous attributes \\
		$\theta_d \in \sR$ & The utility function coefficient associated to attribute $d \in D$ \\
		$\theta_t \in \sR_{\leq 0}$ & The travel time coefficient \\
		$\vv \in \sR^{|A|}$ & The vector of link utilities \\
		$v_{a} \in \sR$ & The link utility associated to link $a \in A$ \\
		$U_{jl} \in \sR$ & The latent (unobservable) utility that traveler $l$ attained to alternative (path) $j \in J_l$ \\
		$V_{jl} \in \sR$ & The observable utility that traveler $l$ attained to alternative (path) $j \in J_l$ \\
		$\epsilon_{jl} \in \sR$ & The latent (unobservable) error in the utility function associated to traveler $l$ and alternative (path) $j \in J_l$ \\
		$\mu \in \sR_{+}$ & Scale parameter of the logit model and of the extreme value Type 1 distribution \\
		\bottomrule
	\end{tabularx}
	\label{table:notation2}
\end{table}

\begin{table}[H]
	\caption{Variables and parameters for statistical inference} 
	\begin{tabularx}{\textwidth}{ll@{}}
		\toprule
		Notations & Definitions \\
		\midrule
		$N$ & The sample of traffic counts \\  
		$\ell(\vtheta): \sR^{|D|} \to \sR$ & The objective function of the \LUE problem \\
		$\hat{\vtheta} \in \sR^{|D|}$ & The vector of estimated utility function coefficients \\
		$\hat{\vtheta}_{\mZ} \in \sR^{|K_{\mZ}|} $ & The vector of estimated utility function coefficients associated to the exogenous attributes \\
		$\vpf(\hat{\vtheta}): \sR^{|D|} \to \sR^{|H|}$ & The vector of path choice probability functions\\
		$p_h(\hat{\vtheta}): \sR^{|D|} \to \sR$ & The path choice probability function associated to path $h \in H$\\
		$\vx(\hat{\vtheta}): \sR^{|D|} \to \sR^{|A|}$ & The vector of traffic count (response) functions \\
		$x_a(\hat{\vtheta}): \sR^{|D|} \to \sR$ & The traffic flow (response) function associated to link $a \in A$  \\
		$\tilde{\mX} = D_{\vtheta} \ \vx(\vtheta) \in \sR^{|A^o| \times |D|}$ & The design matrix in \NLLS and which is equal to the Jacobian matrix of the vector of traffic flow functions respect to $\vtheta$  \\
		$\bar{\vx} \in \sR_{\geq 0}^{|A^o|}$ & The vector of observed traffic counts \\
		$\bar{x}_a \in \sR_{\geq 0}$ & The traffic count measurement at link $a \in A$ \\
		$\bar{T}_{d,H_0} \in \sR$ & The t-test associated to attribute $d \in D$ and under null hypothesis $H_0$ \\
		$\bar{F}_{1,2} \in \sR_{\geq 0}$ & The f-test comparing models 1 and 2 \\
		$\sigma^2 \in \sR_{\geq 0}$ & The variance of the errors in the nonlinear regression\\
		$\hat{\sigma}^2 \in \sR_{\geq 0}$ & The estimated variance of the errors in the nonlinear regression\\
		\bottomrule
	\end{tabularx}
	\label{table:notation3}
\end{table}



\subsection{Extension of \SUE-\logit}
\label{appendix:ssec:derivations-sue-logit}

\subsubsection{Original travel time based formulation}

The standard formulation of the \SUE with \logit assignment (\SUE-\logit) problem assumes a utility function dependent on travel time only.  \citet{Fisk1980} proved that the first order necessary optimality condition of the following optimization problem gives a path flow solution that is \logit distributed:


\begin{mini}|l|
	{\{f_h\}_{h \in H},\{x_a\}_{a \in A}}{\sum_{a \in A}\int_{0}^{x_a} t_a(u) du +\frac{1}{\theta} \sum_{w \in W} \sum_{h \in H_{w}} f_{h}\ln(f_{h})}{}{}
	\addConstraint{\sum_{h \in H_w} {f_h}}{= {q}_w\quad}{\forall w \in W}
	\addConstraint{x_a}{= \sum_{w \in W}\sum_{h \in H_{w}} f_h \delta_{ah}\quad}{\forall a \in A}
	\addConstraint{{f_h}}{\geq 0}{\forall h \in H}
	\label{eq:original-sue-logit}
\end{mini} 

From this formulation is clear that if $|\theta| \to \infty^{+}$, the objective function reduces to the first term associated to the Beckmann transformation, and thus, the \SUE and \DUE optimization problems become equivalent. The \LUE and \ODLUE literature typically define $\theta \in \sR_{+}$ as a dispersion parameter measuring the sensitivity of route choices to travel times \citep{Yang2001} and interpret it as the accuracy of the travelers' perception about travel costs \citep{Daganzo1977, Wang2016} or the level of information about travel costs \citep{Lo2003,Liu1996}. Under a single attribute utility function, the interpretation of the parameter may be irrelevant or difficult to falsify. However, in the case of a multi-attribute utility function, these interpretations would rest importance on the relationship of the  magnitude of the dispersion parameter with the amount of unobservable components of the utility function and which are ignored by the modeler. To understand this connection, it is key to reformulate the problem into a utility based representation. 


\subsubsection{Utility based formulation with a single endogenous attribute}

The objective function in Problem \ref{eq:original-sue-logit} is written in terms of the link performance functions $t_a: \sR_{\geq 0} \to \sR_+ $ instead of the observable component of the travelers' utility function $v_a: \sR \to \sR $  associated to each link $a \in A$. Assume the observable component of the travelers' utility function is given by $v_a = \theta_t t_a$, where $\theta_t = \mu \tilde{\theta}_t$ is the coefficient measuring the preference of travelers' for travel time and it is scaled by a factor $\mu \in \sR_{+}$ proportional to the variance of the unobservable component of the utility function, i.e. $\theta_t = \mu \tilde{\theta_t}$ where $\tilde{\theta_t}$ is the unscaled vector of \logit coefficients and which is not identifiable. Then, the utility based representation of Problem \ref{eq:original-sue-logit} can be written in vectorized form as follows:

\begin{center}
	\scalebox{1}{\parbox{\linewidth}{
			\begin{mini}
				{\vx, \vf}{  \sum_{a \in A}\int_{0}^{x_a} \frac{v_a(u)}{\theta_t}du - \frac{1}{\theta_t} \left\langle \vf, \ln \vf\right\rangle }{}{}
				\addConstraint{\mIq \vf}{= \vq}{}
				\addConstraint{\mIx \vf}{= \vx}{}
				\addConstraint{\vx, \vf}{\geq \vzero}{}
				\label{eq:single-attribute-sue-logit}
			\end{mini}
	}}
\end{center}

%

where $\vx \in \sR_{\geq 0}^{|A|}, \vf \in \sR_{\geq 0}^{|H|}, \vq \in \sR_{+}^{|V\times V|}, \mIq \in \sR_{+}^{|V \times V| \times |H|}, \ \mIx \in \sR_{+}^{|A| \times |H|}$. A key observation respect to Problem \ref{eq:original-sue-logit} is that $\theta = -\theta_t > 0$ and since $\theta_t = \mu \tilde{\theta_t}$ is now clear that the dispersion parameter is also scaled by the scale factor $\mu$ of the \logit model.

\subsubsection{Utility function with an endogenous attribute and multiple exogenous attributes}

Problem \ref{eq:single-attribute-sue-logit} can written as:

\begin{mini}|l|
	{\{f_h\}_{h \in H},\{x_a\}_{a \in A}}{\sum_{a \in A}\int_{0}^{x_a} v^{\prime}_a(u)du- \frac{1}{\theta_t} \sum_{w \in W} \sum_{h \in H_{w}} f_h \ln f_h }{}{}
	\addConstraint{\sum_{h\in H_w} f_h}{= {q}_w\quad}{\forall w \in W}
	\addConstraint{x_a}{= \sum_{rs}\sum_{h \in H_{w}} f_h \delta_{ah}\quad}{\forall a \in A}
	\addConstraint{f_h}{\geq 0}{\forall h \in H}
	\label{eq:multi-attributes-sue-logit-minimization}
\end{mini}

where the link utility function $v_a(u)$ at a traffic flow level $u$ was reparameterized as:

\begin{equation}
	v^{\prime}_a(u) 
	= \frac{v_a(u)}{\theta_t}
	= \frac{1}{\theta_t}\left( \theta_t t_a(u) + \sum_{k \in K_{\mZ}} \theta_k \cdot Z_{ak}\right)
\end{equation}

$Z_{ak}$ represents the value of the exogenous attribute $k \in K_{\mZ}$ at link $a \in A$ and $\theta_t \in \sR_{-}, \vtheta_Z \in \sR^{|K_{\mZ}|}$ are set the of coefficients measuring the travelers' preferences for the endogenous attribute $t$ and the exogenous attributes $z \in K_{\mZ}$. Note that if both terms of the objective function in Problem \ref{eq:multi-attributes-sue-logit-minimization} are multiplied by $\theta_t <0$, the problem becomes a maximization:



\begin{maxi}|l|
	{\{f_h\}_{h \in H},\{x_a\}_{a \in A}}{\sum_{a \in A}\int_{0}^{x_a} v_a(u) du - \sum_{w \in W} \sum_{h \in H_{w}} f_h \ln f_h }{}{}
	\addConstraint{\sum_{h\in H_w} {f_h}}{= {q}_w\quad}{\forall w \in W}
	\addConstraint{x_a}{= \sum_{w \in W}\sum_{h \in H_{w}} f_h \delta_{ah}\quad}{\forall a \in A}
	\addConstraint{f_h}{\geq 0}{\forall h \in H}
\end{maxi}

which, in compact form, can be written as:

\vspace{-0.5cm}

\begin{maxi*}
	{\vx, \vf}{  \sum_{a \in A}\int_{0}^{x_a} v_a(u)du - \left\langle \vf, \ln \vf\right\rangle }{}{}
	\addConstraint{\mIq\vf}{= \vq}{}
	\addConstraint{\mIx\vf}{= \vx}{}
	\addConstraint{\vx,\vf}{\geq \vzero}{}
\end{maxi*}





\subsubsection{Logit assignment of path flows in utility based formulation}

Following the rationale of the proof in \citet{Fisk1980}, we can prove that the path flow solution of the optimization model presented in Problem \ref{eq:multi-attributes-sue-logit-minimization} follows a \logit assignment. The Lagrangian $\mathcal{L}$ of the problem is:

\begin{equation}
	\mathcal{L} = 
	\sum_{a \in A}\int_{0}^{x_a} (t_a(w) + \sum_{k \in K_{\mZ}} (\theta_k/\theta_t)Z_{ak} ) dw - \frac{1}{\theta_t} \sum_{w \in W} \sum_{h \in H_{w}} f_h \ln f_{h}
	+ \sum_{w \in W}\lambda_{w}(\sum_{h \in H_w} {f_h}-{q}_{w})
	\label{eq:lagrangian-two-attributes-sue-logit}
\end{equation}
\noindent The set of first order optimality conditions are: 
\begin{align}
	\frac{\partial L}{\partial {f_h}}
	= 
	\sum_{a \in A} \frac{\partial x_a}{\partial f_h} \Big(t_a(x_a)   + \sum_{k \in K_{\mZ}} (\theta_k/\theta_t)Z_{ak} \Big) 
	- \frac{1}{\theta_t} \left(\ln f_h+f_h \frac{1}{f_h}\right)
	+ \sum_{w \in W} \lambda_{w} \delta^w_{hw} &= 0 \nonumber
	\\   
	\sum_{a \in A} \delta_{ah} \left (t_a(x_a)   + \sum_{k \in K_{\mZ}} (\theta_k/\theta_t)Z_{ak} \right) 
	- \frac{1}{\theta_t} \left(\ln f_h+1\right)
	+  \lambda^h_{w} &= 0
	\quad \forall h \in H
	\label{eq:first-order-optimality-conditions-two-attributes-sue-logit}
\end{align}
where $\delta^w_{hw}$ takes the value 1 if path $h$ belong to O-D pair $w \in W$, and 0 otherwise, and $\lambda^h_{w} = \sum_{w \in W} \lambda_{w} \delta^w_{hw}$. Let's define $V^{\prime}_{h} = \sum_{a \in A} \delta_{ah} (t_a(x_a)   + \sum_{k \in K_{\mZ}} (\theta_k/\theta_t)Z_{ak}$ as a reparameterized utility function associated to path $h \in H$. Now we can find an expression for $f_h$ using Eq. \ref{eq:first-order-optimality-conditions-two-attributes-sue-logit}: 
\begin{align}
	V^{\prime}_{h} - \frac{1}{\theta_t} \left(\ln f_h+1\right)
	&= -\lambda^h_{w} \nonumber \\  
	\ln f_h &= \theta_t\lambda^h_{w} \delta^w_{hw}+ \theta_t V^{\prime}_{h}-1  \nonumber \\
	f_h &= \exp(\theta_t\lambda^h_{w}+ \theta_t V^{\prime}_h-1)  
	\label{eq:path-flow-two-attributes-sue-logit}
\end{align}
Using the conservation constraint of path flows and demand:
\begin{equation}
	\sum_{h \in H_w} f_w= \sum_{h \in H_w} \exp(\theta_t\lambda^h_{w}+ \theta_t V^{\prime}_h-1) = q_w  
\end{equation}

\noindent Noting that the value of $\lambda^h_{w}$ is the same $\forall h \in H_w$:
\begin{align}
	q_w &= \exp(\theta_t\lambda^h_{w})\sum_{h \in H_{w}} \exp(\theta_t V^{\prime}_{h}-1) \nonumber \\
	\exp(\theta_t\lambda^h_{w}) &= \frac{q_w}{\displaystyle\sum_{h \in H_{w}} \exp(\theta_t V^{\prime}_{h}-1)}
	\label{eq:derivation-two-attributes-sue-logit-1}
\end{align}

\noindent Replacing Eq. \ref{eq:derivation-two-attributes-sue-logit-1} into Eq. \ref{eq:path-flow-two-attributes-sue-logit}:
\begin{align}
	f_h
	= \exp(\theta_t\lambda^h_{w}) \exp( \theta_t V^{\prime}_{h}-1)
	=q_w \frac{\exp(\theta_t V^{\prime}_{h})}{\displaystyle\sum_{j \in H_{w}}^{} \exp(\theta_t V^{\prime}_{j})} 
\end{align}


\noindent Substituting by $V_h = V^{\prime}_{h} \theta_t = \sum_{a \in A} \delta_{ah} (\theta_t t_a(x_a) + \sum_{k \in K_{\mZ}} (\theta_k/\theta_t)Z_{ak})$:



\begin{equation}
	f_{h} 
	= 
	q_w \frac{\exp(V_{h})}{\displaystyle\sum_{j \in H_{w}}^{} \exp(V_{h})} 
\end{equation}

\noindent where it is clear that the set of optimal path flows $\{f_h\}_{h \in H} $ follows a \logit assignment.  

%

\subsection{Monotonicity of path choice probabilities and traffic flow functions}
\label{appendix:ssec:monotonicity-traffic-count-functions}

\begin{prop}[monotone path choice probabilities]
	\label{prop:monotonocity-softmax}
	Assume the travelers' utility function is a linear weight between a set of attributes and coefficients. If there are as most two paths to travel between every O-D pair, the path choice probabilities are coordinate-wise monotonic functions respect to the utility function coefficients.
\end{prop}

\begin{proof}
	Let's start from a general case where there are an arbitrary number of alternative paths and attributes and thus, where the choice probabilities are obtained from a softmax function. Define $\theta_t \in \sR$ as the utility function coefficient associated to an attribute $t \in K$ and $\tau_i$ as the utility component associated to the remaining set of attributes of the path $i \in H_{rs}$ of an arbitrary O-D pair $w \in W$. Then, the choice probability of path $i$ in the O-D pair $w \in W$ is:
	
	\begin{equation}
		\displaystyle p_i(\theta_t) = \dfrac{\displaystyle \exp(\theta_t t_i+\tau_i)}{\displaystyle \sum_{j \in H_{w}} \exp(\theta_t t_j+\tau_j) }
	\end{equation}
	
	The first derivative of the softmax function respect to the utility function parameter $\theta$ is:
	
	\begin{align*}
		\dfrac{\partial p_i(\theta)}{\partial \theta} 
		&= \dfrac{\displaystyle \exp(\theta_t t_i+\tau_i)t_i\left(\sum_{j \in H_{w}} \exp(\theta_t t_j+\tau_j)\right)-\left(\sum_{j \in H_{w}} \exp(\theta t_j+\tau_j)t_j\right)\exp(\theta_t t_i+\tau_i)}{\displaystyle \left(\sum_{j \in H_{w}} \exp(\theta_t t_j+\tau_j)\right)^2 }\\
		&= \dfrac{\displaystyle \exp(\theta_t t_i+\tau_i)}{\displaystyle \sum_{j \in H_{w}} \exp(\theta t_j+\tau_j)}\dfrac{\displaystyle \sum_{j \in H_{w}} \exp(\theta_t t_j+\tau_j)(t_i-t_j)}{\displaystyle \sum_{j \in H_{w}} \exp(\theta t_j+\tau_j)}\\
		&= p_i\dfrac{\displaystyle \sum_{j \in H_{w}} \exp(\theta_t t_j+\tau_j)(t_i-t_j)}{\displaystyle \sum_{j \in H_{w}} \exp(\theta_t t_j+\tau_j)}
	\end{align*}
	
	The sign of the derivative is given by: 
		%
	\begin{align}
		\label{eq:sign-softmax}
		\sign\left({\dfrac{\partial h_i(\theta_t)}{\partial \theta}}\right)
		= \sign\left(p_i\dfrac{\displaystyle \sum_{j \in H_{w}} \exp(\theta_t t_j+\tau_j)(t_i-t_j)}{\displaystyle \sum_{j \in H_{w}} \exp(\theta t_j+\tau_j)}\right)
		= \sign\left(\displaystyle \sum_{j \in H_{w}} \exp(\theta t_j+\tau_j)(t_i-t_j)\right)
	\end{align}
	
	Given that exponential functions are always positive, Eq. \ref{eq:sign-softmax} will be negative or positive for any $\theta_t \in \sR$ when $t_i \neq \min{\{t_j\}_{j\in H_{w}}}$ or $t_i \neq \max{\{t_j\}_{j\in H_{w}}}$. Therefore, when there are at most two paths connecting the O-D pair, the path choice probability for any path $i \in H_w$ will be necessarily a monotonic function respect to $\theta_t$. The same analysis can be extended to every utility function coefficient $d\in D$ and O-D pair $w \in W$. This proves that the path choice probabilities are coordinate-wise monotonic functions respect to the utility function coefficients when all paths sets have at most two paths.

	%
	%
	
	
	
\end{proof}


\begin{prop}[monotone traffic flow functions under dominant and non dominant paths]
	\label{prop:monotonicity-softmax-sum}
	Assume the travelers' utility function is a linear weight between a set of attributes and coefficients. Let's define dominated and dominating paths as those paths where an attribute of the utility function reaches its maximum or minimum value, respectively, within each O-D pair. If the set of paths traversing a link are all dominating or dominated paths, the traffic flow function at that link is coordinate-wise monotonic respect to the utility function coefficients
\end{prop}

\begin{proof}
	
	Consider an arbitrary attribute $t \in K$ and define $\theta_t$ as the coefficient weighting that attribute in the travelers' utility function. Denote $\tau_i$ as the utility component associated to the remaining set of attributes in path $i \in H$. The traffic flow function $x_a(\theta): \sR^{|D|} \to \sR$ associated to any link $a \in A$ is given by: 
	
	\begin{equation}
		\label{eq:traffic-count-function-monotonicity-traffic-counts-case1}
		x_a(\theta) 
		= \sum_{i \in H} f_i \delta_{ai}
		= \sum_{i \in H}  \delta^a_{ai} \sum_{w \in W} q_w  p_i(\theta_t) \delta^w_{wi}
		= \displaystyle   \sum_{i \in H}  \delta^a_{ai}  \sum_{w \in W} \delta^w_{wi}\  q_w \dfrac{\exp(\theta_t t_i+\tau_i)}{\displaystyle \sum_{j \in H_{w}} \exp(\theta_t t_j+\tau_j) }
	\end{equation}
	
	where $\delta^a_{ai}$ takes the value 1 if path $i$ traverses link $a$, and 0 otherwise, and $\delta^w_{wi}$ takes the value 1 if path $i$ belong to O-D pair $w$, and 0 otherwise. If the set of paths traversing the link are all dominating or dominated, then $\forall w \in W, \ i \in H_w$ , $t_i = \min\{t_j\}_{j \in H_{w}}$ or $t_i = \max\{t_j\}_{j \in H_{w}}$, respectively. From Proposition \ref{prop:monotonocity-softmax}, we observed that the softmax function associated to each path choice probability $p_i(\theta_t)$ will be either a monotonically decreasing or a increasing function when $t_i = \min\{t_j\}_{j \in H_{rs}}$ or $t_i = \max\{t_j\}_{j \in H_{rs}}$, respectively. By assumption, $x_a(\theta)$ is a positive weighted sum of choice probabilities associated to dominating or dominated paths, respectively,  that is a positive weighted sum of monotonically increasing or decreasing functions. Therefore, $x_a(\theta)$  it is either a monotonically increasing or decreasing function respect to $\theta_t$ and the same analysis can be applied to every coordinate $d \in D$. Finally, $x_a(\theta)$ is a coordinate-wise monotonic function, which completes the proof. 
	
\end{proof}


	%
	%
	%
	%
	%
	%
	%
	%

\begin{prop}[monotone traffic flow functions under binary attributes]
	\label{prop:monotonicity-sigmoid-sum}
	Suppose there are only two paths connecting every O-D pair and that the travelers' utility function only depends on a binary attribute that can take the values 0 or 1. Then, the traffic flow function is monotonic respect to the coefficient weighting that attribute.
\end{prop}

\begin{proof}	
	
	The traffic flow function $x_a(\theta): \sR^{|D|} \to \sR$ associated to any link $a \in A$ can be expressed as:
	
	\begin{align*}
		x_a(\theta) 
		= \sum_{i \in H}  \delta^a_{ai} \sum_{w \in W} q_w  p_i(\theta_t) \delta^w_{wi}
		&= \sum_{i \in H}  \delta^a_{ai} \sum_{w \in W} \delta^w_{wi}\ q_w \dfrac{\exp(\theta \mathbb{I}(z_i = 1))}{\exp(\theta \mathbb{I}(z_i = 1)) + \exp(\theta \mathbb{I}(z_{-i} = 1)}\\
		&= \sum_{i \in H}  \delta^a_{ai} \sum_{w \in W} \delta^w_{wi}\ q_w  \dfrac{1}{1 + \exp(\theta \mathbb{I}(z_{-i} = 1)-(\theta \mathbb{I}(z_{i} = 1))) }
	\end{align*}
	
	where $z_i \in \{0,1\}, \forall i \in H$, 	$\delta^a_{ai}$ takes the value 1 if path $i$ traverses link $a \in A$, and 0 otherwise, and $\delta^w_{wi}$ takes the value 1 if path $i$ belong to O-D pair $w \in W$, and 0 otherwise. Given the existence of two alternatives per O-D pair, we can express the $x_a(\theta)$ in terms of the sigmoid function $\sigma(\cdot)$:
	
	\begin{align*}
		x_a(\theta) 
		&= \sum_{i \in H}  \delta^a_{ai} \sum_{w \in W} \delta^w_{wi}\ q_w   \sigma\left(\theta (\mathbb{I}(z_{-i} = 1)-\mathbb{I}(z_{i} = 1))\right) \\
		&= \sum_{i \in H}  \delta^a_{ai} \sum_{w \in W} \delta^w_{wi}\ q_w  \mathbb{I}(z_i = z_{-i})\sigma(0)
		+ \sigma\left(-\theta\right) \sum_{i \in H}  \delta^a_{ai} \sum_{w \in W} \delta^w_{wi}\ q_w \mathbb{I}(z_i > z_{-i})
		+ \sigma\left(\theta\right) \sum_{i \in H}  \delta^a_{ai} \sum_{w \in W} \delta^w_{wi}\ q_w \mathbb{I}(z_i < z_j)\\
		&= \sum_{i \in H}  \delta^a_{ai} \sum_{w \in W} \delta^w_{wi}\ q_w  \mathbb{I}(z_i = z_{-i})\sigma(0)+ \left(1-\sigma(\theta)\right)\sum_{i \in H}  \delta^a_{ai} \sum_{w \in W} \delta^w_{wi}\ q_w  \mathbb{I}(z_i > z_{-i}) + \sigma\left(\theta\right)\sum_{i \in H}  \delta^a_{ai} \sum_{w \in W} \delta^w_{wi}\ q_w \mathbb{I}(z_i < z_j)\\
		&= \sum_{i \in H}  \delta^a_{ai} \sum_{w \in W} \delta^w_{wi}\ q_w  \mathbb{I}(z_i = z_{-i})\sigma(0)+ \sum_{i \in H}  \delta^a_{ai} \sum_{w \in W} \delta^w_{wi}\ q_w  \mathbb{I}(z_i > z_{-i}) 
		\\ &+ \sigma(\theta)\left(\sum_{i \in H}  \delta^a_{ai} \sum_{w \in W} \delta^w_{wi}\ q_w  \mathbb{I}(z_i < z_j)-\sum_{i \in H}  \delta^a_{ai} \sum_{w \in W} \delta^w_{wi}\ q_w \mathbb{I}(z_i > z_{-i}) \right)
	\end{align*}
	
	To analyze the monotonicity of $x_a(\theta)$ is convenient to compute the first derivative of $x_a(\theta)$ is: 
	
	\begin{equation}
		\label{eq:first-derivative-monotonicity-traffic-count-functions-case1}
		\dfrac{\partial x_a(\theta)}{\partial \theta} = \left(\sum_{i \in H}  \delta^a_{ai} \sum_{w \in W} \delta^w_{wi}\ q_w  \mathbb{I}(z_i < z_{-i})-\sum_{i \in H}  \delta^a_{ai} \sum_{w \in W} \delta^w_{wi}\ q_w \mathbb{I}(z_i > z_{-i}) \right)\sigma(\theta)(1-\sigma(\theta))
	\end{equation}
	
	and to then analyze its sign:
	
	\begin{equation}
		\label{eq:sum-demands-sign-monotonicity-traffic-count-functions-case1}
		\sign\left(\dfrac{\partial x_a(\theta)}{\partial \theta}\right) = \left(\sum_{i \in H}  \delta^a_{ai} \sum_{w \in W} \delta^w_{wi}\ q_w  \mathbb{I}(z_i < z_{-i})-\sum_{i \in H}  \delta^a_{ai} \sum_{w \in W} \delta^w_{wi}\ q_w \mathbb{I}(z_i > z_{-i}) \right)
	\end{equation}

	which does not depend on $\theta$ because $\sigma(\theta)(1-\sigma(\theta)) > 0, \forall \theta \in \sR$ in Eq. \ref{eq:first-derivative-monotonicity-traffic-count-functions-case1}. Then, it is clear that the function $x(\theta)$ is monotonic respect to $\theta$ and this completes the proof. 

\end{proof}

\begin{remark}
	
	Note that the left and right terms in Eq. \ref{eq:sum-demands-sign-monotonicity-traffic-count-functions-case1} are the sums of demand associated to dominated and dominating paths, respectively. Therefore, the traffic flow function will be monotonically increasing or decreasing if the sum associated to the dominating paths is greater or lower, respectively. 
	
\end{remark}



	%


\subsection{Coordinate-wise properties of the objective function}
\label{appendix:ssec:coordinate-wise-properties}

\begin{prop}[Coordinate-wise pseudo-convexity of objective function]
	\label{prop:pseudoconvexity-uncongested-network}
	The objective function of the \LUE problem under an uncongested network is coordinate-wise pseudo-convex if the traffic flow functions are coordinate-wise monotonic respect to each utility function coefficient. 
\end{prop}


\begin{proof}	
	%
	
	Let's define $\vtheta^{1}_d, \vtheta^{2}_d \in \mathbb{R}^{|D|}$ as vectors with all coordinates set to 0, except for the coordinate $t \in D$. Let's be $\theta^{1}_d,\theta^{2}_d \in \mathbb{R}$ the values of the non-zero coordinates in $\vtheta^{1}_d, \vtheta^{2}_d \in \mathbb{R}$. To prove coordinate-wise pseudo-convexity of the objective function $\ell: \sR^{|D|} \to \sR$, it suffices to show that the following holds:
	\begin{align}
		\label{eq:pseudo-convexity-proof1}
		2\left(\vx(\vtheta^{1}_d)-\bar{\vx} \right)(\theta^{2}_d-\theta^{1}_d) \dfrac{\partial \vx}{\partial \theta_d}\Bigg|_{\vtheta = \vtheta^1_d}
		\geq \vzero
		\implies 
		\|\vx(\vtheta^{2}_d)-\bar{\vx}\|^2_2   \geq \|\vx(\vtheta^{1}_d)-\bar{\vx}\|^2_2
	\end{align}
	
	
	By assumption, the traffic count (response) functions are coordinate-wise monotonic. Let's start considering the set $J^{+}$ of functions that are monotonically increasing respect to $\theta_d \in \mathbb{R}$. Since $\dfrac{\partial \vx_j}{\partial \theta_d}\Big|_{\vtheta = \vtheta^1_d}> 0, \forall j \in J^{+}$, the LHS in Eq. \ref{eq:pseudo-convexity-proof1} becomes non-negative in the following two cases: \\
	
	
	$\bullet$ Case (i): $\vx_j(\vtheta^{1}_d)-\bar{\vx}_j \geq 0 \land \theta_d^{2} \geq \theta_d^{1}$. Because the increasing monotonicity of the traffic flow functions $j \in J^{+}$ respect to $\theta_d$ and given that $\theta_d^{2} \geq \theta_d^{1}$:
	\begin{align*}
		\vx_j(\vtheta^2_d) &\geq \vx_j(\vtheta^1_d)\\
		\vx_j(\vtheta^2_d) - \bar{\vx}_j &\geq  \vx_j(\vtheta^1_d) - \bar{\vx}_j
	\end{align*}
	
	Since $\vx(\theta^{1}_d)-\bar{\vx} \geq 0$:
	\begin{align*}
		\vx_j(\vtheta^2_d) - \bar{\vx}_j &\geq  \vx_j(\vtheta^1_d) - \bar{\vx}_j \geq 0 \\
		(\vx_j(\vtheta^2_d) - \bar{\vx}_j)^2 &\geq (\vx_j(\vtheta^1_d) - \bar{\vx}_j)^2
	\end{align*}
	
	$\bullet$ Case (ii): $\vx_j(\vtheta^{1}_d)-\bar{\vx}_j \leq 0 \land  \theta_d^{2} \leq \theta_d^{1}$. Because the increasing monotonicity of the traffic flow functions $j \in J^{+}$ respect to $\theta_d$ and given that $\theta_d^{2} \leq \theta_d^{1}$:
	\begin{align*}
		\vx_j(\vtheta^2_d) &\leq \vx_j(\theta^1_d)\\
		\vx_j(\vtheta^2_d) - \bar{\vx}_j &\leq  \vx_j(\vtheta^1_d) - \bar{\vx}_j
	\end{align*}
	
	Since $\vx_j(\theta^{1}_d)-\bar{\vx}_j \leq 0$:
	\begin{align*}
		\vx_j(\vtheta^2_d) - \bar{\vx}_j &\leq  \vx_j(\vtheta^1_d) - \bar{\vx}_j \leq 0 \\
		(\vx_j(\vtheta^2_d) - \bar{\vx}_j)^2 &\geq (\vx_j(\vtheta^1_d) - \bar{\vx}_j)^2
	\end{align*}
	
	Now consider the set $J^{-}$ of traffic flow functions that are monotonically decreasing respect to $\theta_d \in \mathbb{R}$. Since $\dfrac{\partial \vx_j}{\partial \theta_d}\Big|_{\vtheta = \vtheta^1_d}< 0, \forall j \in J^{-}$, the LHS in Eq. \ref{eq:pseudo-convexity-proof1} becomes non negative in the following two cases: \\

	
	$\bullet$ Case (iii): $\vx_j(\vtheta^{1}_d)-\bar{\vx}_j \geq 0 \land \theta_d^{2} \leq \theta_d^{1}$. Because the decreasing monotonicity of the traffic flow functions $j \in J^{-}$ respect to $\theta_d$ and given that $\theta_d^{2} \leq \theta_d^{1}$:
	
	\begin{align*}
		\vx_j(\vtheta^2_d) &\geq \vx_j(\vtheta^1_d)\\
		\vx_j(\vtheta^2_d) - \bar{\vx}_j &\geq  \vx_j(\vtheta^1_d) - \bar{\vx}_j
	\end{align*}
	
	Since $\vx(\vtheta^{1}_d)-\bar{\vx} \geq 0$:
	\begin{align*}
		\vx_j(\vtheta^2_d) - \bar{\vx}_j &\geq  \vx_j(\vtheta^1_d) - \bar{\vx}_j \geq 0 \\
		(\vx_j(\vtheta^2_d) - \bar{\vx}_j)^2 &\geq (\vx_j(\vtheta^1_d) - \bar{\vx}_j)^2
	\end{align*}
	
	$\bullet$ Case (iv): $\vx_j(\vtheta^{1}_d)-\bar{\vx}_j \leq 0 \land  \theta_d^{2} \geq \theta_d^{1}$. Because the decreasing monotonicity of the traffic flow functions $j \in J^{-}$ respect to $\theta_d$ and given that $\theta_d^{2} \geq \theta_d^{1}$:
	
	\begin{align*}
		\vx_j(\vtheta^2_d) &\leq \vx_j(\vtheta^1_d)\\
		\vx_j(\vtheta^2_d) - \bar{\vx}_j &\leq  \vx_j(\vtheta^1_d) - \bar{\vx}_j
	\end{align*}
	
	Since $\vx_j(\vtheta^{1}_d)-\bar{\vx}_j \leq 0$:
	\begin{align*}
		\vx_j(\vtheta^2_d) - \bar{\vx}_j &\leq  \vx_j(\vtheta^1_d) - \bar{\vx}_j \leq 0 \\
		(\vx_j(\vtheta^2_d) - \bar{\vx}_j)^2 &\geq (\vx_j(\vtheta^1_d) - \bar{\vx}_j)^2
	\end{align*}
	
	Thus, if the traffic flow functions are monotonic, the following condition holds $\forall j \in J^{-} \cup J^{+}$: 
	\begin{align*}
		(\vx_j(\vtheta^2_d) - \bar{\vx}_j)^2 &\geq (\vx_j(\vtheta^1_d) - \bar{\vx}_j)^2
		\implies 
		\sum_{j \in J^{-} \cup \ J^{+}}(\vx_j(\vtheta^2_d) - \bar{\vx}_j)^2 \geq \sum_{j \in J^{-}\cup \ J^{+}} (\vx_j(\vtheta^1_d) - \bar{\vx}_j)^2
	\end{align*}
	
	Putting altogether: 
	
	$$
	2\left(\vx(\vtheta^{1}_d)-\bar{\vx} \right)(\theta^{2}_d-\theta^{1}_d) \dfrac{\partial \vx}{\partial \theta_d}\Bigg|_{\theta = \theta^1_d}
	\geq \vzero
	\implies 
	\|\vx(\vtheta^{2}_d)-\bar{\vx}\|^2_2   \geq \|\vx(\vtheta^{1}_d)-\bar{\vx}\|^2_2
	$$
	
	which proves the pseudo-convexity of the objective function of the \LUE problem respect to respect to $\theta_d$. The analysis conducted for $\theta_d$ can be applied to every coordinate $d \in D$ of $\vtheta \in \sR^{|D|}$, which proves the coordinate-wise pseudo-convexity of the objective function $\ell$ respect to an arbitrary utility function coefficient $\vtheta$.

\end{proof}

\begin{assumption}[Range of response functions]
	\label{assumption:traffic-count-functions-reach-link-counts}
	The range of each response function include the value of the traffic count measurement
\end{assumption}

\begin{remark}
	Assumption \ref{assumption:traffic-count-functions-reach-link-counts} may be tested by checking that there exist values of the vector of utility function coefficients where each response function matches the value of the corresponding traffic count measurement. Note that this condition is analyzed for each traffic flow function independently, and thus, the vector of utility function coefficients does not need to be same for all traffic flow functions. An alternative way to test this assumption would be to approximate the ranges of the traffic flow function with the bounds derived from Proposition \ref{prop:upper-lower-bounds-uncongested-network}, Section \ref{sssec:non-convexity} and then checking if the traffic count measurements fall within those ranges. Furthermore, if the traffic flow functions are dependent on a single attribute, their range could be found by evaluating them at the extreme cases where $\theta \to \infty^-$ and $\theta \to \infty^+$. 
\end{remark}


\begin{prop}[Coordinate-wise vanishing gradient of objective function]
	\label{prop:existence-local-optima-uncongested-network-with-error}
	Suppose that each traffic flow function $x_n(\vtheta), \forall n \in N$ is coordinate-wise monotonic and that the range of each includes the value of the link count measurement $\bar{x}_i,\forall i \in N$ (Assumption \ref{assumption:traffic-count-functions-reach-link-counts}). Then, the gradient of the \LUE objective function in an uncongested network vanishes coordinate-wise at least once.
\end{prop}


\begin{proof}
	%
	%
	%
	Let's be $\theta_d \in \mathbb{R}$ the value of an arbitrary coordinate $t$ of the vector $\vtheta \in \mathbb{R}^{|D|}$ of travelers' utility function coefficients. Assume that the values of all coefficients except for $\theta_d$ are kept constant and that the minimization of the objective function $\ell: \mathbb{R}^{|D|} \to \mathbb{R}$ is performed respect to $\theta_d$ only. Let's now split the set of traffic flow functions $N$ between those that are monotonically decreasing and increasing and denote each set as $N^{\minus}$ and $N^{+}$. Then, the expression of the first derivative of the objective function $\ell: \mathbb{R}^{|D|} \to \mathbb{R}$ can be decomposed as follows:
	
	$$
	\dfrac{\partial \ell(\vtheta)}{\partial \theta_d} = 
	\displaystyle 2\sum_{i \in N^{\minus}} \dfrac{\partial x_i(\vtheta)}{\partial \theta_d} 
	\left(\bar{\vx}_i- \vx_i(\vtheta) \right)
	+ 2\sum_{i \in N^{+}} \dfrac{\partial x_i(\vtheta)}{\partial \theta_d} 
	\left(\bar{\vx}_i- \vx_i(\vtheta) \right)
	$$
	
	
	By assumption, the ranges of the traffic flow functions include the values of the link count measurements, hence for each traffic flow function $i \in N$, $\exists \bar{\vtheta}^i \in \sR^{|D|}: \bar{\vx}_i- \vx_i(\bar{\vtheta^i})= 0$. Note that by the coordinate-wise increasing monotonicity of the set of traffic functions in $i \in N^{+}$,  $\ \bar{\vx}_i- \vx_i(\vtheta) > 0$ if $\theta_d > \bar{\theta}^i_d$. Conversely, if $i \in N^{\minus}$ and $\theta_d > \bar{\theta}_d$,  then $\ \bar{\vx}_i- \vx_i(\vtheta) < 0$.
	
	Let's define $\{\bar{\theta}_d\}_{i \in N}$ as the set that contains the values of the coefficient $\bar{\theta}_d$ that satisfies that $\bar{\vx}_i- \vx_i(\bar{\vtheta^i})= 0$ for the traffic flow function $i \in N$. Define $\bar{\theta}_d^{+}$ as the maximum value in the set and then add to it some arbitrary quantity $\epsilon > 0$. Let's define $\tilde{\vtheta}^{+} \in \mathbb{R}^{|D|}$ as the vector of utility function coefficients associated to the traffic flow function $i$ where $\tilde{\theta}^{+}_d  = \bar{\theta}_d^{+} + \epsilon$. Then, the first derivative of $f$ evaluated at this point can be expressed as:
	
	$$
	\dfrac{\partial \ell(\tilde{\vtheta}^{+}  )}{\partial \theta_d}  = 
	\displaystyle 2\sum_{i \in N^{\minus}} \dfrac{\partial x_i(\tilde{\vtheta}^{+})}{\partial \theta_d} 
	\left(\bar{\vx}_i- \vx_i(\tilde{\vtheta}^{+}) \right)
	+ 2\sum_{i \in N^{+}} \dfrac{\partial x_i(\tilde{\vtheta}^{+})}{\partial \theta_d} 
	\left(\bar{\vx}_i- \vx_i(\tilde{\vtheta}^{+}) \right)
	$$
	
	By the increasing and decreasing coordinate-wise monotonicity of the traffic functions within the sets $N^{\minus}$ and $N^{+}$, $\dfrac{\partial x_i(\vtheta_d)}{\partial \vtheta_d} < 0$ if $i \in N^{\minus}$  and $\dfrac{\partial x_i(\theta_d)}{\partial \theta_d} > 0$  if $i \in N^{+}$. Then, we have that $\theta_d = \bar{\theta}_d^{+} + \epsilon \implies \dfrac{\partial f(\tilde{\vtheta}^{+})}{\partial \theta_d}  \geq 0$. Conversely, if for $\tilde{\vtheta}^{\minus} \in \mathbb{R}^{|D|}, \ \tilde{\theta}^{\minus}_d = \min\left(\{\bar{\theta}^i_d\}_{i \in N} \right) $, then $\dfrac{\partial f(\tilde{\vtheta}^{\minus})}{\partial \theta_d} \leq 0$. Finally, since the first derivative of $f$ is continuous and it changes sign at $\tilde{\vtheta}^{+}, \tilde{\vtheta}^{\minus}, \in \mathbb{R}^{|D|}$, by intermediate value theorem, it must vanish at least once at some feasible point $\tilde{\vtheta} \in \mathbb{R}^{|D|}$. The analysis conducted for $\theta_d$ can be applied to every coordinate $d \in D$ of $\vtheta \in \sR^{|D|}$, which completes the proof.

	
	
	
	
\end{proof}


\subsection{Existence and uniqueness of global minima under no measurement error}
\label{ssec:existence-local-minima-small-network}


Proofs of existence and uniqueness of a global optima typically relies on the convexity of the optimization problem. Despite the non-convexity of our problem, we can rely on Assumption \ref{assumption:sue}, Section \ref{sec:learning-parameters} of \SUE-\logit and to assume absence of measurement error to prove existence and uniqueness:



\begin{prop}[Existence of global minima with no noise in traffic counts]
	\label{prop:existence-uncongested-network-no-error}
	The \LUE problem under an uncongested network has a global optima if traffic count measurements follow \SUE-\logit and they have no measurement error
\end{prop}

\begin{proof}
	By assumption, if $\bar{\vx}$ follows \SUE-\logit,  $\exists \vtheta^{\star} \in \sR^{|D|}: \vx(\vtheta^{\star}) = \bar{\vx} $. Hence, the objective function reaches its lower bound when $f(\vtheta^{\star}) = 0$ and thus $\vtheta^{\star}$ is a global minimizer, which completes the proof. 
\end{proof}
%

\begin{prop}[Uniqueness of global minima with no noise in traffic counts]
	\label{prop:uniqueness-uncongested-network-no-error}
	The \LUE problem under an uncongested network has a unique global optima at $\vtheta^\star \in \mathbb{R}^{|D|}$ if (i) the traffic count measurements $\bar{\vx}$ follow \SUE-\logit and they have no measurement error, (ii) the Jacobian matrix of the objective function has full rank at any feasible point $\vtheta \in \sR^{|D|}$, (iii) the response functions $\vx(\vtheta)$ are strictly monotone
\end{prop}

\begin{proof}
	By Proposition \ref{prop:existence-uncongested-network-no-error} and the assumption that $\bar{\vx}$ follows \SUE-\logit, the optimization problem has a global optima  at $\vtheta^{\star}$. To prove uniqueness, we first derive the first order necessary optimality condition:
	\begin{align}
		\label{eq:first-order-condition-global-minima-no-error}
		\grad_{\vtheta} \| \vx(\vtheta) - \bar{\vx}\|_2^2
		=
		-2\left[D_{\vtheta} \phantom{'}\vx(\vtheta)\right]^{\top} 
		\left(\bar{\vx}- \vx(\vtheta) \right) &= \vzero 
	\end{align}
	
	By assumption, the Jacobian matrix $D_{\vtheta} \ \vx(\vtheta)$ is full rank $\forall \vtheta \in \sR^{|D|}$ and thus $[D_{\vtheta} \ \vx(\vtheta)]^\top \left(\bar{\vx}- \vx(\vtheta) \right) = \vzero$ iff $\vx(\vtheta) = \bar{\vx}$. By the strict monotonicity of $\vx(\vtheta)$, there exists a unique value of $\vtheta \in \mathbb{R}^{|D|}$ such that  $\vx(\vtheta) = \bar{\vx}$. Therefore, this proved that $\vtheta = \vtheta^{\star}$ is the unique the global optima.

\end{proof}



\subsection{Connection between \OLS and \NLLS}
\label{ssec:equivalence-ols-nlls}


Consider the vectorized version of the \NLLS regression equation (Eq. \ref{eq:regression-equation-nlls-learning-problem}, Section \ref{ssec:nlls-estimator}) and let's expressed it in terms of the response function $\vx(\vtheta): \sR^{|D|} \to \sR^{|A|}  $ defined in Eq. \ref{eq:NLLS-link-flow-equation}, Section \ref{sssec:nlls-problem-formulation}. Let's compute the first order Taylor approximation of the response function respect to $\vtheta \in \mathbb{R}^{|D|}$ and around an arbitrary vector $\vtheta_{\vzero} \in \mathbb{R}^{|D|}$:
%
			\begin{equation}
				\label{eq:taylor-approximation-response-function-learning-problem}
				\vx(\vtheta) \approx \vx(\vtheta_{\vzero}) + D_{\vtheta} \ x(\vtheta) \Big|_{\vtheta_{\vzero}} (\vtheta-\vtheta_{\vzero}) 
			\end{equation}

where, for the ease of notation, $\vx(\vtheta) = \vx(\mZ, \vt, \vtheta)$. Replacing back into the \NLLS regression equation (Eq. \ref{eq:regression-equation-nlls-learning-problem}, Section \ref{ssec:nlls-estimator}):
		\begin{align}
			\label{eq:ols-equivalence-learning-problem}
			\rvx &\approx x(\vtheta_{\vzero}) + D_{\vtheta } \ x( \vtheta)\Big|_{\vtheta_{\vzero}} (\vtheta-\vtheta_0) + \rvu \nonumber \\
			\rvx - x(\vtheta_{\vzero}) + D_{\vtheta} \ x( \vtheta)\Big|_{\vtheta_{\vzero}} \vtheta_{\vzero} &\approx D_{\vtheta} \ x(\vtheta)\Big|_{\vtheta_{\vzero}} \vtheta + \rvu \nonumber \\
			\tilde{\rvx}  &\approx \tilde{\mX} \vtheta + \rvu
		\end{align}

From where is clear that Eq. \ref{eq:ols-equivalence-learning-problem} resembles the \OLS regression equation, except that $\tilde{\mX} = D_{\vtheta} \ x(\vtheta)\big|_{\vtheta_{\vzero}}$.

\section{Algorithms}
\label{appendix:sec:implementation}


\subsection{Stochastic network loading}
\label{appendix:ssec:snl}

\begin{algorithm}[H] 
	
	\captionsetup{font=normalsize}
	\caption{Stochastic network loading (\SNL)}
	\label{alg:stochastic-network-loading}
	\scalebox{1}{%
		\begin{minipage}{\linewidth}
			\begin{algorithmic} 
				\Require{$\vtheta \in \sR^{|D|}$, incident matrices $\mIq,\mIx$, non-sparse O-D vector $\vq$, vector of link travel times $\bar{\vt}$, matrix of exogenous link attributes $\mZ$}
				
				\begin{algsubstates}
					\State Travel time initialization: If $\bar{\vt} = \emptyset$, then $\bar{\vt} = \bar{\vt}_f$, where $\bar{\vt}_f$ is the vector of links' free flow  travel times.\\
					\State Computation of link utilities
					$$\vv_{x} \gets \theta_t \bar{\vt} +  \vtheta_Z ^{\top}\mZ$$
					\State Computation of path choice probabilities: 
					$$
					\vpf \gets \exp\left(\displaystyle \mIxT \vv_{x}   \right) \oslash \left(\displaystyle   \mIqT \mIq \exp(\mIxT \vv_{x})\right)
					$$
					\State Computation of path flows: 
					$$
					\displaystyle 
					\vf \gets (\mIqT \vq) \circ  \vpf
					$$
					\State Computation of link flows:
					$$
					\vx  \gets \mIx \vf
					$$
					
				\end{algsubstates}
				
				\State 
				\Return {$\vx, \vf, \vpf$ }
			\end{algorithmic}
		\end{minipage}%
	}
\end{algorithm}





where $\oslash$ is the operator for element-wise division, $\vv_{x} \in \sR^{|A|}$  is the vector of link utilities, $\vpf \in \sR_{]0,1[}^{|H|}$ is a vector of path choice probabilities and $\bar{t}$ is the vector of travel times which is assumed to be exogenous during \SNL. Note that the definition of $\vpf$ in the first step of \SNL correspond to the vectorized form of the path flows at \SUE-\logit presented in Eq. \ref{eq:sue-logit-path-flows-solution} and which is written as a function of link utilities and the network incidence matrices $\mIx, \mIq$.


\subsection{Inner level optimization}
\label{appendix:ssec:implementation-inner-level-optimization}

%


\newlength\myindent
\setlength\myindent{2em}
\newcommand\bindent{%
  \begingroup
  \setlength{\itemindent}{\myindent}
  \addtolength{\algorithmicindent}{\myindent}
}
\newcommand\eindent{\endgroup}

\begin{algorithm}[H] \small %
	\captionsetup{font=normalsize}
	\caption{\texttt{InnerLevelOptimization}}
	\label{alg:inner-level-optimization}
	\scalebox{1}{%
		\begin{minipage}{\linewidth}
			\begin{algorithmic} 
				\Require{\# Iterations $T$, initial vector of estimated coefficients ${\hat{\vtheta}} \in \sR^{|D|}$, matrix of exogenous attributes $\mZ$, vector of links' free flow travel times $\vt^0$, vector of link capacities $\vgamma$, incidence matrices $\mIq,\mIx$, dense O-D vector $\vq$, grid of values $\vlambda_{FW}\in \sR^{|D|}$ in Frank-Wolfe algorithm, proportion of selected O-D pairs in column generation phase $\rho_W$, proportion of generated and selected paths in column generation phase  $k_g, k_s$:}
				
				\State
				
				\State Step 0: Initialization. 
				
				\begin{algsubstates}
					\State Compute initial vector of link utilities 	$\vv_{x} \gets \theta_t \bar{\vt}_f +  {\hat{\vtheta}}_Z ^{\top}\mZ$, where ${\hat{\vtheta}} = [\theta_t \ \theta^{\top}_Z], \ \theta_t \in \sR, \ \theta_Z \in \sR^{|K_{\mZ}|}$
					\State Generate $k\minus$shortest paths $\mathcal{S}_{pq}, \ \forall (p,q) \in W$ based on link utilities $\vv_{x}$
					\State Compute incident matrices $\mIq,\mIx$ from $\mathcal{S}_{pq}$
					\State Perform stochastic network loading: $\vx^{(0)}, \vf^{(0)} \gets  \SNL({\hat{\vtheta}}, \mIq, \mIx, \bar{\vt})$
					\State $i = 0$.
				\end{algsubstates}
				
				\item[]
				
				\State Step 1: Column generation phase. 
				
				\begin{algsubstates}
     					\State Select subset of O-D pairs $W_s \in W$ with the highest travel demand, such that $|W_s| = |W|\rho_W$ 
					\State Generate set of the $k_g\minus$shortest paths $\mathcal{C}_{pq}, \forall (p,q) \in W_s$ based on current link utilities $\vv_{x}$
                    \State Update path sets $\mathcal{S}_{pq} \gets \mathcal{C}_{pq},\forall (p,q) \in W_s$
					
					\State Update incidents matrices $\mIq, \mIx$ from $\mathcal{S}_{pq}, \forall p,q \in W$
					
					
				\end{algsubstates}

                \State
				
				\State Step 2: \SUE-\logit

				\begin{algsubstates}
					
					
					\State Perform stochastic network loading (algorithm \ref{alg:stochastic-network-loading}): 
					
					$$
					\vx, \vf, \vpf \gets  \SNL({\hat{\vtheta}}, \mIq, \mIx, \vq, \bar{\vt}, \mZ)				
					$$

					
					\State Update link travel times
					$$
					\bar{\vt} \gets \bar{\vt}^0(1+\alpha(\vx/\vgamma)^\beta)
					$$
					
					\State Solve a linear search problem: 
					\item[]
					
					\For{$t = 1$ to $|\vlambda_{FW}|$}{
					
					\item[] $\quad \quad \bar{\vf} \gets \lambda_i\bar{\vf}^{(i-1)} + (1-\lambda_i)\bar{\vf}^{(i)}$
					\item[] $\quad\quad \displaystyle \bar{\vx} \gets \mIx \bar{\vf}$
					\item[] $\quad\quad \displaystyle \ell[t] \gets \sum_{a \in A}\int_{0}^{\bar{x}_a} v_a(u,{\hat{\vtheta}}) du - \left\langle \bar{\vf}, \ln \bar{\vf}\right\rangle$
					%
					
					}
					\item[]
					\item[] $
					\lambda^{\star}_i = \arg \min_{\lambda \in \vlambda_g}  Z(\lambda) 
					$		

                    \item[]
     
					\State Update path and link flow solutions, travel times and path choice probabilities 
					
					$$
					\vf^{(i+1)} \gets \lambda^{\star}_i\vf^{(i-1)} + (1-\lambda^{\star}_i)\vf^{(i)}
					$$
					$$
					\vx \gets \mIx \vf
					$$
					$$
					\bar{\vt} \gets \bar{\vt}^0(1+\alpha(\vx/\vgamma)^\beta)
					$$
					$$
					\vpf \gets \exp\left(\displaystyle \mIxT \vv_{x}   \right) \oslash \left(\displaystyle   \mIqT \mIq \exp(\mIxT \vv_{x})\right)
					$$

     \State Go to step 3 if desire level of accuracy have been achieved\footnote{For example, when the relative decrease of the objective function given by $1-\ell[t]^{(i)}/\ell[t]^{(i-1)}$ is lower than a threshold}. Otherwise, $i= i+1$ and start again from Step 2
					
				\end{algsubstates}

                \State
        
        \State Step 3: Paths selection

                \begin{algsubstates}

                \State Select the $k_s\minus$shortest  paths $\mathcal{R}_{pq}, \forall p,q \in W_s$ 
                \State Update path sets $ \mathcal{S}_{pq} \gets \mathcal{R}_{pq} , \forall p,q \in W_s$


                \State Update incidents matrices $\mIq, \mIx$ from $\mathcal{R}_{pq}, \forall p,q \in W$

                \end{algsubstates}
                
				\State			
				
				\Return {$\vx,\vpf, \bar{\vt}$}
			\end{algorithmic}
		\end{minipage}%
	}
\end{algorithm}

\clearpage

%

%

\subsection{Outer level optimization}
\label{appendix:ssec:implementation-outer-level-optimization}


\subsubsection{Algorithm}

\begin{algorithm}[!htbp] 
	\captionsetup{font=normalsize}
	\caption{\texttt{OuterLevelOptimization}}
	\label{alg:outer-level-optimization}
	\scalebox{1}{%
		\begin{minipage}{\linewidth}
			\begin{algorithmic} 
				\Require{\# Iterations in no-refined and refined stages $T_1, T_2$, initial vector of estimated coefficients ${\hat{\vtheta}}_0 \in \sR^{|D|}$, link flows and path choice probabilities $\vx,\vpf$ from inner level problem, choice of optimization algorithm the refined stage (\texttt{refined-method}), vector of observed traffic counts $\bar{\vx}$, learning rates $\eta_1$, dumping parameter for \LM method  $\delta_{LM}$}
				
				\item[]
				
				\State Step 1: no-refined stage
				
				\For{$t = 1 \ldots T_1$}
				
				\State ${\hat{\vtheta}}_{t+1} \gets$\texttt{FirstOrderOptimization}( ${\hat{\vtheta}}$, $\eta_1$, $\bar{\vx}$, $\vx$, $\vpf$, $T = 1$)
				
				
				
				\EndFor
				
				\item[]
				
				\State Step 2: refined stage
				
				\For{$t = 1 \ldots T_2$}
				
				%
				
				\State 
				${\hat{\vtheta}}_{T_1+t+1} \gets$\texttt{SecondOrderOptimization}(method = \LM, ${\hat{\vtheta}}_{T_1+t}$, $\delta_{LM}$, $\bar{\vx}$, $\vx$, $\vpf$, $T = 1$)
				
				
				
				\EndFor
				
				\State
				
				\State 
				\Return {$\bar{{\hat{\vtheta}}}_T = \argmin_{\{{\hat{\vtheta}}_1,\ldots,{\hat{\vtheta}}_{T_1+T_2} \}} \ell_t({\hat{\vtheta}}_t) $ }
			\end{algorithmic}
		\end{minipage}%
	}
\end{algorithm}

\subsubsection{Gradients and Jacobian}
\label{appendix:sssec:gradients}

The first derivative of the objective function $\ell(\cdot)$ respect to a utility function coefficient $\hat{\theta}_d \in \mathbb{R}, \forall d \in K$ can be written in vectorized form as:

%

\begin{equation}
	\dfrac{\partial \ell(\hat{\vtheta})}{\partial \hat{\theta}_d}
	= \dfrac{\partial }{\partial \hat{\theta}_d}\big\|\vx(\hat{\vtheta})-\bar{\vx}\big\|_2^2
	= 2 \left(\dfrac{\partial \vx(\hat{\vtheta})}{\partial \hat{\theta}_d} \right)^\top \left(\vx(\hat{\vtheta})-\bar{\vx} \right)
\end{equation}

where $\vx(\cdot)$ is a vector valued response function that receives as input a vector with the utility function coefficients and it returns a vector $\hat{\vtheta} \in \mathbb{R}^{|D|}$ with the predicted traffic counts among all links in the transportation network:

\begin{equation}
	\vx(\hat{\vtheta}) =  \mIx \left( (\mIqT \vq) \circ  \vpf(\hat{\vtheta}) \right)
\end{equation}

and $\vpf(\hat{\vtheta})$ is a vector valued function that receives as input a vector with the utility function coefficients and it returns a vector with the choice probabilities associated to all paths in the transportation network:
\begin{align}
	\label{eq:NLLS-path-choice-probabilities-equation}
	\vpf(\hat{\vtheta}) &= \exp\left(\displaystyle \mIxT \vv_{x}(\hat{\vtheta}, Z, \bar{t})     \right) \oslash \left(\displaystyle   \mIqT \mIq \exp(\mIxT \vv_{x}(\hat{\vtheta}, \mZ, \bar{\vt})  )\right)
\end{align}

The first derivative of $\vx(\cdot)$ respect to a utility function coefficient $\hat{\theta}_d$ can be written in vectorized form as:

\begin{equation}
	\dfrac{\partial \vx(\hat{\vtheta})}{\partial \hat{\theta}_d} = \mIx \left( (\mIqT \vq)  \circ  \dfrac{\partial \vpf(\hat{\vtheta})}{\partial \hat{\theta}_d}  \right) 
\end{equation}

where

\begin{equation}
	\dfrac{\partial \vpf(\hat{\vtheta})}{\partial \theta_d} =  \left((\mIqT \mIq)\circ \left(\vpf(\hat{\vtheta}) \vpf(\hat{\vtheta})^{\top}\right) \circ   \left[\mZ_d\vone_{|\mZ_d|}^\top -\vone_{|\mZ_d|}\mZ_d^\top \right] \right) \vone_{|\mZ_d|}
\end{equation}

is a vector with the first derivatives of the path choice probabilities respect to the utility function coefficient $\hat{\theta}_d$. $\mZ_d \in \mathbb{R}^{|H|}$ is a column vector with the values of attribute $d \in D$ among all paths and $\vone_{|\mZ_d|} \in \mathbb{R}^{|H|}$ is a column vector with ones. Then, the gradient $\grad_{\hat{\vtheta}} \ \vx(\hat{\vtheta})$ associated to a traffic flow function $i$ is obtained by stacking the first derivatives in a column vector as follows:

\begin{equation}
	\grad_{\hat{\vtheta}} \ x_i(\hat{\vtheta})	= 
	\begin{bmatrix}
		\dfrac{\partial x_i(\hat{\vtheta})}{\partial \hat{\theta}_1} & \ldots & \dfrac{\partial x_i(\hat{\vtheta})}{\partial \hat{\theta}_D}
	\end{bmatrix}^{\top}
\end{equation}

and the Jacobian matrix $D_{\hat{\vtheta}} \ \vx(\hat{\vtheta}) \in \mathbb{R}^{n \times d}$ associated to the traffic flow functions are simply the stacked gradient vectors for each observation $n \in N$:

\begin{equation}
	D_{\hat{\vtheta}} \ x(\hat{\vtheta})
	= \begin{bmatrix}
		\grad_{\hat{\vtheta}} \ x_1(\hat{\vtheta}) ^{\top} \\ \vdots \ \ \\ \\ \grad_{\hat{\vtheta}} \ x_n(\hat{\vtheta}) ^{\top}
	\end{bmatrix}
	= \begin{bmatrix}
		\dfrac{\partial x_1(\hat{\vtheta})}{\partial \hat{\theta}_1} & \ldots & \dfrac{\partial x_1(\hat{\vtheta})}{\partial \hat{\theta}_D} \\ 
		\vdots & & \\
		\dfrac{\partial x_n(\hat{\vtheta})}{\partial \hat{\theta}_1} & \ldots & \dfrac{\partial x_n(\hat{\vtheta})}{\partial \hat{\theta}_D}\\
		
	\end{bmatrix}^{\top}
\end{equation}

Finally, the analytical gradient of the objective function respect to the vector of utility function coefficients is :

\begin{equation}
	\grad_{\hat{\vtheta}} \ \ell(\hat{\vtheta})	= 
	\begin{bmatrix}
		\dfrac{\partial \ell(\hat{\vtheta})}{\partial \hat{\theta}_1} & \ldots & \dfrac{\partial \ell(\hat{\vtheta})}{\partial \hat{\theta}_D}
	\end{bmatrix}^{\top}
\end{equation}

and where

\begin{align}
	\dfrac{\partial \ell(\hat{\vtheta})}{\partial \hat{\theta}_d}
	&= 2 \left(\dfrac{\partial \vx(\hat{\vtheta})}{\partial \hat{\theta}_d} \right)^\top \left(\vx(\hat{\vtheta})-\bar{\vx} \right)\\ \nonumber
	&= 2 \left(\mIx \left( (\mIqT \vq)  \circ  \dfrac{\partial \vpf(\hat{\vtheta})}{\partial \hat{\theta}_d}  \right) \right)^\top \left(\vx(\hat{\vtheta})-\bar{\vx} \right) \\ \nonumber
	&= 2 \left(\mIx \left( (\mIqT \vq)  \circ  \left((\mIqT \mIq)\circ \left(\vpf(\hat{\vtheta}) \vpf(\hat{\vtheta})^{\top}\right) \circ   \left[\mZ_d\vone_{|\mZ_d|}^\top -\vone_{|\mZ_d|}\mZ_d^\top  \right] \right) \vone_{|\mZ_d|}  \right) \right)^\top \left(\vx(\hat{\vtheta})-\bar{\vx} \right)
\end{align}

is a scalar.

\subsubsection{Second derivatives}
\label{appendix:sssec:second-derivatives}

The second derivative of the objective function $\ell(\cdot)$ respect to a utility function coefficient $\hat{\vtheta}_d \in \mathbb{R}, \forall d \in D$ can be written in vectorized form as:
\begin{equation}
	\dfrac{\partial^2 \ell(\hat{\vtheta})}{\partial^2 \hat{\vtheta}_d}
	= \dfrac{\partial }{\partial \hat{\vtheta}_d} \left(\dfrac{\partial \ell(\hat{\vtheta})}{\partial \hat{\vtheta}_d}\right) 
	= 2\bigg(\left(\dfrac{\partial^2 \vx(\hat{\vtheta})}{\partial^2 \hat{\vtheta}_d}\right)^\top \left(\vx(\hat{\vtheta})-\bar{\vx} \right)+\left(\dfrac{\partial \vx(\hat{\vtheta})}{\partial \hat{\vtheta}_d}\right)^\top \dfrac{\partial \vx(\hat{\vtheta})}{\partial \hat{\vtheta}_d}
	\bigg)
\end{equation}

where 

\begin{equation}
	\dfrac{\partial^2 \vx(\hat{\vtheta})}{\partial^2 \hat{\vtheta}_d} = \mIx \left( (\mIqT \vq) \circ  \dfrac{\partial^2  \vpf(\hat{\vtheta}) }{\partial^2 \hat{\vtheta}_d}\right)
\end{equation}

and

\begin{equation}
	\dfrac{\partial^2  \vpf(\hat{\vtheta}) }{\partial^2 \hat{\vtheta}_d}
	= (\mIqT \mIq) \circ \left (\left(\dfrac{\partial  \vpf(\hat{\vtheta}) }{\partial \hat{\vtheta}_d} \right)\vpf(\hat{\vtheta})^{\top}+\vpf(\hat{\vtheta})\left(\dfrac{\partial  \vpf(\hat{\vtheta}) }{\partial \hat{\vtheta}_d}\right)^\top \right) \circ  \left[\mZ_d\vone_{|\mZ_d|}^\top -\vone_{|\mZ_d|}\mZ_d^\top \right]\vone_{|\mZ_d|} 
\end{equation}

is a vector of dimension $H$ with the second derivatives of the path choice probabilities respect to a utility function coefficient $\hat{\vtheta}_d$.

\subsubsection{First order optimization methods}
\label{appendix:ssec:first-order-optimization-methods}

Algorithm \ref{alg:first-order-optimization} shows the pseudo code to implement the first order optimization algorithms. Note that the only difference of normalized gradient descent (\NGD) with vanilla gradient descent is the step of normalization of the gradient.

\begin{algorithm}[H] 
	\captionsetup{font=normalsize}
	\caption{\texttt{FirstOrderOptimization}}
	\label{alg:first-order-optimization} 
	\scalebox{1}{%
		\begin{minipage}{\linewidth}
			\begin{algorithmic} 
				\Require{\# Iterations $T$, initial vector of estimated coefficients $\hat{\vtheta}_0 \in \sR^{|D|}$, incident matrices $\mIq,\mIx$, non-sparse O-D vector $\vq$, path choice probabilities $\vpf$, link flows $\vx$, dense vector with O-D demand $\vq$, matrix of exogenous link attributes $\mZ$, vector of link travel times $\bar{\vt}$, vector of observed traffic counts $\bar{\vx}$, learning rate $\eta$}
				
				\item[]
				\For{$t = 1 \ldots T$}
				
				\State Compute stochastic network loading (Algorithm \ref{alg:stochastic-network-loading}, \ref{appendix:ssec:snl})			
				$$
				\vx, \vf, \vpf \gets  \SNL(\hat{\vtheta}, \mIq, \mIx, \vq, \bar{\vt}, \mZ)		
				$$
				
				\State Compute gradient of the objective function (\ref{appendix:sssec:gradients}):
				$$
				\vg_t:= \nabla_{\hat{\vtheta}} \ell (\hat{\vtheta}, \vx, \vq, \vpf, \bar{\vx}, \bar{\vt}, \mZ, \mIq,\mIx) 
				$$

				
				\State Normalization of gradient
				\State $$g_t \gets \frac{g_t}{\|g_t\|}$$

				\State Solution update:\vspace{-0.5cm}
				\State $$\hat{\vtheta}_{t+1} = \hat{\vtheta}_t - \eta g_t$$\vspace{-0.5cm}
				
				\EndFor

				\State 
				
				\Return {$\bar{\hat{\vtheta}}_T = \argmin_{\{\hat{\vtheta}_1,\ldots,\hat{\vtheta}_T \}} \ell_t(\hat{\vtheta}_t) $ }
			\end{algorithmic}
		\end{minipage}%
	}
\end{algorithm}

\subsubsection{Second order optimization methods}
\label{appendix:ssec:second-order-optimization-methods}

Table \ref{alg:second-order-optimization-methods} shows the pseudo code to implement the Gauss-Newton (\GN) and Levenberg–Marquardt (\LM) algorithms.


\begin{algorithm}[H] 
	\captionsetup{font=normalsize}
	\caption{\texttt{SecondOrderOptimization}}
	\label{alg:second-order-optimization-methods}
	
	\scalebox{1}{%
		\begin{minipage}{\linewidth}
			\begin{algorithmic} 
				\Require{\# Iterations $T$, initial vector of coefficients $\hat{\vtheta}_0 \in \sR^{|D|}$, link and path flows $\vx$, $\vf$ from inner level problem, choice of second order optimization  \method, vector of observed traffic counts $\bar{\vx}$, dumping parameter for \LM method $\delta_{LM}$}
				
				\For{$t = 1 \ldots T$}
				
				\State Compute stochastic network loading (Algorithm \ref{alg:stochastic-network-loading}, \ref{appendix:ssec:snl})			
				$$
				\vx, \vf, \vpf \gets  \SNL(\hat{\vtheta}, \mIq, \mIx, \vq, \bar{\vt}, \mZ)		
				$$
				
				\State Compute Jacobian of the traffic flow functions (\ref{appendix:sssec:gradients}): 
				$$
				J_t:= D_{\hat{\vtheta}} \ \vx (\hat{\vtheta}, \vx, \vq, \vpf, \bar{\vx}, \bar{\vt}, \mZ)
				$$
				
				\If{\method = \GN} 
				\State $$\Delta \hat{\vtheta}_t \gets (J_t^\top J_t)^{-1} J_t^\top (\bar{\vx}-\vx) $$
				\EndIf\\ 
				
				\If{\method = \LM} 
				\State $$\Delta \hat{\vtheta}_t \gets (J_t^\top J_t + \delta_{LM} \ I_{d \times d})^{-1} J_t^\top (\bar{\vx}-\vx) $$
				\EndIf\\ 

				\State Update: 
				\State $$\hat{\vtheta}_{t+1} = \hat{\vtheta}_t + \Delta \hat{\vtheta}_t$$ 
				\EndFor

				\State 
				\Return {$\bar{\hat{\vtheta}}_T = \argmin_{\{\hat{\vtheta}_1,\ldots,\hat{\vtheta}_T \}} \ell_t(\hat{\vtheta}_t) $ }
			\end{algorithmic}
		\end{minipage}%
	}
\end{algorithm}


\pagebreak
\section{Statistical analyses in Fresno, CA}

\subsection{Descriptive statistics}
\label{appendix:ssec:descriptive statistics}


\begin{figure}[H]
	\centering
	
	
	\begin{subfigure}[t]{0.75\columnwidth}
		\centering
		\centering
		\includegraphics[width=\columnwidth, trim= {0cm 0cm 0cm 3cm},clip ]{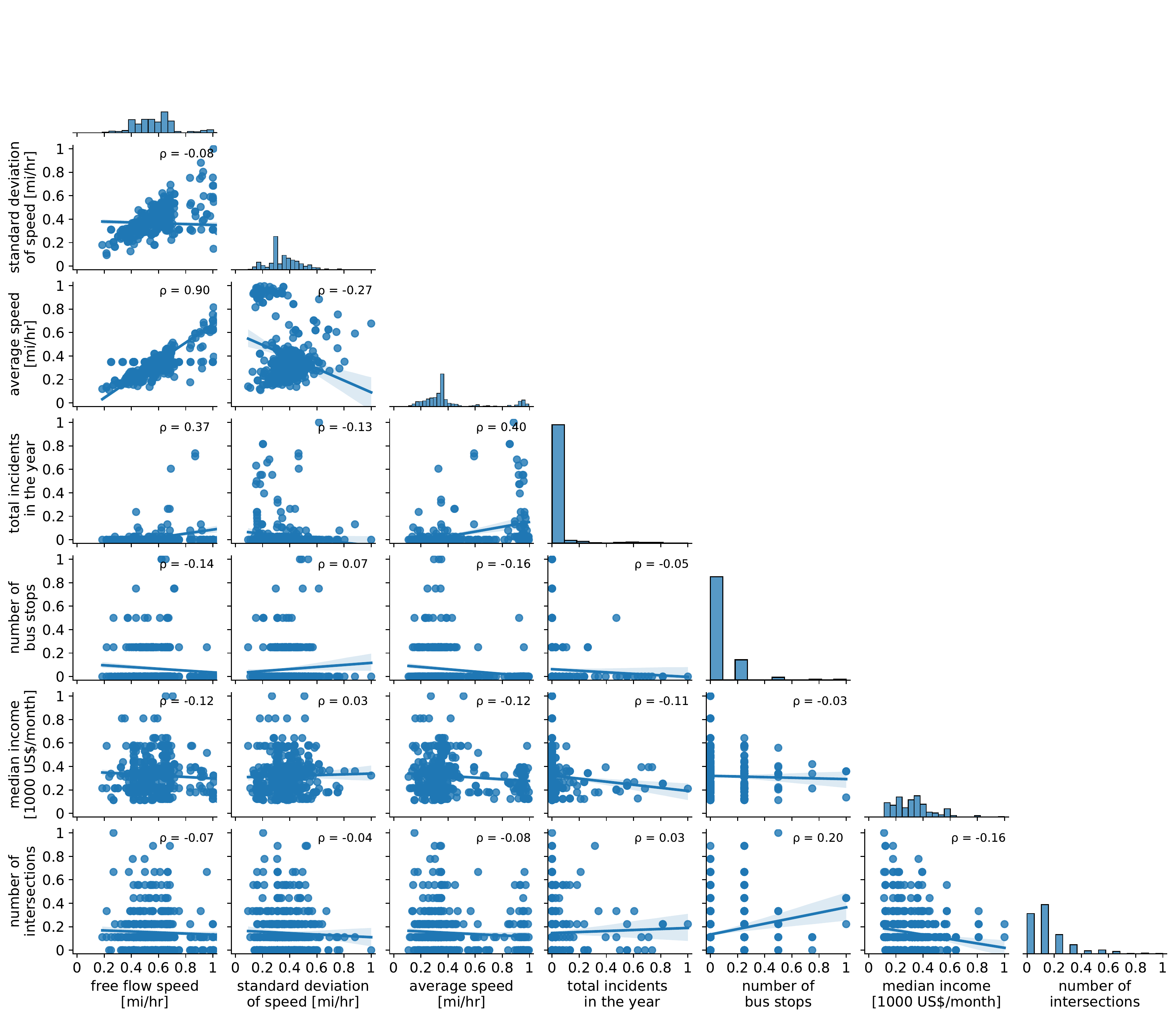}
		\caption{October, 2019}
		\label{fig:features-correlations-2019}
	\end{subfigure}
	
	
	\begin{subfigure}[t]{0.7\columnwidth}
		\centering
		\includegraphics[width=\columnwidth, trim= {0cm 0cm 0cm 3cm},clip ]{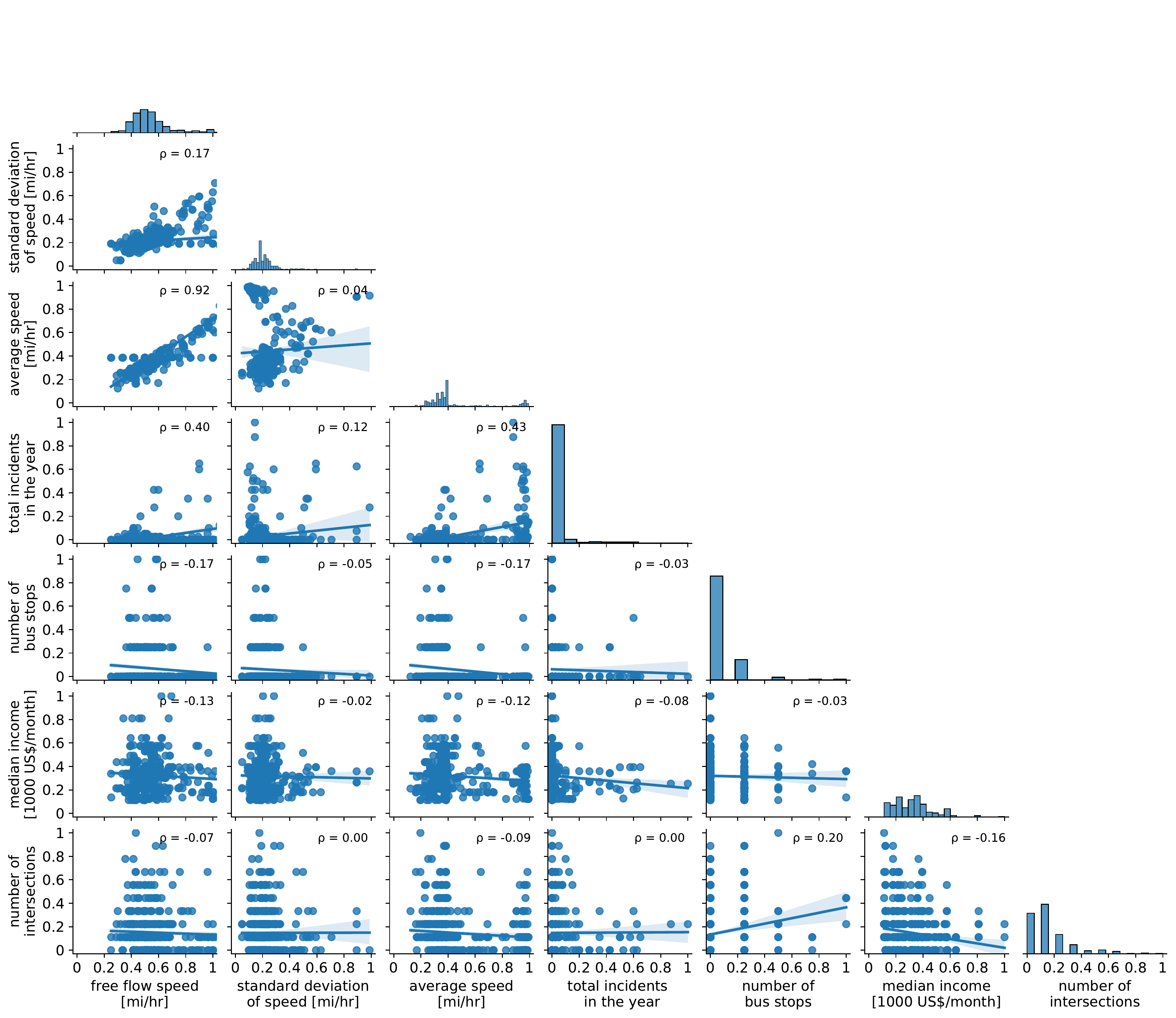}
		\caption{October, 2020}
		\label{fig:features-correlations-2020}
	\end{subfigure}
	
	\caption{Correlation between normalized system level attributes in Fresno, CA}
	
\end{figure}

\subsection{Estimation results}
\label{appendix:ssec:estimation-results}

\begin{figure}[H]
	\centering
	
	
	\begin{subfigure}[t]{0.48\columnwidth}
		\centering
		\includegraphics[width=\columnwidth, trim= {2.3cm 3.1cm 2.3cm 0cm},clip]{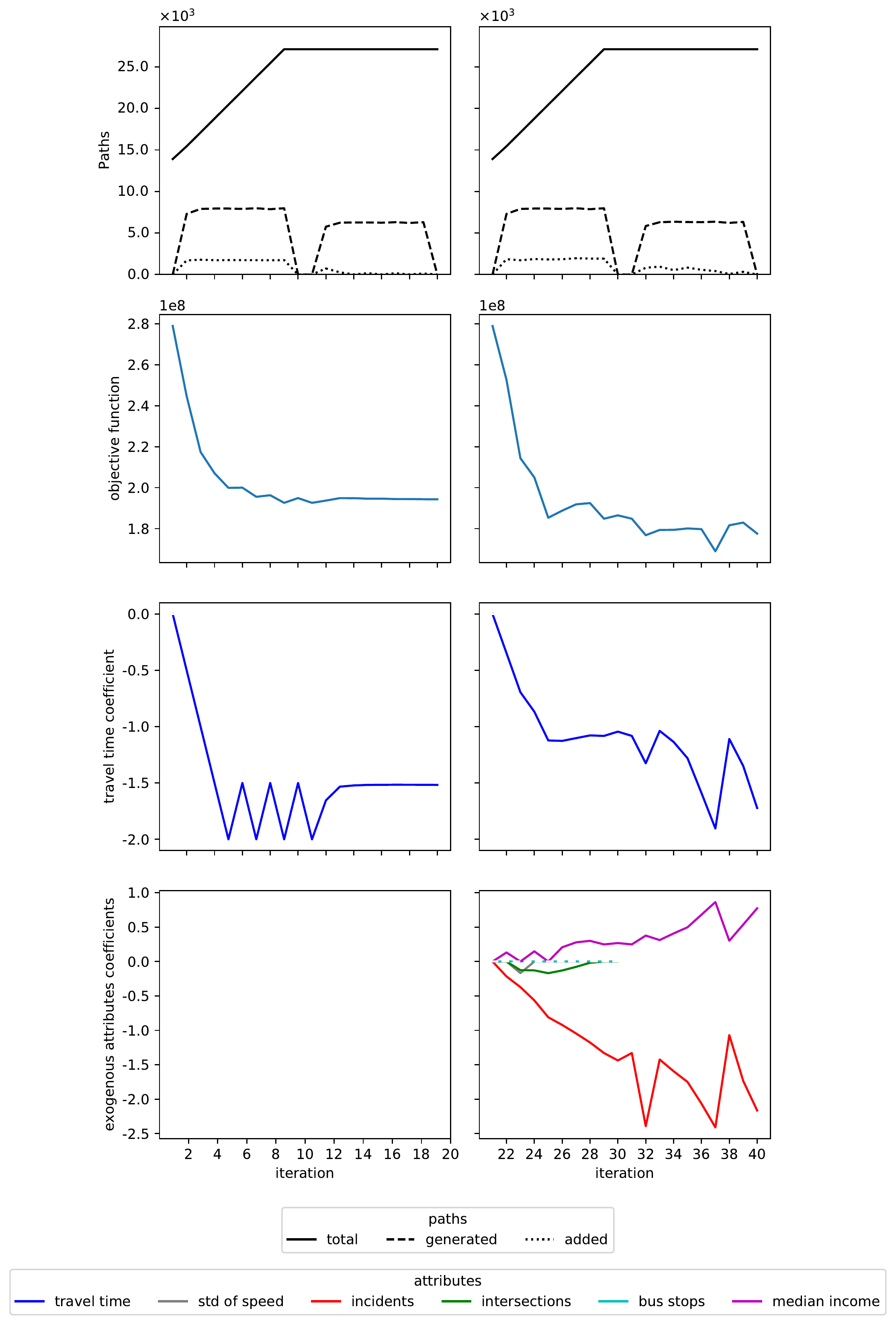}
		\caption{2019}
		\label{subfig:convergence-fresno-2019}
	\end{subfigure}
	\begin{subfigure}[t]{0.48\columnwidth}
		\centering
		\includegraphics[width=\columnwidth, trim= {2.3cm 3.1cm 2.3cm 0cm},clip]{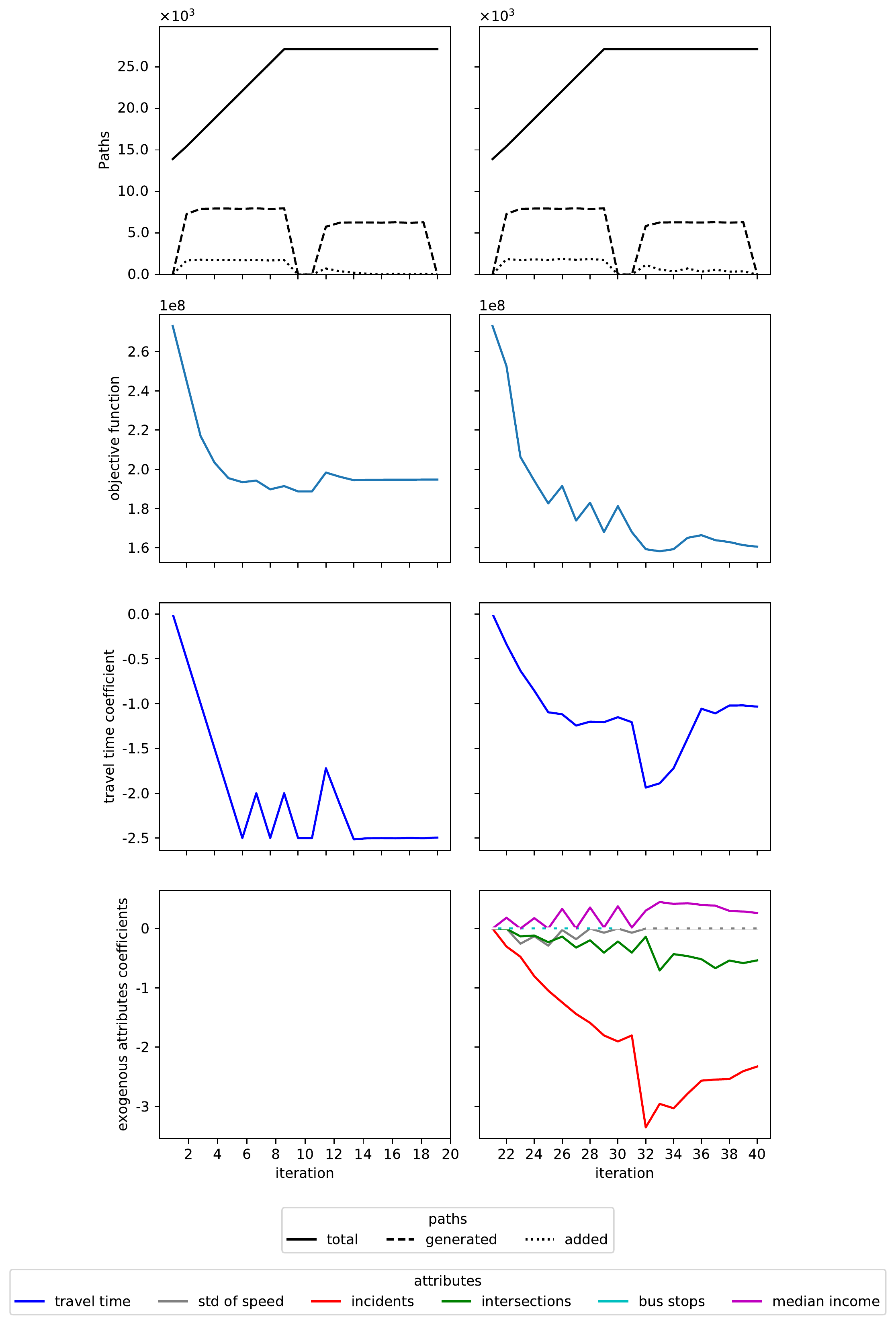}
		\caption{2020}
		\label{subfig:convergence-fresno-2020}
		
	\end{subfigure}
	
	\vspace{0.2cm}
	
	\begin{subfigure}[b]{0.7\columnwidth}
		\centering
		\vskip 0pt
		\includegraphics[width=\columnwidth, trim= {0cm 0cm 0cm 30cm},clip]{figures/fresno/convergence/convergence_traveltime_and_full_models_Fresno_2019.pdf}
	\end{subfigure}
	
	\caption{Comparison of convergence in estimation of baseline and full models}
	
	\label{fig:convergence-baseline-and-full-models-fresno}
\end{figure}

\begin{figure}[H]
	\centering
	
	\begin{subfigure}[t]{0.48\columnwidth}
		\centering
		\includegraphics[width=\columnwidth, trim= {2cm 3.1cm 3cm 0cm},clip]{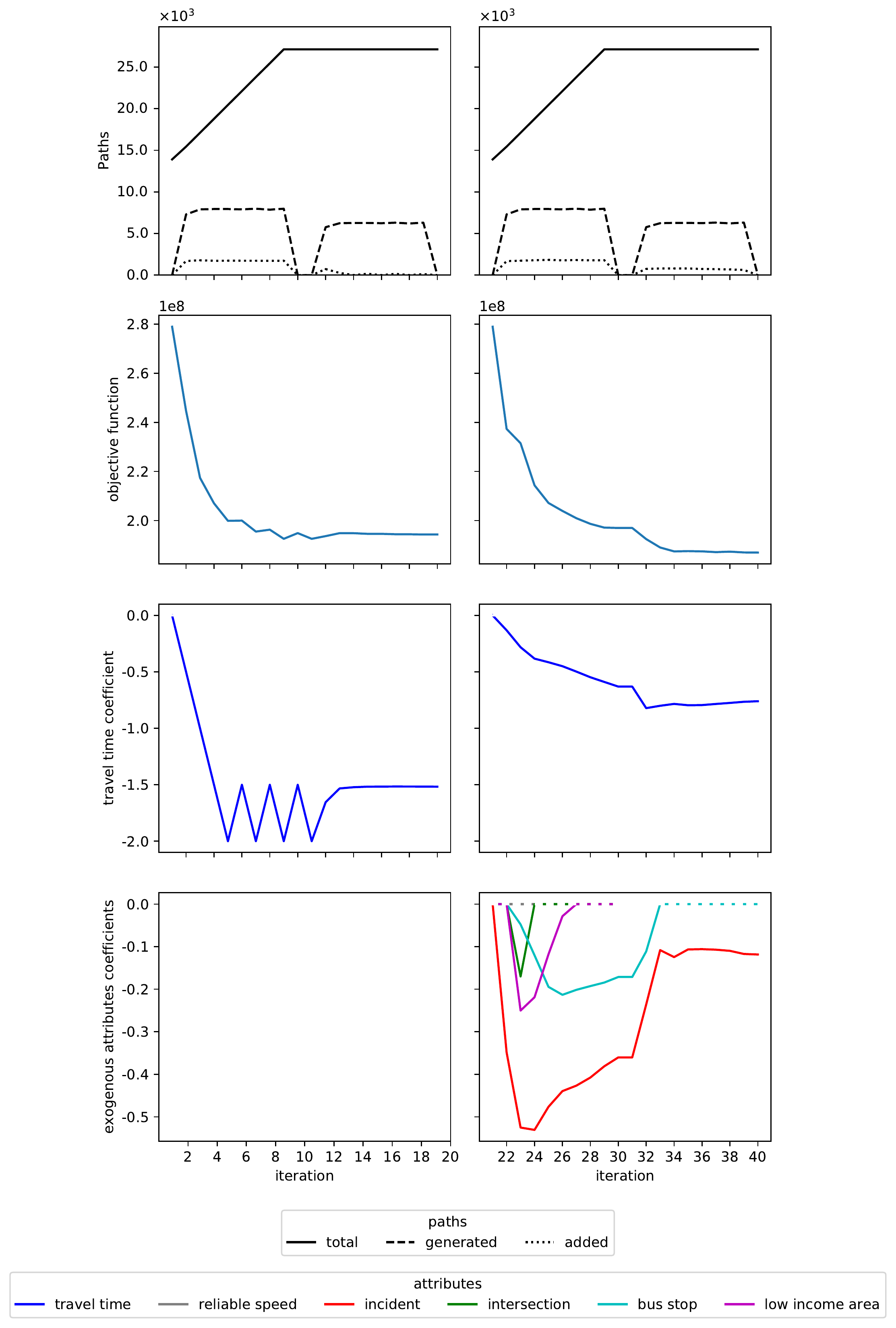}
		\caption{2019}
		\label{subfig:convergence-baseline-and-feature-engineering-models-fresno-2019}
	\end{subfigure}
	\begin{subfigure}[t]{0.48\columnwidth}
		\centering
		\includegraphics[width=\columnwidth, trim= {2cm 3.1cm 3cm 0cm},clip]{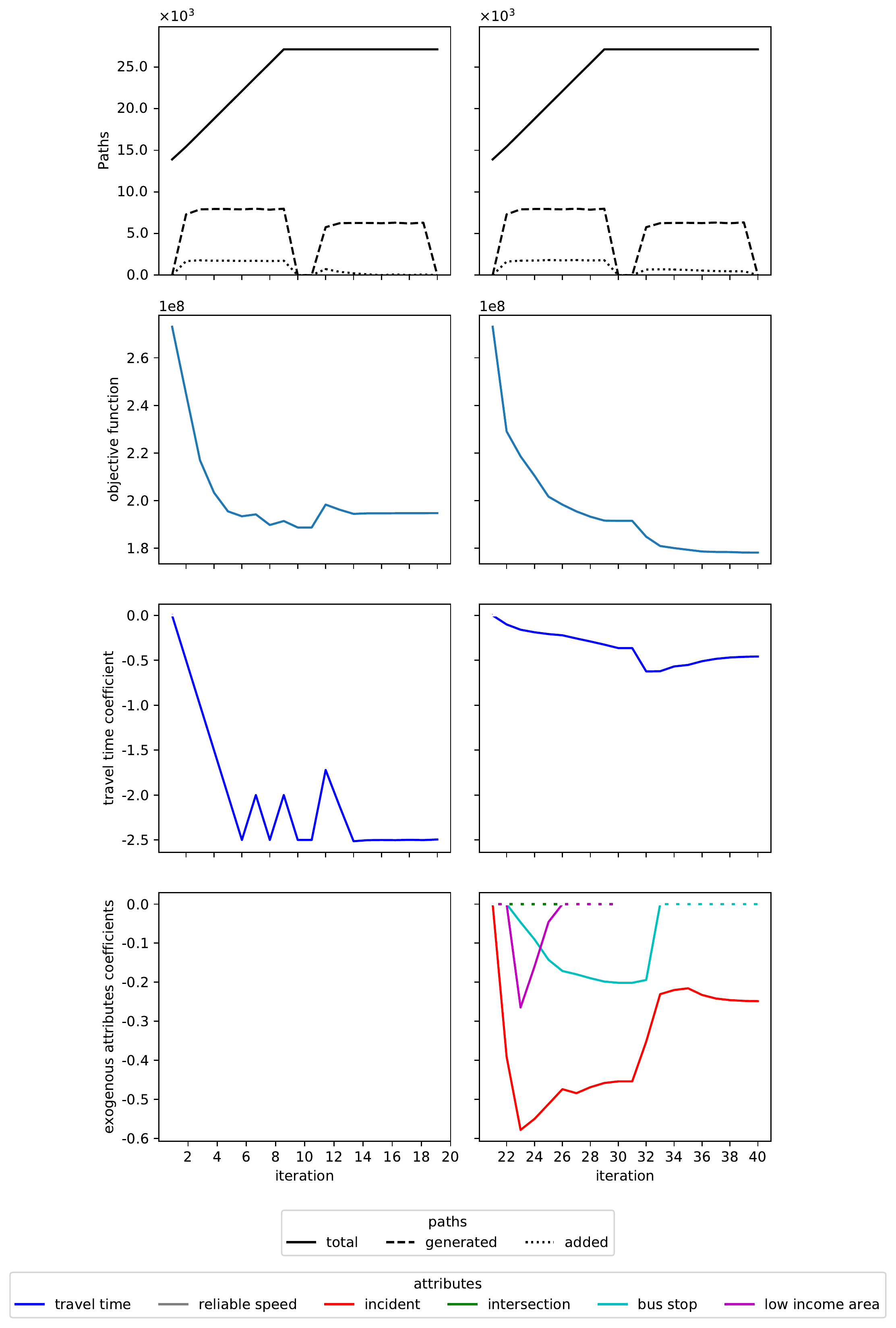}
		\caption{2020}
		\label{subfig:convergence-baseline-and-feature-engineering-models-fresno-2020}
		
	\end{subfigure}
	
	\vspace{0.2cm}
	
	\begin{subfigure}[b]{0.7\columnwidth}
		\centering
		\vskip 0pt
		\includegraphics[width=\columnwidth, trim= {0cm 0cm 0cm 30cm},clip]{figures/fresno/convergence/convergence_traveltime_and_feature_engineering_models_Fresno_2019.pdf}
	\end{subfigure}
	
	\caption{Comparison of convergence in estimation of baseline and binarized models}
	
	\label{fig:convergence-baseline-and-feature-engineering-models-fresno}
\end{figure}